%% file: mscomplexPln.tex
\documentclass[twocolumn]{article}

\usepackage{amssymb}
\usepackage{psfrag}
\usepackage{subfigure}
\usepackage{url}
\usepackage{amsmath}
\usepackage{amsthm}
\usepackage{fancyhdr}
\usepackage{latexsym}
\usepackage{graphicx}
\usepackage{mathrsfs}
\usepackage{pst-all}
\usepackage{enumerate}
\usepackage{color}
\usepackage{colortbl}
\usepackage{fancybox}
\usepackage{appendix}
\usepackage{concrete}
\usepackage[bookmarks,backref=true,linkcolor=black]{hyperref} %,colorlinks
\hypersetup{
  pdfauthor = {},
  pdftitle = {},
  pdfsubject = {},
  pdfkeywords = {},
  colorlinks=true,
  linkcolor= black,
  citecolor= black,
  pageanchor=true,
  urlcolor = black,
  plainpages = false,
  linktocpage
}
\newcommand{\ignore}[1]{}	% argument is ignored
\newcommand{\myPara}[1]{\subsection{#1}}
\newcommand{\B}{\mathbb{B}}

\newcommand{\D}{\mathbb{D}}
\newcommand{\Ds}{\mathbb{D^\ast}}
\newcommand{\I}{\mathbb{I}}
\newcommand{\J}{\mathbb{J}}

\newcommand{\R}{\mathbb{R}}

\newcommand{\Z}{\mathbb{Z}}

\newcommand{\bx}{\mathbf{x}}

\newcommand{\mcB}{\mathcal B}
\newcommand{\gradient}[1]{\nabla #1}
\newcommand{\ol}[1]{\overline{#1}}
\newcommand{\qh}{\operatorname{qh}}
\newcommand{\qv}{\operatorname{qv}}
\newcommand{\ip}[2]{\langle#1,#2\rangle}
\newcommand{\norm}[1]{||\, #1 \,||}

\newcommand{\dt}[1]{\textit{#1}}

\newcommand{\FF}{\mathbb{F}}
\newcommand{\ZZ}{\mathbb{Z}}
\newcommand{\eps}{\varepsilon}
\newcommand{\wt}[1]{\widetilde{#1}}

\newcommand{\intbox}{{\setlength{\unitlength}{.33mm}\,\framebox(4,7){}\,}}

\newtheorem{rem}{Remark}[section]
\newenvironment{remark}{\begin{rem}\em}{\hspace*{\fill}$\Box$\bigskip\noindent\end{rem}}

\newenvironment{main-theorem*}[1]{{\textbf{#1}}}{}
\newtheorem{lem}[rem]{Lemma}
\newtheorem{proposition}{Proposition}[section]
\newtheorem{cor}[rem]{Corollary}
\newtheorem{dfn}[rem]{Definition}
\newtheorem*{fundIneq}{Fundamental Inequality \cite[Theorem 4.4.1]{jw-dedsa-91}}
\newtheorem*{Euler}{Error in Euler's method \cite[Theorem 4.5.2]{jw-dedsa-91}}

\begin{document}

%\begin{frontmatter}



\title{\textbf{Certified Computation of planar Morse-Smale Complexes}}


\author{A. Chattopadhyay \thanks{ School of Computing, University of Leeds, 
  Leeds, UK. {\tt A.Chattopadhyay@leeds.ac.uk}
                      }
\and
G. Vegter\thanks{Johann Bernoulli Institute of Mathematics and
          Computer Science, University of Groningen, The Netherlands. {\tt
            G.Vegter@rug.nl}}
\and
C. K. Yap\thanks{Courant Institute of Mathematical Sciences\, New York University, New York, USA. {\tt
            yap@cs.nyu.edu}}
}

%%%%%%%%%%%%%%%%%%%%%%%%%%%%%%%%%%%%%%%%%%%%%%%%%%%%%%%%%%%%%%%%%%%%%%%%%%%%
%%%%%%%%%%%%%%%%%%%%%%%%%%%%%%%%%%%%%%%%%%%%%%%%%%%%%%%%%%%%%%%%%%%%%%%%%%%%
\maketitle

\begin{abstract}
The Morse-Smale  complex  is an  important tool  for
global  topological  analysis  in  various problems  of  computational
geometry and topology. Algorithms  for Morse-Smale complexes have been
presented        in       case        of        piecewise       linear
manifolds~\cite{EdelsHarZomo03}.  However, previous  research  in this
field is  incomplete in the case  of smooth functions. %  In the current
% paper  we use  interval  arithmetic to  compute a topologically  correct
% approximation of the Morse-Smale  complexes  of  smooth   functions
% on bounded planar domain.  Our
% algorithm can also compute geometrically close Morse-Smale
% complexes.
In the current paper we address the following question: Given an
arbitrarily complex Morse-Smale system on a planar
domain, is it possible to compute its certified (topologically
correct) Morse-Smale complex? Towards this, we develop an algorithm using
interval arithmetic to compute certified critical
points and separatrices forming  the Morse-Smale
complexes  of  smooth   functions on bounded planar domain.  Our
algorithm can also compute geometrically close Morse-Smale complexes.
\end{abstract}

\paragraph*{keyword}

Morse-Smale Complex, Certified Computation, Interval Arithmetic.

%\end{frontmatter}

% \linenumbers

%%%%%%%%%%%%%%%%%%%%%%%%%%%%%%%%%%%%%%%%%%%%%%%%%%%%%%%%%%%%%%%%%%%%%%%%%%%%
%%%%%%%%%%%%%%%%%%%%%%%%%%%%%%%%%%%%%%%%%%%%%%%%%%%%%%%%%%%%%%%%%%%%%%%%%%%%
\input{Tex/1-introduction.tex}
\input{Tex/2-preliminaries.tex}
\input{Tex/3-equations2D.tex}
\input{Tex/4-refinedBoxes.tex}
\input{Tex/5-funnels.tex}
\input{Tex/6-experiments.tex}
\input{Tex/7-conclusion.tex}


\bibliographystyle{plain}
\bibliography{bib/reference}

\begin{appendices}
\input{Tex/appendixMath.tex}
%\input{Tex/appendix_constants.tex}
\end{appendices}

\end{document}

%% file: Tex/1-introduction.tex
%%%%%%%%%%%%%%%%%%%%%%%%%%%%%%%%%%%%%%%%%%%%%%%%%%%%%%%%%%%%%%%%%%%%
%
% Introduction
%
%%%%%%%%%%%%%%%%%%%%%%%%%%%%%%%%%%%%%%%%%%%%%%%%%%%%%%%%%%%%%%%%%%%%
\section{Introduction}
\ignore{% 
\begin{figure}[h!]
  % \begin{flushleft}
  \begin{center}
    \psfrag{a}{\tiny{:~saddle}}  
    \psfrag{b}{\tiny{:~maximum}} 
    \psfrag{c}{\tiny{:~minimum}}
    \excludedFig{\includegraphics[width=0.35\textwidth]{postscript/symbols.eps}}
  \end{center}
  % \end{flushleft}
  \begin{center}
    \excludedFig{\includegraphics[width=0.28\textwidth]{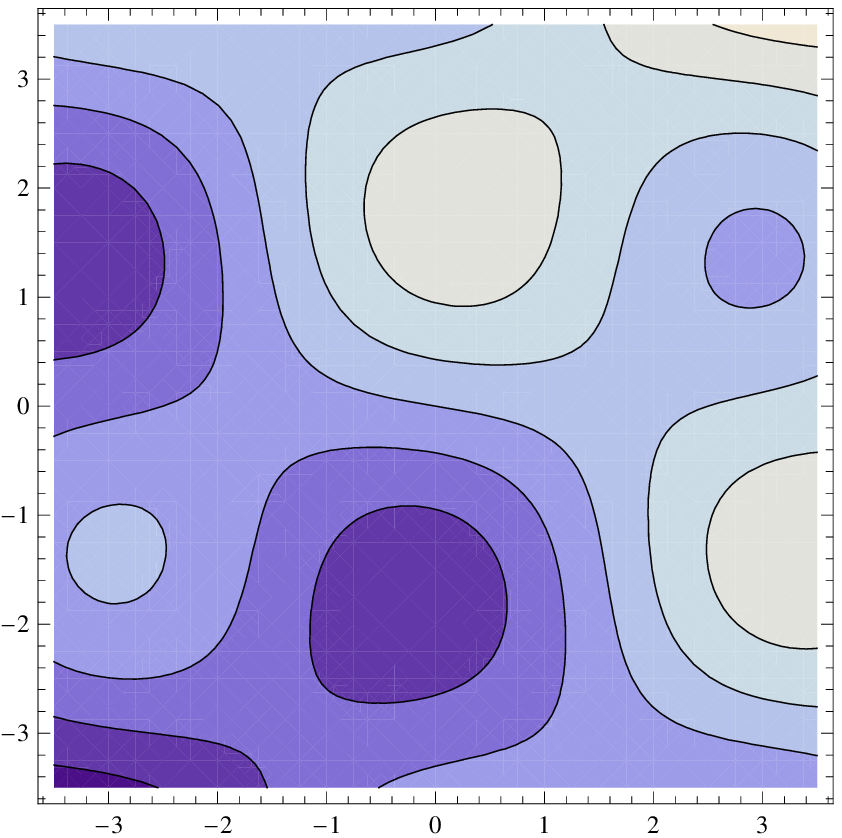}~~
      \includegraphics[width=0.28\textwidth]{postscript/test2eps0pt5angpiby10.eps}}~~
    % \vspace*{-2ex}
   \caption{{MS-Complex of the function $\cos x\;\sin y+0.2\,(x+y)$}: (1) Contour plot, (2) MS-Complex}
    \label{fig:mscomplex-1}
  \end{center}
\end{figure}
}%ignore
Geometrical shapes occurring
in the real world are often extremely complex.  To analyze them,
one associates a sufficiently smooth scalar field
with the shape, e.g., a density function or a function interpolating gray values.
Using this function, topological and geometrical information about the shape
may be extracted, e.g., by computing its \textit{Morse-Smale complex}.
The cells of this complex are maximal connected sets consisting of orthogonal trajectories of the contour
lines---curves of steepest ascent---with the same critical point of the function as origin and the
same critical point as destination.
The leftmost plots in Figures \ref{fig:output1}\subref{fig:squares} and 
\ref{fig:output1}\subref{fig:circular} 
illustrate the level sets of such a density function $h$, and 
the rightmost pictures the Morse-Smale complex
of $h$ as computed by the algorithm in this paper.
This complex reveals the global topology of the shape. 
Recently, the Morse-Smale complex has been successfully applied
in different areas like molecular shape analysis, image analysis, data
and vector field simplification, visualization and detection of voids and clusters in galaxy
distributions~\cite{CazalsChazals03, GyuBreHamPas08}.
\begin{figure}[h!]
  \vspace*{-1ex}
  \centering
  \subfigure[Contour plot (left) and Morse-Smale complex (right) of $h(x,y)=\cos x\;\sin y+0.2\,(x+y)$
      	inside box {$[-3.5,3.5]\times[-3.5,3.5]$}. CPU-time: 11 seconds.]{
          \def\width{0.21\textwidth}
     \includegraphics[width=\width]{postscript/contour-test2.eps}~~
     \includegraphics[width=\width]{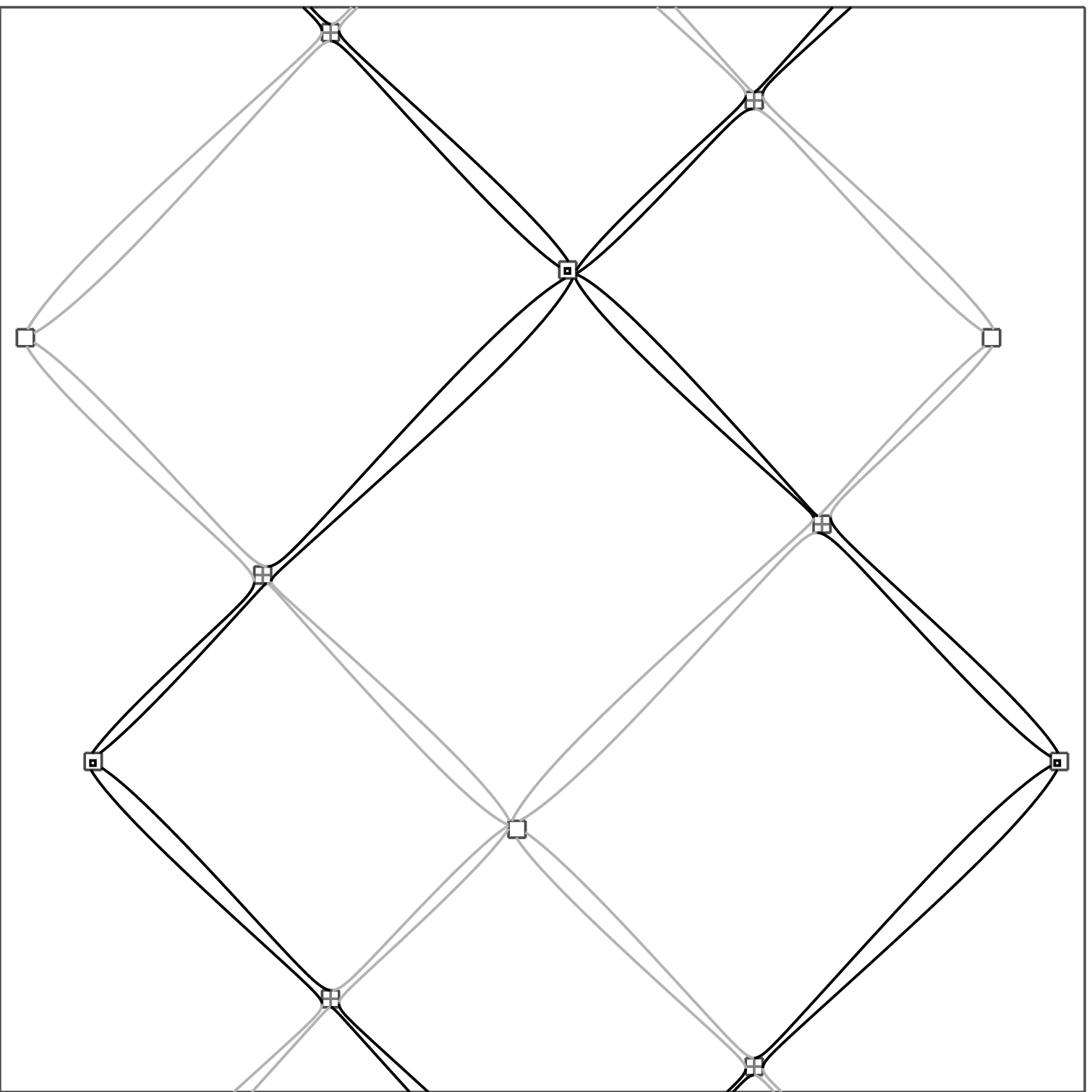}
    \label{fig:squares}
  }~~~\\
  \subfigure[Contour plot (left) and Morse-Smale Complex (right) of 
  $h(x,y) = 10x -\frac{13}{2}(x^2+y^2)+\frac{1}{3}\,(x^2+y^2)^2$
  inside box {$[-5,5]\times[-5,5]$}. CPU-time: 0.5 seconds.]{
    \def\width{0.21\textwidth}
    \includegraphics[width=\width]{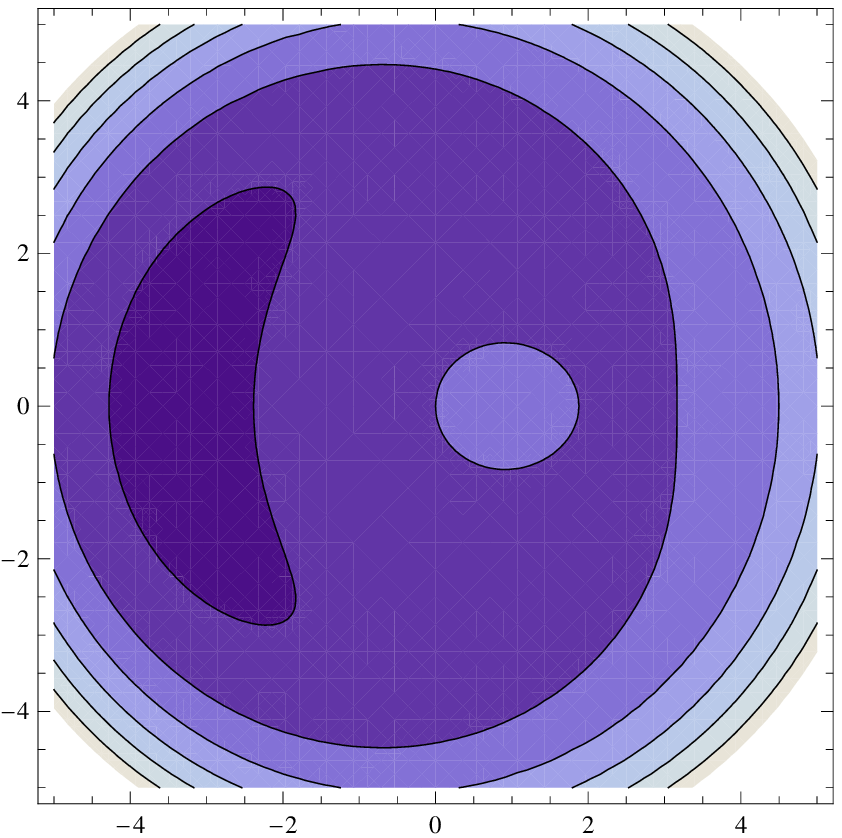}~~
    \includegraphics[width=\width]{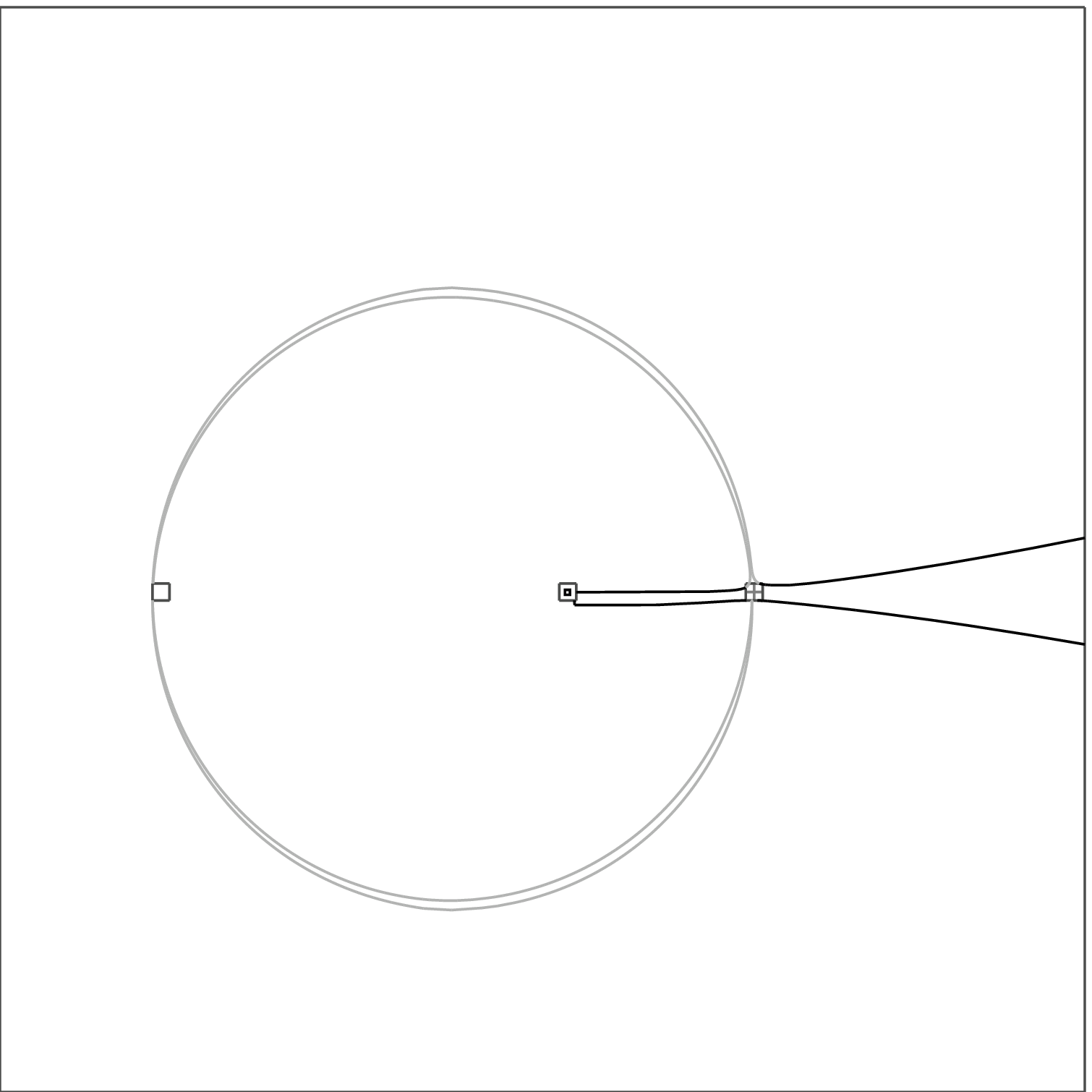}
    \label{fig:circular}
  }
  
  \vspace*{-2ex}
  \caption{Contour plots of Morse-Smale functions, and their Morse-Smale complexes.}
  \label{fig:output1}
  \vspace*{-2ex}
\end{figure}

% \vspace*{-3ex}
\myPara{Problem statement} 
%%%
A Morse function $h: \mathbb{R}^2 \rightarrow \mathbb{R}$
is a real-valued function with \textit{non-degenerate critical points} (i.e.,
critical points with non-singular Hessian matrix).
Non-degenerate critical points are \textit{isolated}, and are either maxima, or minima, or saddle points.
They correspond to singular points of the \textit{gradient vector field} 
$\nabla h$ of $h$, of type sink, source or saddle, respectively.
Regular integral curves of the gradient vector field $\nabla h$ are orthogonal trajectories of the
regular level curves of $h$.
We are interested in the configuration of integral curves of the
gradient vector field. 
An \textit{unstable (stable) separatrix} of a saddle point is the set of all
regular points whose forward (backward) integral curve emerges from the
saddle point. Section~\ref{sec:prelim} contains a more precise definition.
A non-degenerate saddle point has two stable and two unstable separatrices.
A \textit{Morse-Smale function} is a Morse function whose stable and
unstable separatrices are disjoint.
In particular, the unstable separatrices flow into a sink, and separate the unstable regions of two
sources. Similarly, the stable separatrices emerge from a source, and separate the stable regions of
two sinks.
The corresponding gradient vector field is called a \textit{Morse-Smale system} (MS-system).
The \textit{Morse-Smale complex} (MS-complex for short) is a complex consisting
of all singularities, separatrices and the open regions forming their complement, of the
\textit{MS-system}. In other words, a cell of the MS-complex is a maximal connected subset of the
domain of $h$ consisting of points whose integral curves have the same origin and destination.
See also~\cite{EdelsHarNatPas03, EdelsHarZomo03, pm-gtdsi-82} and Section~\ref{sec:prelim}.
The MS-complex describes the global  structure of a Morse-Smale function.

Existing algorithms for MS-complexes~\cite{EdelsHarNatPas03, EdelsHarZomo03} compute the
complex of a piecewise linear function on a piecewise linear manifold,
or, in other words, of a discrete gradient-like vector field.
When $h$ is an analytic function, we cannot use these algorithms
without first creating a piecewise linear approximation $\wt{h}$.
However, the MS-complex of $\wt{h}$ is not guaranteed to be combinatorially equivalent to
the MS-complex of the smooth vector field. The topological correctness depends on how close
the approximation $\wt{h}$ is to $h$.  
%
% DEFINITION:
%
% Here topological correctness means the computation of a complex
% (a cell-decomposition of the
% domain), which is \textit{topologically equivalent} 
% (homeomorphic, or even isotopic, under a map that %this (typo?)
% is also a cell-equivalence) and
% \textit{$\eps$-close} to the 'true' Morse-Smale complex
% (for some a priori specified $\eps>0$).
%
%
% Chee proposes this ALTERNATIVE:
	Here ``topological correctness'' of the computed MS-complex
	$\widetilde{M}$
	means that there is a homeomorphism $f$ of the domain that
	induces a homeomorphism of each cell $\widetilde{c}\in\widetilde{M}$ 
	to a cell $c\in M$ where $M$ is the true MS-complex, and, moreover,
	this induced map $\widetilde{c}\mapsto c$ is an isomorphism
	of $\widetilde{M}$ and $M$.   An isomorphism of
	two MS-complexes preserves the types of cells and 
	their incidence relations.  We can also require $f$ to
	be an \textit{$\eps$-homeomorphism} for some specified $\eps>0$,
        i.e., the distance of a point and its $f$-image does not exceed $\eps$.
As far as we know this problem has never been rigorously studied.  
Therefore, the main problem of this paper is
{\em to compute a piecewise-linear complex that is 
%isotopic
$\eps$-homeomorphic
to the MS-complex of a smooth Morse-Smale function $h$}.
In short, we seek an exact computation in
the sense of the Exact Geometric Computation (EGC) paradigm~\cite{LiYap04}.
Note that it is unclear whether many fundamental problems from
analysis are exactly computable in the EGC sense.
In particular, the current state-of-the-art in EGC does not (yet) provide a good approach for coping
with degenerate situations, and, in fact, this paradigm needs to be extended to
incorporate degeneracies. Therefore, we have to \textit{assume} that the gradients we start out with
are Morse-Smale systems. However, \textit{generic} gradients are Morse-Smale
systems~\cite{pm-gtdsi-82}, so the presence of degenerate singularities and of saddle-saddle
connections is exceptional.
Note that in restricted contexts, like the class of polynomial functions, absence of degenerate
critical points (the first, and local, Morse-Smale condition) can be detected. However, even
(most) polynomial gradient systems cannot be integrated explicitly, so absence of saddle-saddle connections
(the second, and global, Morse-Smale condition) cannot be detected with current approaches. 
Detecting such connections even in a restricted context remains a challenging open problem.

\myPara{Our contribution} We present an algorithm for computing such a
certified approximation of the MS-complex of a given smooth
Morse-Smale function on the plane, %(Figure~\ref{fig:mscomplex-1}). 
as illustrated in Figures \ref{fig:output1}\subref{fig:squares}-\subref{fig:circular}.
In particular, the algorithm produces:

\begin{itemize}
\item (arbitrarily small) isolated  \textit{certified boxes} each containing a unique
  saddle, source or sink;
%  Each box has the local topological property of a saddle, source or sink;

\item  \textit{certified initial and terminal intervals} (on the boundary of saddleboxes),
each of which is guaranteed to contain a unique point corresponding to a
stable or unstable separatrix;

\item disjoint \textit{certified funnels (strips)} around each separatrix, each
  of which contains exactly one separatrix and can be as close to
  the separatrix as desired.
\end{itemize}

\medskip\noindent
\textbf{Note.} The current version is an extensive elaboration of our previous
paper~\cite{cvy12} by incorporating all the theoretical results necessary to
establish our method of certified Morse-Smale complex computation. The aim
in \cite{cvy12} was more on providing an water-tight algorithm; however,
the scope of showing all the theoretical details was
limited. We complete that analysis part in the current extensive version. 
In Section~\ref{sec:eqSolving}, under certified critical-box computation, we provide the details of the relevant lemmas
which were missing in~\cite{cvy12}. In Section~\ref{sec:singularities}, we establish rigorous
theoretical foundations for refining the saddle-, source- and sink
boxes that are used in computing the initial and terminating intervals of the stable and
unstable separatrices. All the theoretical results in this section are
new additions to the current version.
The final method section (Section~\ref{sec:sepStrips}) for the computation of disjoint
certified funnels (strips) is now restructured into three subsections,
each completes the relevant theoretical and algorithmic analysis.

%\noindent
\myPara{Overview} Section~\ref{sec:prelim} starts with a brief review of Morse-Smale systems, their
singular points and their invariant manifolds. We also recall the basics of Interval
Arithmetic, the computational context which provides us with the necessary
certified methods.
The construction of the Morse-Smale complex of a gradient system $\gradient{h}$
consists of two main steps: constructing disjoint certified boxes for its singular
points, and constructing disjoint certified strips (funnels) enclosing its separatrices.
Singular points of the gradient system are computed by solving the system of
equations $h_x(x,y) = 0, h_y(x,y) = 0$. This is a special instance of the more
general problem of solving a generic system of two equations $f(x,y) = 0$, $g(x,y) =
0$. Generic means that the Jacobi matrix at any solution is non-singular, 
or, geometrically speaking, that the two curves $f=0$ and $g=0$ intersect transversally.
In our context, this genericity condition reduces to the fact that at singular points
of the gradient $\gradient{h}$ the Hessian matrix is non-singular.
Section~\ref{sec:eqSolving} presents a method to compute disjoint isolating boxes for
all solutions of such generic systems of two equations in two unknowns. This
method yields disjoint isolating boxes for the singular points of the gradient
system.
In Section~\ref{sec:singularities} these boxes are refined further. Saddle-boxes are
augmented with four disjoint intervals in their boundary, one for each intersection
of the boundary with the stable and unstable separatrices of the enclosed saddle
point. We also show that these intervals can be made arbitrarily small, which is
crucial in the second stage of the algorithm.
Sink- and source-boxes are refined by computing boxes---not necessarily
axis-aligned---around the sink or source on
the boundaries of which the gradient system is transversal (pointing into the sink-box
and out of the source-box). This implies that all integral curves reaching (emerging
from) such a refined sink-box (source-box) lie inside this box beyond (before)
the point of intersection. 

Section~\ref{sec:sepStrips} describes the second stage of the algorithm, in which
isolating strips (funnels) for the stable and unstable separatrices are constructed.
The boundary curves of funnels enclosing an unstable separatrix are polylines with
initial point on a saddle box and terminal point on a sink box. The gradient vector
field is transversally pointing inward at each point of these polylines. 
The initial points of the polylines are connected by the unstable interval through
which the separatrix leaves its saddle box, and, hence, enters the funnel.
The terminal points of these polylines 
lie on the boundary of the same sink-box. See also Figure~\ref{fig:3funnels}. Given this direction of the gradient
system on the boundary of the funnel, the unstable separatrix enters the sink-box and
tends to the enclosed sink, which is its $\omega$-limit. 
Although the width of the funnel may grow exponentially in the distance from the
saddle-box, this growth is controlled. We exploit the computable (although very
conservative) upper bound on this growth rate to obtain funnels that isolate
separatrices from each other, and, hence, form a good approximation of the
Morse-Smale complex together with the source- and sink-boxes.
These upper bounds are also used to prove that the algorithm, which may need several
subdivision steps, terminates.

We have implemented this algorithm, using Interval Arithmetic. 
Section~\ref{sec:implementation} presents sample output of our algorithm.
%Appendix~
\ref{sec:math} contains guaranteed error bounds for the Euler method for
solving ordinary differential equations, and \ref{sec:smallSepIntervals}
sketches a method for narrowing the separatrix intervals in the boundaries of the
saddle boxes.

% {\sc Remark. }
%   In this paper, the domain $\D$ of $h$ is a finite union of
%   axis-aligned dyadic boxes.
%   Furthermore, we (have to) assume that all stable and unstable separatrices of
%   the saddle points are transversal to the boundary.
%   Computationally this means that any sufficiently close approximation of these
%   separatrices is transversal to the boundary as well.
% \end{remark}

%%%%%%%%%%%%%%%%%%%%%%%%%%%%%%%%%%%%%%%%%%%%%%%%%%
%\vspace*{-2ex}
\myPara{Related Work}
Milnor~\cite{Milnor68} provides a basic set-up for Morse theory.
The survey paper~\cite{Biasotti08}, focusing on
geometrical-topological  properties of real functions, gives an
excellent overview of recent works on MS-complexes.
Originally, Morse theory was developed for smooth functions on smooth manifolds.
Banchoff~\cite{Banchoff70}  introduced the equivalent definition of
critical points on polyhedral surfaces.
Many of the recent developments on MS-complexes are based on this definition.
A completely different discrete version of Morse theory is provided by
Forman~\cite{Forman98}.

\medskip\noindent
\textbf{Different methods for computation.}
In the literature there are two different method for computing the
Morse-Smale complexes: (a) boundary based approaches and (b) region
based approaches. Boundary based methods compute boundaries of the
cells of the MS-complex, i.e., the integral curves connecting
a saddle to a source, or a saddle to a sink~\cite{Takahashi95, Bajaj98,EdelsHarZomo03}.
On the other hand, watershed algorithms for image segmentation are
considered as region based approaches~\cite{Meyer94}.  
Edelsbrunner et.al~\cite{EdelsHarZomo03} computes the Morse-Smale complex of
piecewise linear manifolds using a paradigm called
\textit{Simulation of Differentiability}. In higher dimensions they
give an algorithm for computing Morse Smale complexes of piecewise linear
3-manifolds~\cite{EdelsHarNatPas03}.

Morse-Smale complexes have also been applied
in shape analysis and data simplification. Computing MS-complexes is
strongly related to vector field visualization~\cite{HelmanHess91}. In
a similar  context, designing vector fields on surfaces has been
studied for many graphics applications~\cite{ZhangTurk06}.
Cazals et.al.~\cite{CazalsChazals03} applied discrete Morse theory to
molecular shape analysis.

This paper contributes to the
emerging area of Exact Numerical Algorithms
for geometric problems \cite{yap:praise:09}. 
Recent algorithms of this genre
(e.g., \cite{pv-imis-07,lin-yap:cxy:11})
are numerical subdivision algorithms
based on interval function evaluation and sign evaluation.

%% file: Tex/2-preliminaries.tex
% \newpage
%\appendix
%\section{APPENDIX A: PRELIMINARIES}
\section{Preliminaries}
\label{sec:prelim}
%\myPara{Morse functions}
%\paragraph*{Morse functions}
In this section we briefly review the necessary mathematical
background on Morse functions, Morse-Smale systems, their
singular points and their invariant manifolds. We also recall the
basics of our computational model and Interval
Arithmetic which are necessary for our
certified computation algorithm.

\subsection{Mathematical Background}
\label{subsec:mathbgr}
\medskip\noindent\textbf{Morse functions.}~
A function    $h:\mathcal{D} \subset   \mathbb{R}^2 \rightarrow
\mathbb{R}$ is  called a \textit{Morse  function} if all  its critical
points     are    non-degenerate.      The     \textit{Morse    lemma}
~\cite{Milnor68} states  that  near a  non-degenerate  critical
point  $a$ it is  possible to  choose local  co-ordinates $x,y$ in
which $h$ is expressed  as  $h(x,y)=h(a)\pm x^2\pm  y^2$.  
Existence  of
these local  co-ordinates implies that  non-degenerate critical points
are isolated. 
The number of minus  signs is called the \textit{index} $i_h(a)$ of
$h$ at  $a$.  Thus a  two variable Morse  function has three  types of
non-degenerate critical  points: minima  (index 0), saddles  (index 1)
and maxima (index 2).

%\myPara{Integral curves}
\medskip\noindent\textbf{Integral curves.}
An integral curve $\bx:I\subset\mathbb{R}\rightarrow{\mathcal{D}}$ passing
through a point $p_0$ on $\mathcal{D}$ is a unique maximal curve satisfying:
% \begin{equation}
% %  \label{eq:grad}
%   \left\{
% \begin{array}{cc}
%   \dot{\bx}(t)=&\nabla{h(\bx(t))}\\
%   \bx(0)=&p_0\\
% \end{array}\right.  
% \end{equation}
$ %\begin{equation*}
   \dot{\bx}(t)=\nabla{h(\bx(t))}, ~~~ \bx(0) = p_0,
$ %\end{equation*}
for all $t$ in the \textit{interval} $I$. Integral curves corresponding to the gradient
vector field of a smooth function
$h:\mathcal{D}\rightarrow \mathbb{R}$ have the following properties:

\medskip\noindent
1.~Two integral curves are either disjoint or same.
\\
2.~The integral curves cover all the points of $\mathcal{D}$.
\\
3.~The integral curves of the gradient vector field of $h$ form a
  partition of $\mathcal{D}$.
\\
4.~The integral curve $\bx(t)$ through a critical point $p_0$ of $h$ is
  the constant curve $\bx(t)=p_0$.
\\
5.~The integral curve $\bx(t)$ through a regular point $p$ of $h$ is
  injective, and if $\displaystyle\lim_{t\rightarrow \infty}\,\bx(t)$ or
  $\displaystyle\lim_{t\rightarrow -\infty}\,x(t)$ exists, it is a critical
  point of $h$. This implies integral curves corresponding to gradient
  vector field are never closed curves.
\\
6.~The function $h$ is strictly increasing along the integral curve
  of a regular point of $h$.
\\
7.~Integral curves are perpendicular to the regular level sets of $h$.

%\myPara{Stable and unstable manifolds} 
\medskip\noindent\textbf{Stable and unstable manifolds.} 
Consider the integral curve
$x(t)$ passing    through     a    point     $p$. If the    limit
$\displaystyle\lim_{t\rightarrow   \infty}\,\bx(t)$ exists, it is   called   the
$\omega$-limit  of  $p$ and  is  denoted  by $\omega(p)$.   Similarly,
$\displaystyle\lim_{t\rightarrow   -\infty}\,\bx(t)$   is   called   the
$\alpha$-limit  of  $p$ and  is  denoted  by  $\alpha(p)$ -- again
provided this limit exists.
The  stable manifold  of   a  singular  point   $p$  is  the   set
$W^s(p)=\{q\in \mathcal{D} \mid \omega(q)=p\}$.   Similarly, the  unstable
manifold of  a
singular     point      $p$     is     the      set     $W^u(p)=\{q\in
\mathcal{D} \mid \alpha(q)=p\}$.   Here  we  note  that both  $W^s(p)$  and
$W^u(p)$ contain the singular point $p$ itself~\cite{hirsch74}. 

Now, the stable and unstable manifolds of a saddle point are
1-dimensional manifolds. A stable manifold of a saddle point consists of
two integral curves converging to the saddle point. Each of these
integral curves (not including the saddle point) are called the stable
separatrices of the saddle point. Similarly, an unstable manifold of a saddle point consists of
two integral curves diverging from the saddle point and each of these
integral curves (not including the saddle point) are called the unstable
separatrices of the saddle point.

%\myPara{The Morse-Smale complex}
\medskip\noindent\textbf{The Morse-Smale complex.}
A Morse function on $\mathcal{D}$ is called a Morse-Smale (MS) function if
its stable and unstable separatrices are disjoint. In
particular, a Morse-Smale function on a two-dimensional domain has
no integral curve connecting two saddle points, since in that case a stable
separatrix of one of the saddle points would coincide with an unstable separatrix of the
other saddle point. 
The MS-complex associated with a MS-function
$h$ on $\mathcal{D}$ is the subdivision of $\mathcal{D}$ formed by the
connected components of the intersections $W^s(p)\cap W^u(q)$, where
$p$, $q$ range over all singular points of $h$.
\begin{figure}
  \centering
  \psfrag{a}{\scriptsize{Sink}}
  \psfrag{b}{\scriptsize{Saddle}}
  \psfrag{c}{\scriptsize{Source}}
  \psfrag{d}{\tiny{Stable Separatrix}}
  \psfrag{e}{\tiny{Unstable Separatrix}}
  \includegraphics[scale=0.40]{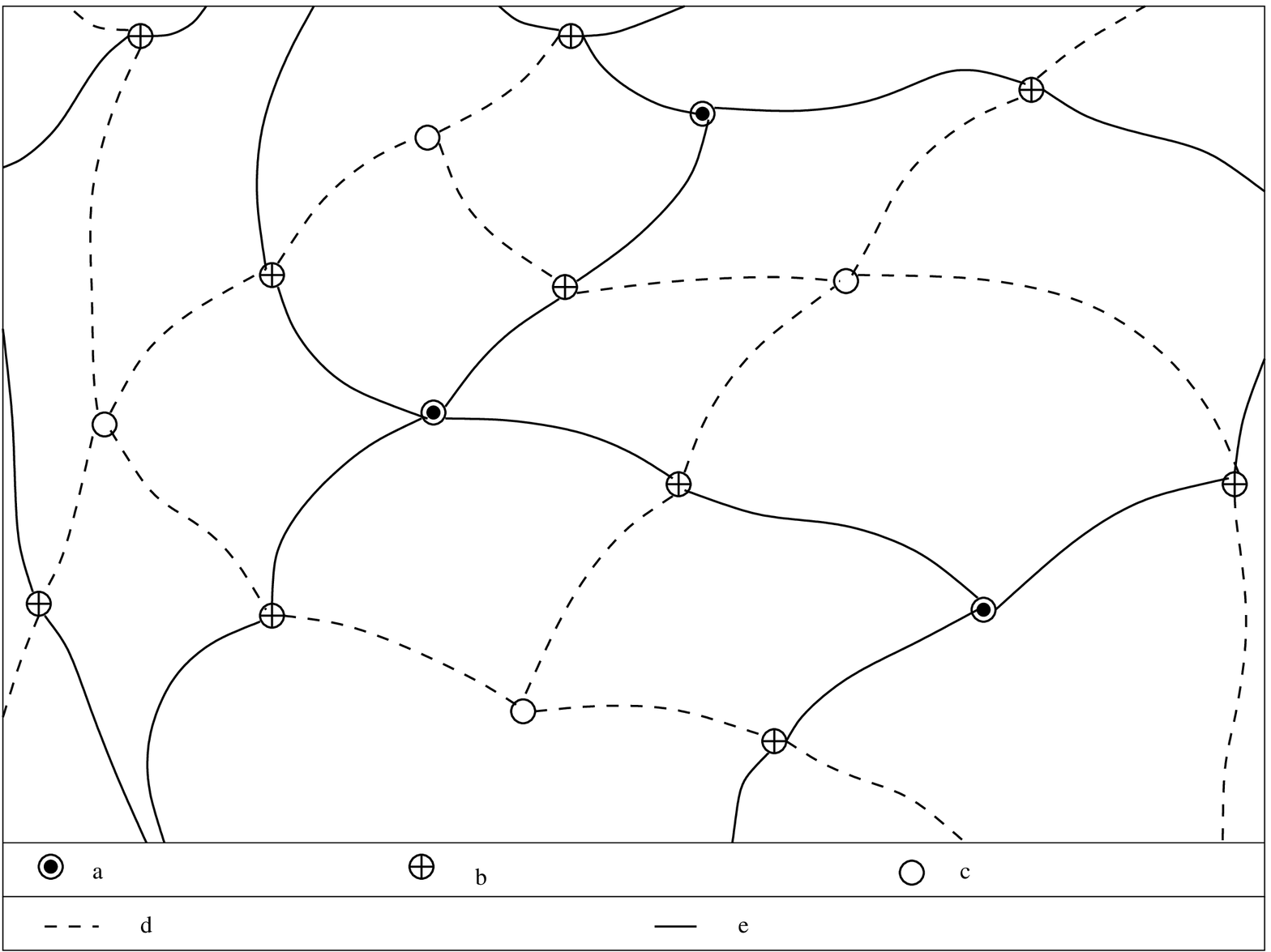}
 \vspace*{-2ex}
  \caption{Morse-Smale complex}
  \label{fig:mscomplex}
\end{figure}
% \noindent
If $\mathcal{D} = \mathbb{R}^2$, then, according to the Quadrangle Lemma~\cite{EdelsHarZomo03}, 
each region of the MS-complex is a quadrangle with vertices of index
$0, 1, 2, 1$, in this order around the region.  

\ignore{%duplicated
\begin{remark}
  \label{rem:transversality}  
  In this paper, the domain $\mathcal{D}$ is a finite union of axis-aligned
  boxes. Furthermore, we (have to) assume that all stable and unstable separatrices of
  the saddle points are transversal to the boundary.
  Computationally this means that any sufficiently close approximation of these
  separatrices is transversal to the boundary as well.
\end{remark}
}%ignore

%\myPara{Stability of equilibrium points}
\medskip\noindent\textbf{Stability of equilibrium points.}
We note that a gradient vector field of a MS-function
$h:\mathcal{D} \rightarrow \mathbb{R}$ can have three kinds of
equilibria or \textit{singular points}, namely, sinks
(corresponding to maxima of $h$), saddles (saddles of $h$) and sources
(corresponding to minima of $h$). These singular points can be
distinguished based on the local behavior of the integral curves around
those points. Locally, a sink has a neighborhood, which is
a stable 2-manifold. Similarly, locally
a source has a neighborhood, which is an unstable 2-manifold. Locally, a
saddle has a stable 1-manifold and an unstable 1-manifold crossing each other
at the saddle point. A sink is called a stable equilibrium
point, where as a source or a saddle is called unstable equilibrium point.
We note that, a source corresponding to a MS-function $h$ is a sink
corresponding to the function $-h$.

%ignore

%%%%%%%%%%%%%%%%%%%%%%%%%%%%%%%%%%%%%%%%%%%%%%%%%%%%%%%%%%%%%%%%%%%%%
\subsection{Computational Model}
\label{subsec:model}
%\vspace*{-2ex}
Our computational model
%\footnote{See Appendix~\ref{app:prelim} for additional details.}
has two simple foundations:
(1) BigFloat packages
and (2) interval arithmetic (IA) \cite{m-p-96}.
Together, these are able to deliver efficient and guaranteed results in 
the implementation of our algorithms.
A BigFloat package is a software implementation of exact ring
($+,-,\times$) operations, division by $2$, and exact comparisons,
over the set $\FF=\{m2^n: m,n\in\ZZ\}$ of \dt{dyadic numbers}.
In practice, we can use IEEE machine arithmetic as a filter for 
BigFloat computation to speed up many computation.
\textit{Range functions} form a basic tool of IA:
given any function $F:\mathbb{R}^m \rightarrow \mathbb{R}^n$,
a \dt{range function} $\intbox{F}$ for $F$ computes for each $m$-~dimensional
interval $I$ (i.e., an $m$-box) an $n$-dimensional interval
$\intbox{F}(I)$, such that $F(I) \subset \intbox{F}(I)$. A range function is
said to be \textit{convergent} if the diameter of the output interval
converges to $0$ when the diameter of the input interval shrinks to $0$.
Convergent range functions exist for the basic operators and
functions, so all range functions are assumed to be convergent.
Moreover, we assume that the sign of functions can be evaluated exactly at
dyadic numbers.  All our boxes are
\textit{dyadic boxes}, meaning that their corners have dyadic coordinates.
% A useful technique here is the interval form of the implicit
% function theorem (see Appendix~\ref{app:prelim}.)

%\ignore{
\medskip\noindent\textbf{Interval implicit function theorem.}
To introduce a useful tool from IA, we recall some notation for
\textit{interval matrices}.
An $n\times n$ interval matrix $\intbox M$ is defined as

$$\intbox{M}=\{M|M_{ij}\in\intbox M_{ij}, i,j\in \{1\ldots n\}\}$$
Also note that we write 
 $$0\notin \det \intbox{M}$$
if there exists no matrix $M\in\intbox{M}$ such that
$\det{M}=0$ where $\det$ represents the determinant of the corresponding matrix.

\noindent
% Consider a system of equations:
% \begin{equation}
%   \label{eq:sys-equns}
%   f(x,y)=0, ~~ g(x,y)=0,
% \end{equation}
% and let $F:\R^2\rightarrow\R^2$ be the function with components $f$
% and $g$.

% \noindent
% Let $I=I_x\times I_y$ be a 2D-interval (box) in
% $\mathbb{R}^2$. 
% Define the \textit{interval Jacobian determinant} $\intbox J_{x,y}F$ as the
% $2\times 2$ interval determinant given by

% \begin{displaymath}
%   \intbox J_{x,y}F(I)=
%   \begin{vmatrix}
%     \intbox\dfrac{\partial f}{\partial x}(I) &  \intbox\dfrac{\partial f}{\partial y}(I)\\[1.6ex] 
%     \intbox\dfrac{\partial g}{\partial x}(I) &  \intbox\dfrac{\partial g}{\partial y}(I)
%   \end{vmatrix}
% \end{displaymath}

% \noindent
% The following theorem gives a necessary condition guaranteeing that
% there is at most one solution of the system (\ref{eq:sys-equns}) in $I$.

% \begin{theorem}
%   \label{thm:intervalIFT}
%   (\textbf{Interval Implicit Function Theorem, Snyder~\cite{s-gmcgc-92,s-iacg-92}})
%   Let $F:\mathbb{R}\times\mathbb{R}\rightarrow\mathbb{R}$ be a
%   $C^{1}$-map. If $I\subset
%   \mathbb{R}\times \mathbb{R}$ is a box for which 
%   \begin{equation*}
%     0 \notin \intbox J_{x,y}F(I),    
%   \end{equation*}
%   then the system~(\ref{eq:sys-equns}) has at most one solution in $I$.  
% \end{theorem}

% \ignore{
If $I=I_x\times I_y$ is a 2D-interval (box) in $\mathbb{R}^2$, 
the \textit{interval Jacobian determinant} $\intbox \frac{\partial(f,g)}{\partial(x,y)}$ is the
$2\times 2$ interval determinant given by
\begin{displaymath}
  \intbox \frac{\partial(f,g)}{\partial(x,y)} (I)=
  \begin{vmatrix}
    \intbox\dfrac{\partial f}{\partial x}(I) &  \intbox\dfrac{\partial f}{\partial y}(I)\\[1.6ex] 
    \intbox\dfrac{\partial g}{\partial x}(I) &  \intbox\dfrac{\partial g}{\partial y}(I)
  \end{vmatrix}
\end{displaymath}
% \vspace*{-2ex}
% \myPara{Interval implicit function theorem}
% An $n\times n$ \dt{interval matrix} $\intbox M$ is represented by an $n\times n$
% array of intervals $\intbox M_{ij}$ $(i,j\in\{1\dd n\}$, and it represents the set
% 	$\{M|M_{ij}\in\intbox M_{ij}, i,j\in \{1\ldots n\}\}$
% %
% of ordinary matrices.  We write $0\notin \det \intbox{M}$
% if there does not exist any matrix $M\in\intbox{M}$ such that $\det{M}=0$.
% Consider a system of equations:
% \begin{equation}
%   \label{eq:systemEq}
%   f(x,y)=0, ~~ g(x,y)=0,
% \end{equation}
% and let $F:\R^2\rightarrow\R^2$ be the function with components $f$
% and $g$.
% If $I=I_x\times I_y$ is a 2D-interval (box) in $\mathbb{R}^2$, 
% the \textit{interval Jacobian determinant} $\intbox J_{x,y}F$ is the
% $2\times 2$ interval determinant given by

% \begin{displaymath}
%   \intbox J_{x,y}F(I)=
%   \begin{vmatrix}
%     \intbox\dfrac{\partial f}{\partial x}(I) &  \intbox\dfrac{\partial f}{\partial y}(I)\\[1.6ex] 
%     \intbox\dfrac{\partial g}{\partial x}(I) &  \intbox\dfrac{\partial g}{\partial y}(I)
%   \end{vmatrix}
% \end{displaymath}

% \noindent
% The following gives a sufficient condition for
% the system (\ref{eq:systemEq}) to have at most one solution in $I$.

\begin{proposition}{(\textbf{Interval Implicit Function Theorem, Snyder~\cite{s-gmcgc-92,s-iacg-92}})}
  \label{thm:intervalIFT}
  % ~\\
  % (\textbf{Interval Implicit Function Theorem, Snyder~\cite{s-gmcgc-92,s-iacg-92}})
  Let $F:\mathbb{R}\times\mathbb{R}\rightarrow\mathbb{R}$ be a
  $C^{1}$-map with components $f$ and $g$. If $I\subset
  \mathbb{R}\times \mathbb{R}$ is a box for which 
   $ 0 \notin \intbox \frac{\partial(f,g)}{\partial(x,y)} (I)$, 
  then the system $f(x,y)=0, ~~ g(x,y)=0$ has at most one solution in $I$.  
\end{proposition}
% }

\begin{remark}
%  \marginpar{Move elsewhere?}
\label{rem:transversality}  
\medskip%\noindent
%{\sc Remark. }
  In this paper, the domain $\D$ of $h$ is a finite union of
  axis-aligned dyadic boxes.
  Furthermore, we (have to) assume that all stable and unstable separatrices of
  the saddle points are transversal to the boundary.
  Computationally this means that any sufficiently close approximation of these
  separatrices is transversal to the boundary as well.
\end{remark}

%% file: Tex/3-equations2D.tex
\section{Isolating boxes for singularities of gradient  fields}
\label{sec:eqSolving}
As a first step towards the construction of the Morse-Smale complex of
$h$ we construct disjoint isolating boxes for the singular points of
$\gradient{h}$. To this end, we first show how to compute isolating
boxes for the solutions of a generic system of two equations in two unknowns, which
are confined to a bounded domain in the plane. This domain is a finite union
of dyadic boxes.
Applying this general method to the case in which the two equations are
defined by the components of the gradient vector field $\gradient{h}$ we
obtain isolating boxes for the singularities of this gradient field.

%%%%%%%%%%%%%%%%%%%%%%%%%%%%%%%%%%%%%%%%%%%%%%%%%%%%%%%%%%%%%%%%%%
%
% Setting: generic system of equations
%
%%%%%%%%%%%%%%%%%%%%%%%%%%%%%%%%%%%%%%%%%%%%%%%%%%%%%%%%%%%%%%%%%%
\subsection{Certified solutions of systems of equations}
\label{sec:genericSystems}
We consider a system of equations 
% \begin{equation}
%   \label{eq:systemEq}
%   \begin{split}
%     f(x,y) &= 0, \\
%     g(x,y) &= 0,
%   \end{split}
% \end{equation}
\begin{equation}
  \label{eq:systemEq}
  f(x,y) = 0, \quad g(x,y) = 0,
\end{equation}
where $f$ and $g$ are $C^1$-functions defined on a bounded axisparallel box
$\D \subset \R^2$ with \textit{dyadic vertices}.
Furthermore, we assume that the system has only
non-degenerate solutions, i.e., the Jacobian determinant is non-zero at
a solution. In other words, 
\begin{equation}
\label{eq:nonzerojacob}
  \left.\frac{\partial(f,g)}{\partial(x,y)}\right|_{(x_0,y_0)} \neq 0
\end{equation}
for $(x_0,y_0)$ satisfying (\ref{eq:systemEq}). Geometrically, this means
that the curves given by $f(x,y) = 0$ and $g(x,y) = 0$ are regular near
a point of intersection $(x_0,y_0)$ and the intersection is transversal.
We will denote these curves by $Z_f$ and $Z_g$, respectively.
Note that this condition is satisfied by Morse-Smale systems, since
in that case $f=h_x$ and $g=h_y$, so the Jacobian determinant is
precisely the Hessian determinant.

Since the domain $\D$ of $f$ and $g$ is compact and we assume (\ref{eq:nonzerojacob}), the system
(\ref{eq:systemEq}) has finitely many solutions in $\D$.
% Our goal is to determine \textit{certified boxes} for this set of
% solutions, i.e., we construct disjoint axisparallel boxes $\B_1, \ldots,
% \B_m$ each containing exactly one solution, and such that each solution
% is contained in one of these boxes.
Our goal is to construct a collection of axis-aligned boxes $\B_1,\ldots,\B_m$ and 
$\B'_1,\ldots,\B'_m$ such that
(i) box $\B_i$ is concentric with and strictly contained in $\B'_i$,
(ii) the boxes $\B'_i$ are disjoint,
(iii) each solution of (\ref{eq:systemEq}) is contained in one of the boxes $\B_i$, and
(iv) each box $\B'_i$ contains exactly one solution (contained inside the enclosed
box $\B_i$).
The box pair $(\B_i,\B'_i)$ is \textit{certified}: $\B_i$ contains a solution,
and $\B_i'$ provides positive clearance to other solutions.
In fact, the sequence of boxes will satisfy the following stronger
conditions; See also Figure~\ref{fig:certifiedBox}.
\begin{figure}[h]
  \centering
\psfrag{ci}{$\B_i$}
  \psfrag{bi}{$\B_i^{'}$}
%  \psfrag{ci}{$\D_i$}
  \psfrag{Zf}{$Z_f$}
  \psfrag{Zg}{$Z_g$}
%   \includegraphics[width=0.3\textwidth]{postscript/certifiedBox_1}
%   \qquad
  \includegraphics[width=0.3\textwidth]{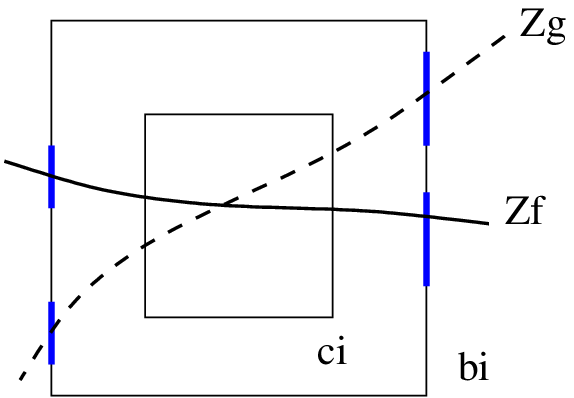}
  \caption{A certified box-pair $(\B_i, \B_i^{'})$, with isolating intervals in its
    boundary for $Z_f\cap\partial\B_i^{'}$ and $Z_g\cap\partial\B_i^{'}$.
  The curves $Z_f$ and $Z_g$ intersect inside $\B_i^{'}$ iff these intervals
  are interleaved.}
  \label{fig:certifiedBox} 
\end{figure}
\begin{enumerate}
% \item 
%   The solution in $\B_i$ is contained in the concentric box $\D_i$
%   with the same center as $\B_i$ and half the size of $\B_i$;
\item
  The curves $Z_f$ and $Z_g$ each intersect the boundary of $\B_i^{'}$
  transversally in two points.
\item 
  There are disjoint intervals $\I^0_i(f)$, $\I^0_i(g)$, $\I^1_i(f)$ and
  $\I^1_i(g)$ in the boundary of $\B_i^{'}$ (in this order), where the first
  and the third interval each contain one point of intersection of $Z_f$ and
  $\partial\B_i^{'}$, and the second and fourth interval each contain one point of
  intersection of $Z_g$ and $\partial \B_i^{'}$.
\item 
  The (interval) Jacobian determinant of $f$ and $g$ does not vanish on $\B_i^{'}$,
  i.e.,
  \begin{equation*}
    0 \not\in \Box \frac{\partial(f,g)}{\partial(x,y)} (\B_i^{'}).
  \end{equation*}
\end{enumerate}

%%%%%%%%%%%%%%%%%%%%%%%%%%%%%%%%%%%%%%%%%%%%%%%%%%%%%%%%%%%%%%%%%%
%
% Transversal intersections
%
%%%%%%%%%%%%%%%%%%%%%%%%%%%%%%%%%%%%%%%%%%%%%%%%%%%%%%%%%%%%%%%%%%
\subsection{Construction of certified box pairs}
\label{sec:const-box-pair}
We first subdivide the domain $\D$ into equal-sized boxes $\I$ (called grid-boxes),
until all boxes satisfy certain conditions to be introduced now.
For a (square, axis-aligned) box $\I$, let $N_{\varrho}(\I)$ be 
the box obtained by multiplying box $\I$ from its center by a factor of $1+\varrho$,
where $\tfrac{1}{2}\leq \varrho \leq 1$. 
We also denote $N_1(\I)$ by $N(\I)$; this is the box formed by the
union of $\I$ and its eight neighbor grid-boxes. 
We shall call $N(\I)$ the surrounding box of $\I$.
% \marginpar{Notation: $\gradient{f}$}
The algorithm subdivides $\D$ until all grid-boxes $\I$ satisfy 
$C_0(\I) \vee \bigl(\,C_1(\I) \wedge C_2(\I)\,\bigr)$, where
the clauses $C_i(\I)$, $i=0,1,2$, are the following predicates:
\begin{align*}
  C_0(\I):  \quad & 0 \not\in \Box f(\I) \vee 0 \not \in \Box g(\I) 
  \\[1.2ex]
  C_1(\I):  \quad & 0 \not\in\Box \frac{\partial(f,g)}{\partial(x,y)}(N(\I))
%  \quad \text{where~} \frac{\partial(f,g)}{\partial(x,y)} \text{~is the Jacobian determinant}
  \\[1.2ex]
C_2(\I):  \quad & C_2(\I,f) \wedge C_2(\I,g),
  \\[1.2ex]
  \qquad \text{ where } & \\
  C_2(\I,f) &= \ip
            {\intbox\frac{\gradient{f}}{\norm{\gradient{f}}}(N(\I))}%
            {\intbox\frac{\gradient{f}}{\norm{\gradient{f}}}(N(\I))}
            \geq \cos\frac{\pi}{30}.
\end{align*}
If $C_0(\I)$ holds, box $\I$ does not contain a solution, so it is discarded.
The second predicate guarantees that $N(\I)$ contains at most one
solution. This is a consequence of Interval Implicit Function Theorem
\cite{sny92,sny-book-92}.
% \marginpar{Geometric interpr!}
Condition $C_2(\I)$ is a \textit{small angle variation condition}, guaranteeing that
the variation of the unit normals of the curves $Z_f \cap N(\I)$ and $Z_g \cap N(\I)$
do not vary too much, so these curves are regular, and even `nearly linear'
(the unit normal of $Z_f$ is the normalized gradient of $f$). Here $\ip{\cdot}{\cdot}$ 
denotes the interval version of the standard inner product on $\R^2$.
The lower bound for the angle variation of $\gradient{f}$ is generated
by the proof of Lemma~\ref{lemma:isolatingIntervals}.
%
% Condition $C_4(\I)$ guarantees that the angle between the normals of these curves
% make a positive angle with \textit{computable} lower bound $\alpha(\I)$, which is
% crucial in the computation of isolating intervals for the intersection of these
% curves with $\partial N(\I)$.

\begin{remark}
Condition $C_1(\I)$ implies that there is a computable
positive lower bound $\alpha(\I)$
on the angle between $\nabla{f}(p)$ and $\nabla{g}(q)$ where $p,q$ range
over the surrounding box of $\I$.  More precisely, to compute $\alpha(\I)$, we
first compute a lower bound $L$ on the quantity
$\frac{\partial(f,g)}{\partial(x,y)} \frac{1}{\|\nabla f\|\cdot\|\nabla g\|}$.
This $L$ may be obtained by an interval evaluation of this quantity at $N(\I)$;
note that $L>0$ iff condition $C_1(\I)$ holds.
We define $\alpha(\I)$ as $\arcsin(L)$.
% \ignore{ %done
%   \marginpar{To do!}
%   Condition $C_4(\I)$ should be dropped. Instead, condition $C_1(\I)$ implies that 
%   there is a positive lower bound on the angle between $\nabla{f}(p)$ and
%   $\nabla{g}(q)$, where $p$ and $q$ range over the box $\I$. 
%   Our computational model implies that a positive lower bound $\alpha(\I)$ for this
%   angle is \textit{computable}.
%   }%ignore
\end{remark}

% Our algorithm will construct disjoint certified boxes surrounding
% a box $\I$.
% Since the surrounding boxes $N(\I)$ and $N(\J)$ of disjoint boxes $\I$ and $\J$ may
% intersect, we consider smaller surrounding boxes for a grid box $\I$.
% If $N(\J) \cap \I = \emptyset$, then $N_{1/2}(\I)$ and $N_{1/2}(\J)$ have disjoint
% interiors. This is a key observation with regard to the correctness of our
% algorithm.

%%%%%%%%%%%%%%%%%%%%%%%%%%%%%%%%%%%%%%%%%%%%%%%%%%%%%%%%%%%%%%%%%%
Our algorithm will construct disjoint certified boxes surrounding
a box $\I$. As observed earlier, the surrounding boxes $N(\I)$ and $N(\J)$ of
disjoint boxes $\I$ and $\J$ may intersect. 
Since our algorithm will construct disjoint certified boxes surrounding
a box $\I$, its surrounding box should be smaller than the box $N(\I)$. 
To achieve this, note $N_\varrho(\I)$ be the box obtained by multiplying
box $\I$ from its center by a factor of $1+\varrho$, where $\tfrac{1}{2}
\leq \varrho \leq 1$. In particular, $N(\I) = N_1(\I)$. If $N(\J) \cap \I = \emptyset$, 
then $N_{1/2}(\I)$ and $N_{1/2}(\J)$ have disjoint interiors. 
% for boxes $\I$ and $\J$ for which $N(\J)$ and $\I$ do not have a common interior point.
This is a key observation with regard to the correctness of our
algorithm.
\begin{lem}
  \label{lemma:isolatingIntervals}
  Let $\I$ be a box such that conditions $\neg C_0(\I)$, $C_1(\I)$,
  and $C_2(\I)$ hold. Let $d$ be the length of its
  edges, and let $\frac{1}{2} \leq \varrho \leq 1$.
  \\
  1.~If $Z_f$ intersects $\I$, it intersects the boundary of $N_\varrho(\I)$
  transversally at exactly two points. At a point of intersection of
  $Z_f$ and an edge $e$ of $\partial N_\varrho(\I)$ the angle between $Z_f$ and
  $e$ is at least $\frac{1}{15}\pi$.
  \\
  2.~If $\I$ contains a point of intersection of $Z_f$ and $Z_g$, then
  the points of intersection of $Z_f$ and $\partial N_\varrho(\I)$ are at
  distance at least   $2\varrho d\tan\tfrac{1}{2}\alpha(\I)$
  from the points of intersection of $Z_g$ and $\partial N_\varrho(\I)$. 
  On $\partial N_\varrho(\I)$, the points of intersection with $Z_f$ and
  with $Z_g$ are alternating.
\end{lem}
\begin{figure}[h]
  \centering
  \psfrag{p}{$p$}
  \psfrag{q}{$q$}
  \psfrag{op}{$\overline{p}$}
  \psfrag{oq}{$\overline{q}$}
  \psfrag{r}{$r$}
  \psfrag{s}{$s$}
  \psfrag{beta}{$\beta$}
  \psfrag{grf}{$\gradient{f}(s)$}
  \psfrag{Zf}{$Z_f$}
  \psfrag{I}{$\partial\I$}
  \psfrag{NI}{$\partial N_\varrho(\I)$}
  \psfrag{d4}{$\varrho d$}
  \psfrag{d34}{$(1+\varrho)d$}
  \includegraphics[width=0.40\textwidth]{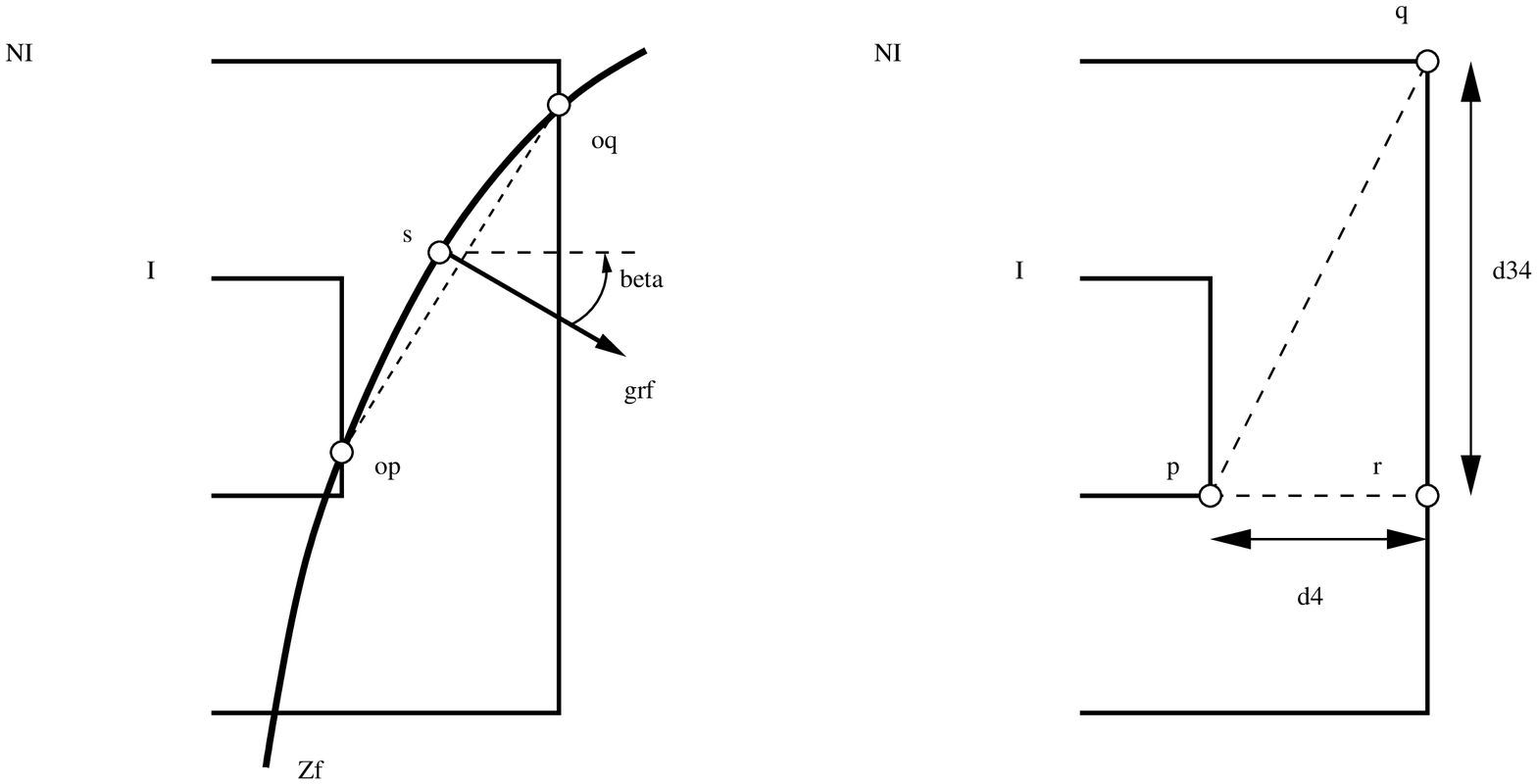}
  \caption{Lower bound on angle of intersection of $Z_f$ and the boundary of the
    surrounding box $N_\varrho(\I)$ in case $Z_f$ intersects $\I$.}
  \label{fig:lowerBound}
\end{figure}
\begin{proof}
  1.~Assume that $Z_f$ intersects a vertical edge of $\partial N_\varrho(\I)$ at
  $\ol{q}$. Let $\ol{p}$ be a point on $Z_f \cap \I$.
  See Figure~\ref{fig:lowerBound}. Then there is a point $s$ on the curve
  segment $\ol{p}\ol{q}$ at which the gradient of $f$ is perpendicular to this
  line segment $\ol{p}\ol{q}$. Let $\beta$ be the (smallest) angle between
  $\gradient{f}(s)$ and the horizontal direction.
  Referring to the rightmost picture in Figure~\ref{fig:lowerBound}, we
  see that this angle is not less than $\angle pqr =
  \tfrac{1}{2}\pi-\angle qpr$. Since $\norm{q-p} = d\sqrt{1+2\varrho+2\varrho^2}$,
  it follows that
  \begin{equation*}
    \beta \geq \frac{\pi}{2} - \arccos\frac{\varrho}{\sqrt{1+2\varrho+2\varrho^2}} > \frac{\pi}{10},
  \end{equation*}
  where the last inequality holds since $\frac{1}{2}\leq \varrho \leq 1$.
  Since condition $C_2(N(\I),f)$ holds, the angle between the gradients of
  $f$ at $s$ and $\ol{q}$ is less than $\frac{\pi}{30}$. Therefore, the
  angle between $Z_f$ and the vertical edge of $\partial N_\varrho(\I)$ at $\ol{q}$
  is at least $\frac{\pi}{10} - \frac{\pi}{30} = \frac{\pi}{15}$.
  \begin{figure}[h]
    \centering
    \psfrag{p}{$p$}
    \psfrag{q}{$q$}
    \psfrag{r}{$r$}
    \psfrag{op}{$\ol{p}$}
    \psfrag{o}{$\ol{\alpha}$}
    \includegraphics[width=0.2\textwidth]{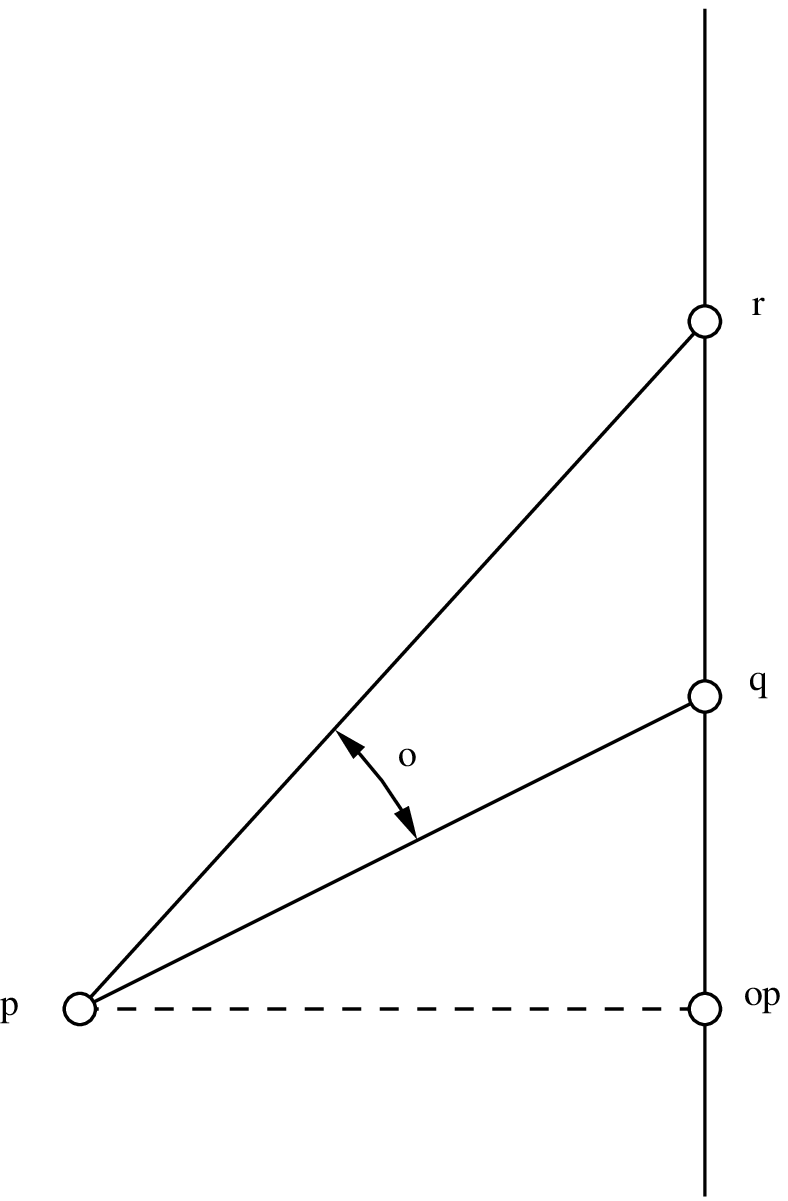}
    \caption{A lower bound for the distance between
      points of $Z_f$ and $Z_g$ on the boundary of $N_\varrho(\I)$.}
    \label{fig:distanceIntersections}
  \end{figure}
  
  \medskip\noindent
  2.~Let $p \in \I$ be the point of intersection of $Z_f$ and $Z_g$.
  Suppose $Z_f$ and $Z_g$ intersect the vertical edge $e$ of $\partial
  N_\varrho(\I)$ in $q$ and $r$, respectively. Then there is a point $s$ on $Z_f$ between 
  $p$ and $q$ at which the gradient of $f$ is perpendicular to $pq$, and there is a
  point on $Z_g$ between $p$ and $r$ at which the gradient of $g$ is perpendicular to
  $pr$. % In view of condition $C_4(\I)$, the angle between $pq$ and $pr$ is greater
  % than or equal to $\alpha(\I)$.
  Let $\alpha(\I)$ be the lower bound on the angle between $\nabla f(p)$ and $\nabla g(q)$ where $p$, $q$
  range over the box $\I$.
  See   Figure~\ref{fig:distanceIntersections}, where $\ol{\alpha} \geq \alpha(\I)$.
  For fixed $\ol{\alpha}$, the distance between $q$ and $r$ is minimal if the
  projection $\ol{p}$ of $p$ on the edge $e$ is the midpoint of $qr$, in which
  case $\norm{r-q} = 2\norm{p-\ol{p}}\,\tan\tfrac{1}{2}\ol{\alpha}$.
  Since $\norm{p-\ol{p}} \geq \varrho d$ and $\ol{\alpha} \geq \alpha(\I)$, 
  the distance between $q$ and $r$ is at least
  $2 \varrho d\tan\tfrac{1}{2}\alpha(\I)$. 
\end{proof}

The following result gives an estimate on the position of the points at
which $Z_f$ intersects the boundary of the surrounding box of $\I$ in
case $Z_f$ intersects an edge of $\partial \I$ in at least two points.
For an edge $e$ of the inner box $\I$ let $l$ and $r$ be the points of intersection
of the line through $e$ and the edges of the surrounding box $N_\varrho(\I)$,
perpendicular to $e$. See also Figure~\ref{fig:twoIntersections}.
The dyadic intervals on the boundary of this surrounding box with length at most
$\frac{2}{\sqrt{3}}(1+\varrho)d$, centered at $l$ and $r$, respectively, 
are denoted by $L_\varrho(e)$ and $R_\varrho(e)$, where $d$ is the length of the
edges of $\I$, and $\tfrac{1}{2} \leq \varrho \leq 1$. Here, we note that in the follwing lemma~\ref{lemma:twoIntersections}
 intervals $L_\varrho(e)$ and $R_\varrho(e)$ can be made corresponding dyadic intervals by 
replacing its real endpoints with suitable conservative dyadic numbers satisfying the conditions of the lemma.

\begin{lem}
  \label{lemma:twoIntersections}
  Let $\I$ be a box such that $\neg C_0(\I)$, $C_1(\I)$ and $C_2(\I)$ hold, and let 
  $e$ be one of its edges. Let $\tfrac{1}{2} \leq \varrho \leq 1$.
  \begin{enumerate}
  \item 
    If $Z_f$ intersects an edge $e$ of the boundary of $\I$ in at least two  points,
    then it transversally intersects $\partial N_\varrho(\I)$ in exactly two points,
    one in each of the dyadic intervals $L_\varrho(e)$ and $R_\varrho(e)$.
  \item 
    If $Z_f$ intersects $\partial N_\varrho(\I)$ in the dyadic intervals $L_\varrho(e)$ and
    $R_\varrho(e)$, then these intersections are transversal, and $Z_f$ intersects
    $\partial N_\varrho(\I)$ at exactly two points, one in each of these intervals. 
  \end{enumerate}
\end{lem}
\begin{figure}[h]
  \centering
  \psfrag{p}{$p$}
  \psfrag{q}{$q$}
  \psfrag{r}{$r$}
  \psfrag{s}{$s$}
  \psfrag{l}{$l$}
  \psfrag{d}{$d$}
  \psfrag{NI}{$\partial N_\varrho(\I)$}
  \psfrag{rhod}{$\varrho d$}
  \psfrag{grf}{$\nabla f$}
  \psfrag{Zf}{$Z_f$}
  \psfrag{Lre}{$L_\varrho(e)$}
  \psfrag{Rre}{$R_\varrho(e)$}
  \psfrag{d25}{}
  \psfrag{d25a}{$\frac{1}{\sqrt{3}}(1+\varrho)d$}
  \includegraphics[width=0.38\textwidth]{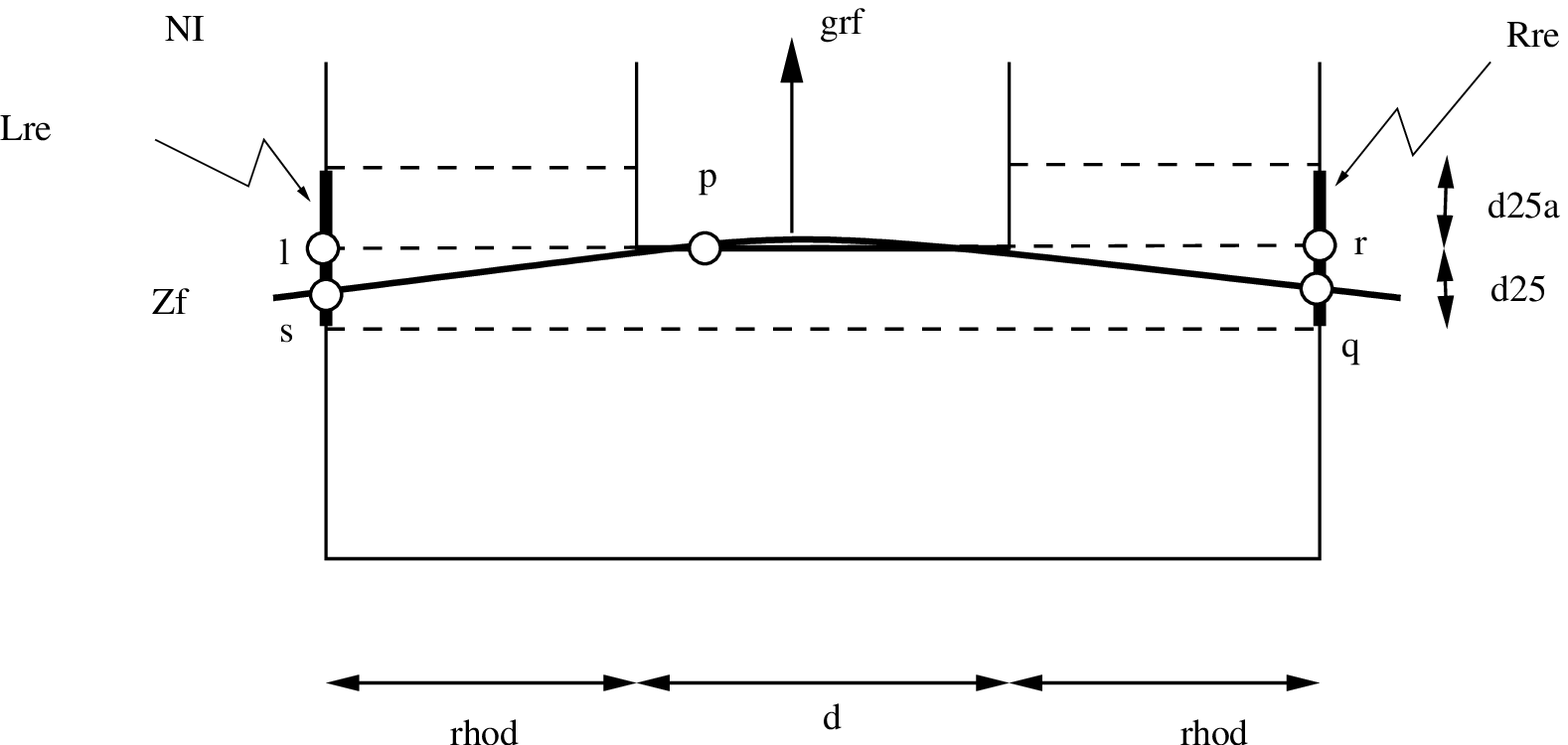}
  \caption{Intervals containing points $Z_f \cap \partial N_\varrho(\I)$.}
  \label{fig:twoIntersections}
\end{figure}
\begin{proof}
  1.~There is a point on $Z_f$ between the points of intersection with $e$
  at which the gradient of $f$ is perpendicular to $e$.
  See Figure~\ref{fig:twoIntersections}. The small normal variation
  condition $C_2(N(\I),f)$ implies that $Z_f$ does not intersect any of the
  two edges of $\partial N_\varrho(\I)$ parallel to $e$, and that it intersects
  each of the edges of $\partial N_\varrho(\I)$ perpendicular to $e$
  transversally in exactly one point. Let $q$ be the point of
  intersection with the edge containing $r$. Then there is a point on
  $Z_f$ at which the gradient of $f$ is perpendicular to the line
  segment $pq$. Since the angle between the gradients of $f$ at two
  points of $N_\varrho(\I)$ does not differ by more than $\frac{1}{30}\pi$, 
  we have $\angle qpr < \frac{1}{30}\pi$.
  Therefore,
  \begin{equation*}
    \norm{q-r} = \norm{p-r} \tan\angle qpr < (1+\varrho) d \tan\tfrac{\pi}{30} <
    \tfrac{1}{\sqrt{3}}(1+\varrho) d.
  \end{equation*}
  In other words, $Z_f$ intersects $R_\varrho$. Similarly, $Z_f$ intersects
  $L_\varrho$. 
  The small normal variation condition $C_2(N(\I))$ implies that there are no
  other intersections with the edges of $\partial N_\varrho(\I)$.

  \medskip\noindent
  2.~Let the points of intersection of $Z_f$ and $R_\varrho(e)$ and $L_\varrho(e)$ be 
  $q$ and $s$, respectively. Then there is a point on $Z_f$ at which $\gradient{f}$
  is perpendicular to $qs$. Since the angle of $qs$ and the vertical direction is at
  most
  \begin{equation*}
    \arctan \frac{(1+\varrho)/8}{1+\varrho} = \arctan\frac{1}{8} < \frac{\pi}{25},
  \end{equation*}
  it follows from the small normal variation condition $C_2(\I)$ that the gradient of
  $f$ at any point of $N_\varrho(\I)$ makes an angle of at most 
  $\frac{\pi}{25} + \frac{\pi}{15} < \frac{\pi}{10}$ 
  with the vertical direction. This rules out multiple intersections with the vertical
  edges of $\partial N_\varrho(\I)$. 
  It also implies that $Z_f$ lies above the polyline $qms$, where $m$ is
  the intersection of the line through $q$ with slope
  $\tan\frac{\pi}{10}$ and the line through $s$ with slope $-\tan\frac{\pi}{10}$.
  Therefore, all points of $Z_f$ lie at distance at most 
  $\frac{1}{4}(1+\varrho)d + \frac{1}{2}(1+\varrho)d\,\tan\frac{\pi}{10} < \varrho d$
  from the line through $e$, so $Z_f$ does not intersect the edges of
  $\partial N_\varrho(\I)$ parallel to $e$.
\end{proof}

%%%%%%%%%%%%%%%%%%%%%%%%%%%%%%%%%%%%%%%%%%%%%%%%%%%%%%%%%%%%%%%%%%
%
% Construction of certified boxes
%
%%%%%%%%%%%%%%%%%%%%%%%%%%%%%%%%%%%%%%%%%%%%%%%%%%%%%%%%%%%%%%%%%%
\subsection{Towards an algorithm}
\label{sec:certifiedBoxes}
After the first subdivision step, we have constructed a finite set $\mcB$ of
boxes, all of the same size, such that $C_0(\I) \wedge C_1(N(\I)) \wedge C_2(N(\I))$
holds for each box $\I$.
% , and such that for each box $\I$ there is an angle 
% $\alpha(\I) \in (0,\pi)$ for which the angle condition (\ref{eq:angle}) holds.
% The algorithm presented in this section will decide, for each box $I \in \mcB$,
% whether there is a point of intersection of $Z_f$ and $Z_g$ inside the enlarged box
% $N_{1/2}(\I)$. If there is such a point of intersection, the level curves $Z_f$ and
% $Z_g$ each intersect the boundary of the box $N_{1/2}(\I)$ in exactly two points,
% called the $f$-points and $g$-points, respectively. The $f$-points and $g$-points
% are interleaved on the boundary of $N_{1/2}(\I)$, and the algorithm will construct
% a certificate in the form of four interleaved isolating intervals for these points of
% intersection. The box $N_{1/2}(\I)$, together with its four isolating intervals, is
% called a \textit{certified box} for a solution of the system of equations
% (\ref{eq:systemEq}). Finally, two distinct certified boxes have disjoint interiors.
For each grid-box $\I$, the algorithm calls one of the following:
\begin{itemize}
\item {\sc Discard}($\I$), if it decides that $\I$ does not contain a solution.
It marks box $\I$ as processed.
\item {\sc ReportSolution}($\I$). It returns the certified pair $(N_{1/2}(\I),N(\I))$, and
marks all boxes contained in $N(\I)$ as processed. 
\end{itemize}
% $\bullet$ 
% \\
% $\bullet$ 

In the latter case a solution is found inside $N_{1/2}(\I)$, but, as will become
clear later, it may not be contained in the smaller box $\I$. 
In view of $C_1(\I)$ none of the grid-boxes in $N(\I)$ contain a solution
different from the one reported, so they are marked as being processed.

Decisions are based on evaluation of the signs of $f$ and $g$ at the vertices of
the grid-boxes (or at certain dyadic points on edges of grid-boxes).
An edge of a box is called \textit{bichromatic} (\textit{monochromatic}) for $f$ if
the signs of the value of $f$ at its vertices are opposite (equal).

% \noindent
\paragraph{Algorithm, case 1: $\I$ has a bichromatic edge for $f$ and a bichromatic edge
  for $g$}
Then $Z_f$ and $Z_g$ intersect $\I$, and, according to
Lemma~\ref{lemma:isolatingIntervals}, part 1, both curves intersect the boundary of
$N_{1/2}(\I)$ transversally in exactly two points.
For each of the two points in $\partial N_{1/2}(\I)$ the algorithm computes an
isolating interval---called an $f$-interval---on $\partial N_{1/2}(\I)$ of length
$\tfrac{1}{2} d \tan\tfrac{1}{2}\alpha(\I)$. The two $g$-intervals are computed
similarly.
% If box $\I$ contains a point of intersection of $Z_f$ and $Z_g$, then the $f$-intervals and the
% $g$-intervals are disjoint and interleaving on $\partial N_{1/2}(\I)$. 
If the $f$- and $g$-intervals are not interleaving, there is no solution of 
(\ref{eq:systemEq}) in box $\I$---even though there may be a solution in
$N_{1/2}(\I)$---and {\sc Discard}($\I$) is called.
This follows from Lemma~\ref{lemma:isolatingIntervals}, part 2.
If the intervals are interleaving, then there is a point of intersection
inside $N_{1/2}(\I)$, so the algorithm calls {\sc ReportSolution}($\I$).
%\medskip
% \noindent
\paragraph{Algorithm, case 2: $\I$ contains no bichromatic edge for $f$ ($g$),
and at least one bichromatic edge for $g$ ($f$, respectively)}
We only consider the case in which all edges of $\I$ are monochromatic for $f$.
Then the algorithm also evaluates the sign of $f$ at the vertices of the box
$N_{1/2}(\I)$.
If $N_{1/2}(\I)$ has no \textit{disjoint} bichromatic edges (as in the fourth and fifth
configuration of Figure~\ref{fig:signPatterns}), the isocurve $Z_f$ does
not intersect $\I$, so the algorithm calls {\sc Discard}($\I$).
\begin{figure}[h]
  \centering
  \includegraphics[width=0.45\textwidth]{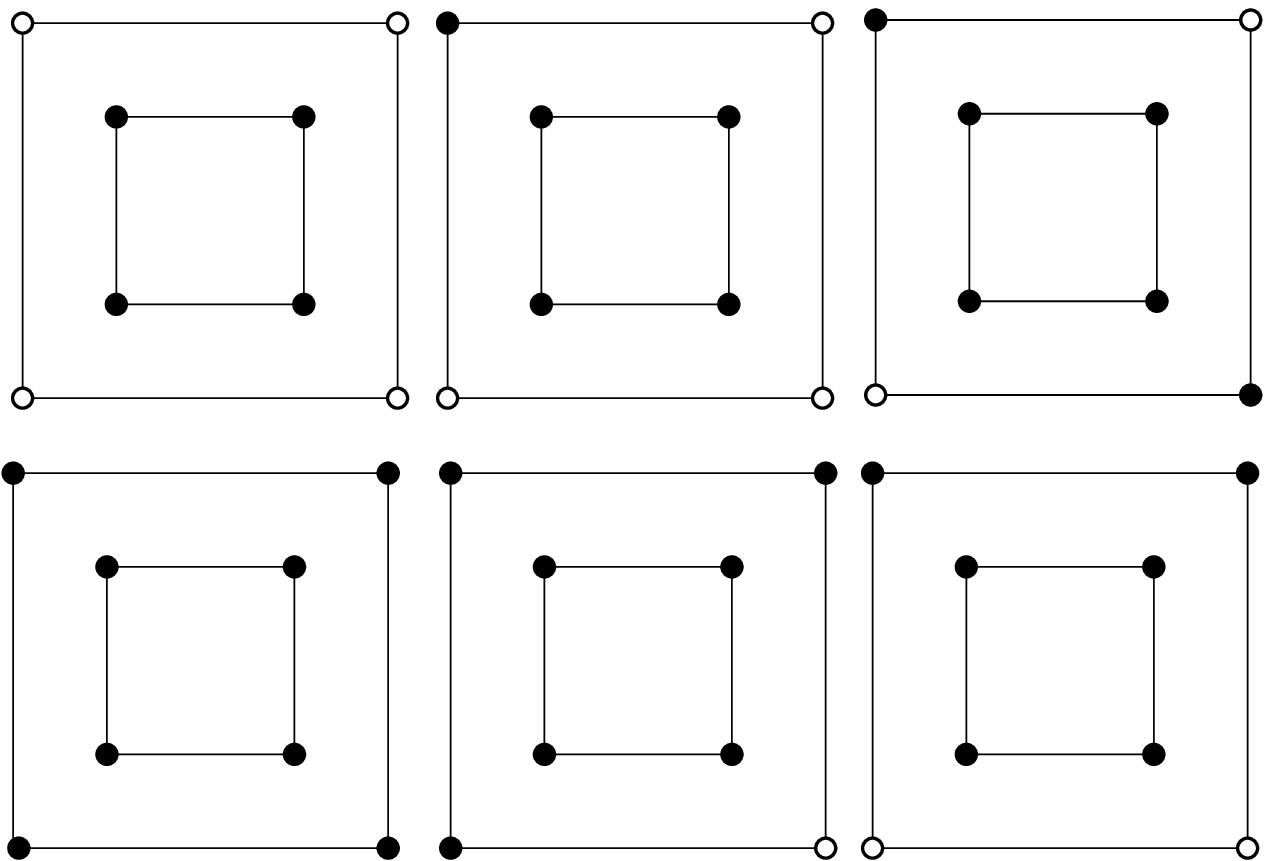}
  \caption{Sign patterns of the box $N_{1/2}(\I)$ enclosing the grid-box $\I$
    with monochromatic edges for $f$. The three top configurations
    are ruled out by the small normal variation condition $C_2(\I)$. 
    The fourth, fifth and sixth configuration are all possible, but only in the
    sixth situation $Z_f$ may intersect the inner box.
  }
  \label{fig:signPatterns}
\end{figure}
To deal with the remaining case, in which $N_{1/2}(\I)$ has two \textit{disjoint}
bichromatic edges (as in the sixth configuration in Figure~\ref{fig:signPatterns}) we
need to evaluate the sign of $f$ at certain dyadic points of these bichromatic edges, followed
from Lemma~\ref{lemma:twoIntersections}.
%\paragraph{Algorithm, case 2 (cont'd).}
By evaluating the signs of $f$ at the (dyadic) endpoints of the interval 
$L_{\varrho}(e)$ and $R_{\varrho}(e)$ the algorithm decides whether they contain a
point of intersection with $Z_f$.
If at least one of these intervals is disjoint from $Z_f$, then {\sc Discard}($\I$)
is called. 
Otherwise, the algorithm computes isolating $f$- and $g$-intervals of length
$\tfrac{1}{2} d \tan \tfrac{1}{2}\alpha(\I)$. 
As in case~1, the algorithm calls {\sc ReportSolution}($\I$) if these intervals are
interleaving, and {\sc Discard}($\I$) otherwise.
%\medskip
% \noindent
\paragraph{Algorithm, case 3: all edges of $\I$ are monochromatic for both $f$ and $g$}
Again, let $e$ be the (unique) edge of $\I$ closest to the edge of $N_{1/2}(\I)$ which is
monochromatic for $f$, at whose vertices the sign of $f$ is the opposite of the sign
of $f$ at the vertices of $\I$.
Edge $e'$ of $\I$ is defined similarly for $g$. 
\\
\textbf{Case 3.1: $e=e'$}. 
In this case $Z_f$ or $Z_g$ does not intersect $\I$. Indeed, 
if $Z_f$ intersects $\I$, it intersects $e$ in at least two points, so there is
a point $p \in Z_f$ at which $\gradient{f}(p)$ is perpendicular to $e$. 
Condition $C_1(\I)$ guarantees that $\gradient{g}$ is nowhere parallel to
$\gradient{f}(p)$, so $Z_g$ does not intersect $e$, and, hence, does not intersect
$\I$. Therefore, {\sc Discard}($\I$) is called.
\\
\textbf{Case 3.2: $e \neq e'$}. 
If $Z_f$ does not intersect $L_{\varrho}(e)$ or $R_{\varrho}(e)$, or
if $Z_g$ does not intersect $L_{\varrho}(e')$ or $R_{\varrho}(e')$, then, as in
case~2, the algorithm calls {\sc Discard}($\I$).
Otherwise, $L_{\varrho}(e)$ or $R_{\varrho}(e)$ are isolating $f$-intervals which are
disjoint from the isolating $g$-intervals $L_{\varrho}(e')$ or $R_{\varrho}(e')$.
If $e$ and $e'$ are perpendicular, then these $f$- and $g$-intervals are
interleaving, and, hence, {\sc ReportSolution}($\I$) is called. Otherwise, there is
no solution in $\I$, so {\sc Discard}($\I$) is called.

\paragraph{Refinement: disjoint surrounding boxes}
We would like distinct isolating boxes $\I,\J$ to have disjoint surrounding
boxes $N(\I), N(\J)$.  There is a simple way to ensure this: we just
use the predicate $C_1(N(\I))$ to instead of $C_1(\I)$ in the above subdivision
process.  Then, if the interior of $N(\I)\cap N(\J)$ is non-empty, we can
discard any one of $\I$ or $\J$.

% \paragraph{Refinement: disjoint surrounding boxes.}
% The algorithm can be refined to compute isolating boxes for all
% solutions, such that for two distinct isolating boxes $\I$ and $\J$ also the
% surrounding boxes $N(\I)$ and $N(\J)$ are disjoint.
% To this end we recursively subdivide any pair of boxes $\I$ and $\J$ for
% which $N(\I) \cap N(\J) \neq \emptyset$, yielding a nested sequence of 
% box pairs $(\I_n,\J_n)$, such that
% $\operatorname{width}(\I_{n+1}) = \frac{1}{2}\operatorname{width}(\I_n)$
% and
% $\operatorname{width}(\J_{n+1}) = \frac{1}{2}\operatorname{width}(\J_n)$.
% After a finite number of subdivisions 
% $\operatorname{width}(\I_n) + \operatorname{width}(\J_n)$ is less than
% half the distance between the zeros contained in these boxes, 
% so even the surrounding boxes $N(\I_n)$ and $N(\J_n))$ are disjoint. 
% Since all zeros are non-degenerate, their number is finite, so there is
% a positive lower bound on their distances.
% Therefore, this subdivision process terminates with all
% surrounding boxes disjoint.

%%% Local Variables: 
%%% mode: latex
%%% TeX-master: "saddle_analysis"
%%% End: 

%% file: Tex/4-refinedBoxes.tex
%%%%%%%%%%%%%%%%%%%%%%%%%%%%%%%%%%%%%%%%%%%%%%%%%%%%%%%%%%%%%%%%%%
%
% Isolating boxes for sinks, sources and saddles
%
%%%%%%%%%%%%%%%%%%%%%%%%%%%%%%%%%%%%%%%%%%%%%%%%%%%%%%%%%%%%%%%%%%
\section{Isolating boxes for sinks, sources and saddles}
\label{sec:singularities}
In a first step, described in Section~\ref{sec:eqSolving}, we have constructed
certified disjoint isolating boxes 
$\B_1^{'}\ldots,\B_m^{'}$  the singular points of $\gradient{h}$ in the
domain $\D$ of $h$.
Let $\D^*$ be the closure of $\D \setminus (\B_1^{'}\cup \cdots \cup \B_m^{'})$.
Obviously, $\D^{*}$ is a compact subset of $\R^2$.

In a second step towards the construction of the MS-complex, we refine
the saddle-, sink- and sourceboxes.
In Section~\ref{sec:refinedSaddle} we show how to augment each saddlebox by
computing four arbitrarily small
disjoint intervals in its boundary, one for each intersection of a
stable or unstable separatrix with the box boundary. 
Subsequently, in Section~\ref{sec:refinedMinMax}, we show how to 
construct for each source or sink of $\gradient{h}$ (minimum or
maximum of $h$) a box on the boundary of which the gradient
field is pointing outward or inward, respectively.
These boxes are contained in the source- and sinkboxes constructed in
the previous section, but are not necessarily axes-aligned.

\subsection{Refining saddle boxes: Isolating separatrix intervals}
\label{sec:refinedSaddle}
To compute \textit{disjoint} certified separatrix intervals we consider
wedge shaped regions with apex at the saddle point, enclosing the
unstable and stable manifolds of the saddle point.
Even though the saddle point is not known exactly, we will show how to
determine certified intervals for the intersection of these wedges and
the boundary of a saddle box.

First we determine the eigenvalues and eigenvectors of the Hessian of
$h$ at a point $(x_0,y_0)$ in the interior of the saddlebox $\I$---its
center point, say---and 
consider these as good approximations to the eigenvalues and
eigenvectors of the Hessian, i.e., the linear part of $\gradient{h}$, at
the saddle point.
Let $H$ be the Hessian, i.e.,
\begin{equation}
  \label{eq:Hessian}
  H = 
  \begin{pmatrix}
    h_{xx} & h_{xy} \\[1.2ex]
    h_{xy} & h_{yy}
  \end{pmatrix},
\end{equation}
and let $H^0$ be the Hessian evaluated at $(x_0,y_0)$.
The eigenvalues $\lambda_u$ and $\lambda_s$ of $H^0$ are given by
\begin{align*}
  \lambda_u &= \tfrac{1}{2}\,(h^0_{xx}+h^0_{yy} + \sqrt{(h^0_{xx}-h^0_{yy})^2+4(h^0_{xy})^2})\\[1.2ex]
  \lambda_s &= \tfrac{1}{2}\,(h^0_{xx}+h^0_{yy} - \sqrt{(h^0_{xx}-h^0_{yy})^2+4(h^0_{xy})^2}),
\end{align*}
and the corresponding eigenvectors are
\begin{equation}
  \label{eq:eigenvectors}
  V^u = 
  \begin{pmatrix}
    h^0_{xy}\\[1.2ex]
    \lambda_u-h^0_{xx}
  \end{pmatrix},
  \qquad
  V^s = 
  \begin{pmatrix}
    h^0_{xx}-\lambda_u\\[1.2ex]
    h^0_{xy}
  \end{pmatrix}.
\end{equation}
The singular point is a saddle, so we have $\lambda_s < 0 < \lambda_u$.
Since $H^0$ is a symmetric matrix, its eigenvectors are orthogonal. More precisely,
\begin{equation*}
  V^s =
  \begin{pmatrix}
    -V^u_2 \\
    V^u_1
  \end{pmatrix}
  = 
  R_{\pi/2}(V^u).
\end{equation*}
Here $R_\alpha$ denotes counterclockwise rotation over an angle $\alpha$.
Therefore, 
\begin{equation}
  \label{eq:VsVu}
  \norm{V^s}=\norm{V^u}, \text{ and } \det(V^u,V^s)=\norm{V^u}^2.
\end{equation}
The stable and unstable eigenvectors $V^s$ and $V^u$ are good
approximations of the tangent vectors of the stable and unstable
manifolds of the saddle point. These invariant manifolds are contained
in wedge-shaped regions, which are defined as follows.
\begin{dfn}
  Let the (orthogonal) vectors $V^u$ and $V^s$ be the stable and
  unstable eigenvectors of the Hessian of $h$ at the center of a saddle
  box $\I$, and let $\beta\in(0,\frac{\pi}{8})$. 
  The \textit{unstable wedge} $C^u_\beta$ is the set of points in the
  surrounding box $N(\I)$ at which the (unsigned) angle between $\gradient{h}$ and
  $V^u$ is at most $\beta$. See Figure~\ref{fig:wedge}.
  Similarly, the \textit{stable wedge} $C^s_\beta$ is the set of points in the
  surrounding box $N(\I)$ at which the (unsigned) angle between $\gradient{h}$ and
  $V^s$ is at most $\beta$.
\end{dfn}
\begin{figure}[h]
  \centering
  \psfrag{Gplus}{$\Gamma^u_{\beta}$}
  \psfrag{Gmin}{$\Gamma^u_{-\beta}$}
  \psfrag{Vs}{$V^s$}
  \psfrag{Vu}{$V^u$}
  \psfrag{Cs}{$C^s_\beta$}
  \psfrag{Cu}{$C^u_\beta$}
  \psfrag{I}{$\I$}
  \psfrag{NI}{$N(\I)$}
  \psfrag{grh}{$\gradient{h}$}
  \psfrag{Xb}{$X_\beta$}
  \includegraphics[width=0.40\textwidth]{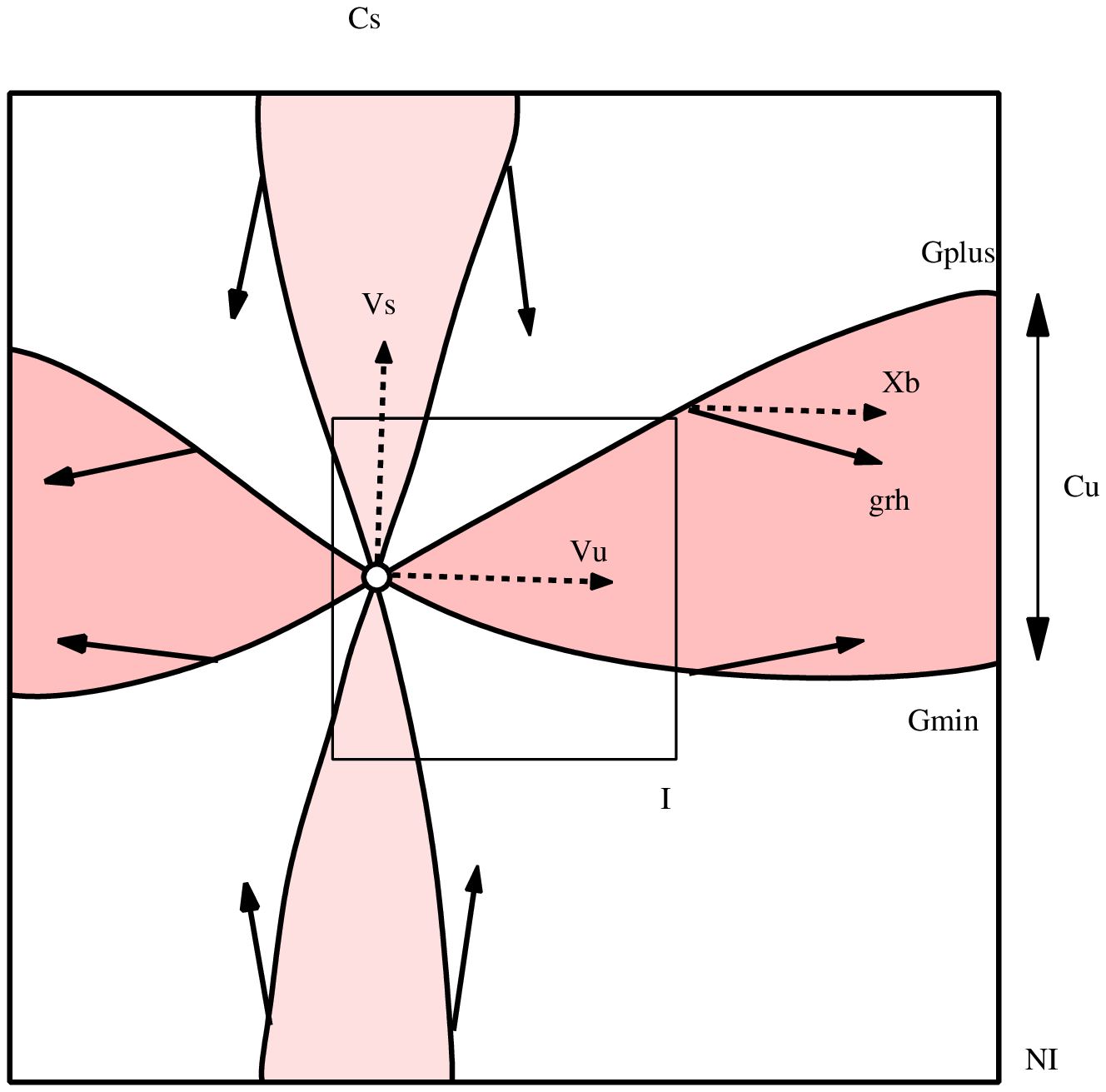}
  \caption{The unstable wedge $C^u_\beta$ enclosing the unstable
    separatrix,
    and the stable wedge $C^s_\beta$ enclosing the stable separatrix.
    The gradient vector field $\gradient{h}$, represented by solid arrows, is
    transversally pointing inward along the boundary of the unstable
    wedge,
    and outward along the boundary of the stable wedge.
    At points of the unstable wedge boundary $\Gamma^u_{\beta} \cup \Gamma^u_{-\beta}$ the vector field
    $X_\beta$ is parallel to $V^u$ or $-V^u$, so $\gradient{h}$ makes an angle
    $-\beta$ with $V^u$ or $-V^u$ there.
  }
  \label{fig:wedge}
\end{figure}

\noindent
The saddle point belongs to both the stable and the unstable wedge.
Since $V^u$ and $V^s$ are orthogonal, and $0 < \beta < \frac{\pi}{8}$,
this is the only common point of the stable and unstable wedge.

\paragraph{Conditions}
We now introduce additional conditions, which guarantee that each of the wedge-boundaries consist 
of two regular curves, cf lemma~\ref{lemma:smallAngleGradPsi}. In fact, these conditions 
guarantee that $C^u_\beta$ and $C^s_\beta$ are really wedge-shaped.
Fix $a > 1$, and let $\delta > 0$ be an arbitrarily small constant (to be specified
later).
At the point $(x_0,y_0)$ we have $HV^u = \lambda_u V^u$,
$HV^s=\lambda_sV^s$, so $N(\I)$ can be taken small enough to guarantee
that the following condition is satisfied at all points of $N(\I)$:
\\[1.2ex]
\textbf{Condition I($\mathbf{a,\I}$).} At every point of the box $N(\I)$ the following inequalities
hold:
\begin{align*}
  \frac{1}{a}\lambda_u\cdot\norm{V^u} 
  & \leq \norm{HV^u} \leq a \lambda_u\cdot\norm{V^u},\\[1.2ex]
  \frac{1}{a}|\lambda_s|\cdot\norm{V^s} 
  & \leq \norm{HV^s} \leq a |\lambda_s|\cdot\norm{V^s}.
\end{align*}
At the point $(x_0,y_0)$ we also have $\ip{HV^u}{V^u}=\lambda_u\norm{V^u}^2$,
$\ip{HV^s}{V^s}=\lambda_s\norm{V^s}^2$, and
$\ip{HV^s}{V^u}=0$.
Therefore, for any $\delta > 0$, the box $N(\I)$ can be taken small enough such that
the following condition holds:

\medskip\noindent
\textbf{Condition II($\mathbf{\delta,\I}$).} At every point of the box $N(\I)$ the
following inequalities hold:
\begin{align}
  \ip{HV^u}{V^u} & \geq \tfrac{1}{2}\lambda_u\,\norm{V^u}^2  \label{eq:HVuu}\\[1.2ex]
  \ip{HV^s}{V^s} & \geq \tfrac{1}{2}|\lambda_s|\,\,\norm{V^u}^2 \label{eq:HVss} \\[1.2ex]
  |\ip{HV^s}{V^u}| & \leq \delta\,\norm{V^u}^2    \label{eq:HVsu}
\end{align}
Since $H$ is symmetric, (\ref{eq:HVsu}) also implies 
$|\ip{HV^u}{V^s}| \leq \delta\,\norm{V^u}^2$.

\bigskip\noindent
On the boundary of $C^u_\beta$ the gradient field makes a (signed)
angle $\pm\beta$ or $\pi\pm\beta$ with $V^u$, or, in other words,
$X_{\pm\beta}$ is (anti)parallel to $V^u$.
Again, $X_\beta$ is the vector field obtained by rotating $\gradient{h}$
over an angle $\beta$. 
So let $\Gamma^u_{\pm\beta}$ be the curve along which the vector field
$X_{\pm\beta}$ is (anti)parallel to the unstable eigenvector $V^u$.
Then the boundary of the unstable wedge 
is the union of the two curves $\Gamma^u_{\beta}$ and $\Gamma^u_{-\beta}$.
The curve $\Gamma^u_{\beta}$ is defined by the equation
\begin{equation}
  \label{eq:Gamma}
  \psi^u_\beta(x,y) := \det(V^u,X_\beta(x,y)) = 0.
\end{equation}
Obviously, the saddle point lies on $\Gamma^u_{\beta}$.
The function $\psi^u_{-\beta}$ is defined similarly.

Similarly, the boundary of the stable wedge is the union of
curves $\Gamma^s_{\pm\beta}$, along which the vector field
$X_{\pm\beta}$ is (anti)parallel to $V^s$.
The curves $\Gamma^s_{\pm\beta}$ are defined by the equation
\begin{equation*}
  \psi^s_{\pm\beta}(x,y) := \det(V^s,X_{\pm\beta}(x,y)) = 0.
\end{equation*}
The following technical result provides computable upper bounds for the
angle variation of the normals of the boundary curves of the stable and
unstable wedges.
\begin{lem}
  \label{lemma:smallAngleGradPsi}
  Let $\omega_1 \in (0,\tfrac{\pi}{4})$ (to be specified later), 
  let $a > 1$, and let $\I$ be such that
  Condition~$\operatorname{I}(a,\I)$ holds.
  Let $0 < \beta < \frac{\pi}{4}$ and $\delta > 0$ such that
  \begin{align}
    \sin\beta & \leq \frac{\sin\omega_1}{4a^2\sqrt{2}}\,
    \min(\bigl|\frac{\lambda_s}{\lambda_u}\bigr|,\bigl|\frac{\lambda_u}{\lambda_s}\bigr|),
    \label{eq:sinbeta} \\[1.2ex]
    \delta & \leq \frac{\sin\omega_1}{8a}\min(|\lambda_s|,|\lambda_u|),
    \label{eq:delta1} \\[1.2ex]
    \delta & \leq \frac{\tan\beta}{4}\min(|\lambda_s|,|\lambda_u|).
    \label{eq:delta2} 
  \end{align}
  If Condition~$\operatorname{II}(\delta,\I)$ also holds, then at any
  point of $N(\I)$
  \begin{equation}
    \label{eq:angleU}
    \frac{\pi}{2} - \omega_1 
    \leq \operatorname{angle}(\gradient{\psi^u_\beta},V^u) 
    <
    \frac{\pi}{2} 
    <
    \operatorname{angle}(\gradient{\psi^u_{-\beta}},V^u) \leq \frac{\pi}{2} +\omega_1.
  \end{equation}
  and
  \begin{equation}
    \label{eq:angleS}
    \frac{\pi}{2} - \omega_1 
    \leq \operatorname{angle}(\gradient{\psi^s_\beta},V^s) 
    <
    \frac{\pi}{2} 
    <
    \operatorname{angle}(\gradient{\psi^s_{-\beta}},V^s) \leq \frac{\pi}{2} +\omega_1.
  \end{equation}
%  Similar inequalities hold for $\gradient{\psi}^s_{\pm\beta}$ and $V^s$.
  In particular, the angle variation of any of the gradients
  $\gradient{\psi}^u_{\pm\beta}$ and $\gradient{\psi}^s_{\pm\beta}$ over
  $N(\I)$ is less than $2\omega_1$.
\end{lem}
\begin{proof}
  We only show that the angle variation of $\gradient{\psi}^u_{\beta}$ over
  $N(\I)$ is less than $2\omega_1$.
  Since 
  \begin{equation*}
    X_\beta 
    = 
    h_x\,  
    \begin{pmatrix}
      \cos\beta \\
      \sin\beta
    \end{pmatrix}
    +
    h_y\,  
    \begin{pmatrix}
      -\sin\beta \\ 
      \cos\beta
    \end{pmatrix},
  \end{equation*}
  the function $\psi^u_\beta$ satisfies
  $
  \psi^u_\beta = A_\beta\,h_x + B_\beta\,h_y,
  $
  where
  \begin{equation*}
    A_\beta = 
    \det(V^u,
    \begin{pmatrix}
      \cos\beta \\ 
      \sin\beta
    \end{pmatrix}
    )
    \text{~~and~~}
    B_\beta = 
    \det(V^u,
    \begin{pmatrix}
      -\sin\beta \\ 
      \cos\beta
    \end{pmatrix}
    ),
  \end{equation*}
  so
  \begin{equation*}
    \begin{pmatrix}
      A_\beta \\
      B_\beta
    \end{pmatrix}
    =
    \begin{pmatrix}
      \cos\beta & \sin\beta \\
      -\sin\beta & \cos\beta
    \end{pmatrix}
    \begin{pmatrix}
      -V^u_2 \\
      V^u_1
    \end{pmatrix}
    =
    (\cos\beta)\,V^s + (\sin\beta)\,V^u.
  \end{equation*}
  Therefore,
  \begin{equation}
    \label{eq:gradpsi}
    \gradient{\psi^u_\beta} = 
    \begin{pmatrix}
      h_{xx} & h_{xy} \\[1.2ex]
      h_{xy} & h_{yy}
    \end{pmatrix}
    \begin{pmatrix}
      A_\beta \\
      B_\beta
    \end{pmatrix}
    =
    \cos\beta\, (HV^s) +\sin\beta\, (HV^u).
  \end{equation}
  Condition $\operatorname{I}(a,\I)$ implies that
  $HV^u \neq 0 \text{~~and~~} HV^s \neq 0$, and $\{HV^u, HV^s\}$ are independent vectors,
  at all points of $N(\I)$, so the gradient of $\psi^u_\beta$ is nonzero at every
  point of $N(\I)$, so $\Gamma^u_\beta$ is a regular curve.

  \noindent
  Expression (\ref{eq:gradpsi}) for $\gradient{\psi^u_\beta}$ implies that
  \begin{equation*}
    \norm{\gradient{\psi^u_\beta}}^2 =
    \cos^2\beta\,\norm{HV^s}^2 + 2\sin\beta\cos\beta\,\ip{HV^s}{HV^u} +
    \sin^2\beta\,\norm{HV^u}^2.
  \end{equation*}
  Using the Cauchy-Schwarz inequality $|\ip{HV^s}{HV^u}| \leq
  \norm{HV^s}\cdot\norm{HV^u}$ and the fact that $\beta > 0$ we get
  \begin{align}
    \label{eq:normgrpsi2}
    \norm{\gradient{\psi^u_\beta}}^2 
    &\geq 
    \cos^2\beta\,\norm{HV^s}^2 - 2\sin\beta\cos\beta\,\norm{HV^s}\cdot\norm{HV^u} \nonumber \\
    & {} +\sin^2\beta\,\norm{HV^u}^2 \nonumber \\
    &=
    \cos^2\beta\,
    \bigl( \norm{HV^s} - \norm{HV^u}\,\tan\beta\bigr)^2 
  \end{align}
  Since
  $\sin\beta \leq \dfrac{\sin\omega_1}{4a^2\sqrt{2}}\,\bigl|\dfrac{\lambda_s}{\lambda_u}\bigr|$
  and $0 < \beta < \frac{\pi}{4}$, it follows
  from Condition~$\operatorname{I}(a,\I)$ that
  \begin{equation*}
    \norm{HV^u}\,\tan\beta
    \leq     
    a\lambda_u\,\norm{V^u}\,\frac{\sin\beta}{\sqrt{2}}
    <
    \frac{|\lambda_s|}{2a}\,\norm{V^u}.
  \end{equation*}
  Using Condition~I again we get
  \begin{equation*}
    \norm{HV^s} - \norm{HV^u}\,\tan\beta
    \geq 
    \dfrac{|\lambda_s|}{a}\,\cdot\norm{V^u} - \frac{|\lambda_s|}{2a}\,\norm{V^u}
    =
    \frac{|\lambda_s|}{2a}\,\norm{V^u}.
  \end{equation*}
  In view of (\ref{eq:normgrpsi2}) we get, using $\cos\beta \geq \frac{1}{\sqrt{2}}$:
  \begin{equation}
    \label{eq:normgrpsi}
    \norm{\gradient{\psi^u_\beta}}
    \geq 
    \frac{|\lambda_s|}{2a\sqrt{2}}\,\norm{V^u}.
  \end{equation}
  Expression (\ref{eq:gradpsi}) for $\gradient{\psi^u_\beta}$ also implies that
  \begin{align*}
    \ip{\gradient{\psi^u_\beta}}{V^u}
    &=  
    \cos\beta\,\ip{HV^s}{V^u} + \sin\beta\,\ip{HV^u}{V^u} \\[1.2ex]
    &=
    \cos\beta\ip{HV^u}{V^u}\,(\frac{\ip{HV^s}{V^u}}{\ip{HV^u}{V^u}}+\tan\beta).
  \end{align*}
  Condition II and (\ref{eq:delta2}) imply
  \begin{equation*}
    \bigl|\frac{\ip{HV^s}{V^u}}{\ip{HV^u}{V^u}}\bigr|
    \leq 
    \frac{2\delta}{\lambda_u} 
    \leq 
    \tfrac{1}{2}\tan\beta.
  \end{equation*}
  Since $\beta > 0$, this implies $\ip{\gradient{\psi^u_\beta}}{V^u} > 0$ on $N(\I)$.

  \medskip\noindent
  According to Condition~II we have $|\ip{HV^s}{V^u}| \leq \delta\,\norm{V^u}^2$,
  whereas the Cauchy-Schwarz inequality implies
  $|\ip{HV^u}{V^u}| \leq \norm{HV^u}\cdot\norm{V^u} \leq a \lambda_u\norm{V^u}^2$.
  Therefore,
  \begin{equation*}
    0 \leq
    \ip{\gradient{\psi^u_\beta}}{V^u}
    \leq
    (\delta\cos\beta+a\lambda_u\sin\beta)\,\norm{V^u}^2.
  \end{equation*}
  Together with (\ref{eq:normgrpsi}) this implies
  \begin{equation*}
    \ip
    {\frac{\gradient{\psi^u_\beta}}{\norm{\gradient{\psi^u_\beta}}}}%
    {\frac{V^u}{\norm{V^u}}}
    \leq
    \frac{2a\sqrt{2}}{|\lambda_s|}\,(\delta\cos\beta+a\lambda_u\sin\beta)
  \end{equation*}
  Given the upper bounds (\ref{eq:delta1}) for $\delta$ and
  (\ref{eq:sinbeta}) for $\sin\beta$, we get
  \begin{equation*}
    \frac{2a\sqrt{2}}{|\lambda_s|}\,\delta\cos\beta
    \leq 
    \frac{4a}{|\lambda_s|}\,\delta 
    \leq 
    \tfrac{1}{2}\,\sin\omega_1 \text{~~and~~}
    \frac{2a\sqrt{2}}{|\lambda_s|}\,a\lambda_u\sin\beta
    \leq
    \tfrac{1}{2}\,\sin\omega_1.
  \end{equation*}
  Therefore, 
  \begin{equation*}
    0
    <
    \ip
    {\frac{\gradient{\psi^u_\beta}}{\norm{\gradient{\psi^u_\beta}}}}%
    {\frac{V^u}{\norm{V^u}}}
    \leq
    \sin\omega_1 = \cos(\tfrac{\pi}{2}-\omega_1)
  \end{equation*}
  At all points of $N(\I)$ we then have
  \begin{equation*}
    \frac{\pi}{2} - \omega_1 \leq \operatorname{angle}(\gradient{\psi^u_\beta},V^u) < \frac{\pi}{2}.
  \end{equation*}
  Since $V^u$ is constant, the angle variation of $\gradient{\psi^u_\beta}$ over
  $N(\I)$ does not exceed $2 \omega_1$.
\end{proof}

The main result of this subsection states that, under suitable
conditions, the intersection of the boundary of a saddle box and the
stable and unstable wedges can be computed. Moreover, at all points of
these intersections the gradient vector field is transversal to the
boundary of the saddle box, and, even stronger, at these points there is
a computable positive lower bound for the angle of the gradient vector
field and the boundary of the saddle box.
\begin{cor}
  \label{cor:wedge}
  Let $a$ and $\omega_1$ be constants such that $a > 1$, and 
  $\omega_1 = \tfrac{1}{3}\,\arctan\tfrac{1}{2}$.
  Let $\beta \in (0,\omega_1)$ and $\delta >0$ such that
  (\ref{eq:sinbeta}), (\ref{eq:delta1}) and (\ref{eq:delta2}) hold.

  If $\I$ is a saddle box with concentric surrounding box
  $N(\I)$ satisfying Condition $I(a,\I)$ and Condition
  $\operatorname{II}(\delta,\I)$,
  then 
  \begin{enumerate}
  \item
    The saddle point is the only common point of the stable wedge
    $C^u_\beta$ and the unstable wedge $C^s_\beta$.
  \item
    The gradient vector field $\gradient{h}$ is transversal at points on
    the boundary of these wedges, different from the saddle point:
    on the boundary of the unstable wedge it points inward, except at
    the saddle point, and on the boundary of the stable wedge it points
    outward, except at the saddle point.
  \item 
    The unstable wedge $C^u_\beta$ contains the unstable separatrices of
    the saddle point, and the stable wedge $C^s_\beta$ contains the
    stable separatrices.
  \item 
    The unstable wedge intersects the boundary of $N(\I)$ in two
    intervals, called the unstable intervals. Similarly, the stable
    wedge intersects the boundary of $N(\I)$ in two intervals, called
    the stable intervals. These four intervals are disjoint, and
    the unstable and stable intervals occur alternatingly on the
    boundary of $N(\I)$. 
    At each point of a stable or unstable interval the (unsigned) angle
    between $\gradient{h}$ and the boundary edge containing this point
    is at least $\omega_1$.
    Moreover, there are computable isolating intervals for each stable
    and unstable interval.
%    \marginpar{Explain isolating}
  \end{enumerate}
\end{cor}

\begin{proof}
1.~If $p \in C^u_\beta \cap C^s_\beta$, then
$\gradient{h}(p)$ makes an angle $\beta \in (0,\frac{\pi}{4})$ with both
$V^u$ and $V^s$. Therefore, $\gradient{h}(p) = 0$, since these vectors
are orthogonal. Hence, $p$ is a singular point of $\gradient{h}$ inside
$N(\I)$, which is the saddle point.
\\[1.2ex]
2.~Recall from Lemma~\ref{lemma:smallAngleGradPsi} that
$\ip{\gradient{\psi^u_{\beta}}}{V^u}$ is positive on $N(\I)$.  Let $p$ be the saddle
point of $\gradient{h}$ in $\I$.  Since $X_\beta$ is parallel to $V^u$ on one
component of $\Gamma^u_\beta\setminus\{p\}$, and parallel to $-V^u$ on the other
component, it follows that $X_\beta$ is pointing into the unstable wedge along both
components. See again Figure~\ref{fig:wedge}.  Since $\gradient{h}$ is
obtained by rotating $X_\beta$ over $-\beta$, and $X_{-\beta}$ over $\beta$,
also $\gradient{h}$ is pointing into the unstable wedge.  
Similarly, $\ip{\gradient{\psi^u_{-\beta}}}{V^u}$
is negative along both components of $\Gamma^u_{-\beta}\setminus\{p\}$, so
$\gradient{h}$ is also pointing into the unstable wedge along each of these
components.
%%%%%%%%

\noindent
A similar argument shows that $\gradient{h}$ is pointing outward
along each of the boundary components of the stable wedge $C^s_\beta$,
except at the saddle point.
\\[1.2ex]
3.~The second part of the lemma implies that the unstable wedge is
forward invariant under the flow of the gradient vector field
$\gradient{h}$. In particular, it contains the unstable separatrices of
the saddle point. 
Similarly, the stable wedge $C^s_\beta$ is backward invariant, so it contains the
stable separatrices of the saddle point.
\\[1.2ex]
\begin{figure}[h]
  \centering
  \psfrag{p}{$p$}
  \psfrag{q}{$q$}
  \psfrag{op}{$\overline{p}$}
  \psfrag{oq}{$\overline{q}$}
  \psfrag{r}{$r$}
  \psfrag{s}{$s$}
  \psfrag{beta}{$\geq\omega_0$}
  \psfrag{grf}{$\gradient{\psi^u_\beta}(s)$}
  \psfrag{Zf}{$\Gamma^u_{\beta}$}
  \psfrag{I}{$\partial\I$}
  \psfrag{NI}{$\partial N(\I)$}
%   \psfrag{d4}{$\varrho d$}
%   \psfrag{d34}{$(1+\varrho)d$}
  \psfrag{d4}{$d$}
  \psfrag{d34}{$2d$}
  \includegraphics[width=0.45\textwidth]{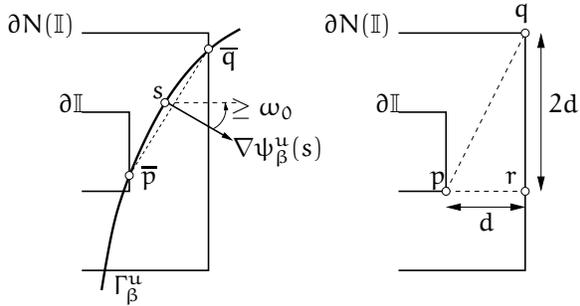}
  \caption{Lower bound on angle of intersection of $\Gamma^u_\beta$ and the boundary of the
    surrounding box $N(\I)$.}
  \label{fig:wedgeAngle}
\end{figure}
4.~Suppose $\Gamma^u_{\beta}$ intersects an edge $e$ of the surrounding
box at a point $\ol{q}$, see Figure~\ref{fig:wedgeAngle}.
We first show that the angle of $\gradient{h}(\ol{q})$ and $e$ is bounded
away from zero. 
To see this, observe that there is a point $s \in \Gamma^u_\beta$ at
which $\gradient{\psi^u_\beta}(s)$ is perpendicular to $\ol{p}\,\ol{q}$.
Therefore, the angle between $\gradient{\psi^u_\beta}(s)$ and the normal
of $e$ is at least  $\omega_0$, where 
$\omega_0 = \arctan\tfrac{1}{2} = 3 \omega_1$.
\\[1.2ex]
The angle between $V^u$ and $\gradient{\psi^u_\beta}$ lies in the
interval $[\tfrac{\pi}{2}-\omega_1,\tfrac{\pi}{2})$,
cf~Lemma~\ref{lemma:smallAngleGradPsi}, so the angle between  $V^u$ and $e$
is at least $\omega_0-\omega_1 = 2 \omega_1$.

At any point of $C^u_\beta$ the angle between $\gradient{h}$ and
$V^u$ is at most $\beta$---by the definition of $C^u_\beta$---so at any
point of $\in C^u_\beta \cap e$ the angle between $\gradient{h}$ and $e$
is at least  $2\omega_1-\beta \geq \omega_1$.

\medskip
To find isolating intervals for the intersection of the stable and
unstable wedges with the boundary of the surrounding box
$N(\I)$, we compute isolating intervals for the intersection of each
of the four curves $\psi^u_{\pm\beta} = 0, \psi^s_{\pm\beta} = 0$
with this boudary.
The normal to each of the curves $\psi^u_{\pm \beta}=0$ makes an angle of at
least $\frac{\pi}{2}-4\omega_1$ with each of the curves $\psi^s_{\pm
  \beta}=0$. This follows from (\ref{eq:angleU}) and (\ref{eq:angleS}),
and the fact that $V^u$ and $V^s$ are perpendicular.
Since $\omega_1 = \tfrac{1}{3}\,\arctan\tfrac{1}{2} < \tfrac{\pi}{20}$,
so the angle between each of the curves $\psi^u_{\pm \beta}=0$ 
and each of the curves $\psi^s_{\pm\beta}=0$ is at least
$\frac{\pi}{2}-4\tfrac{\pi}{20} = \tfrac{3\pi}{10}$, which is bounded
away from zero.
Therefore, the method of Section~\ref{sec:eqSolving} provides such isolating
intervals.
\end{proof}

As will become clear in the certified construction of the MS-complex, we
need to be able to provide certified separatrix intervals of arbitrarily
small width, without refining the saddle box:
\begin{lem}
  \label{lemma:smallSepIntervals}
  Let $\I$ be a separatrix box satisfying the conditions of
  Corollary~\ref{cor:wedge}.
  Then the isolating separatrix intervals in the boundary of
  $N(\I)$ can be made arbitrarily small.
\end{lem}

\noindent
The proof of this result is rather technical. For a sketch we refer to 
\ref{sec:smallSepIntervals}.

\subsection{Refining boxes for maxima and minima}
\label{sec:refinedMinMax}
To construct the MS-complex, the algorithm needs to determine when an
unstable (stable) separatrix will have a given maximum (minimum) of $h$
as its $\omega$-limit ($\alpha$-limit).
For each maximum (minimum) the algorithm determines a certified box
such that the gradient vector field points inward (outward) on the
boundary of the box.
Unfortunately, we cannot always choose an axis-aligned box, as will
become clear from the following example.

Let $h(x,y) = -5x^2-4xy-y^2$, then the origin is a sink of the gradient
vector field
\begin{equation*}
  \gradient{h}(x,y) = 
  \begin{pmatrix}
    -10x-4y \\
    -4x - 2y
  \end{pmatrix}.
\end{equation*}
This vector field is horizontal along the line $y = -2x$, which
intersects the horizontal edges of every axis aligned box centered at
the sink $(0,0)$. In other words, the vector field is not transversal on
the boundary of any such box.

However, there is a box aligned with (approximations) of a pair of
eigenvectors of the linear part of the gradient vector field at (or,
near) its singular point, for which the vector field is transversal 
to the boundary.
To see this, we refine the sink-box to obtain three concentric 
axis aligned boxes $\J_1 \subset \J_2 \subset \J_3$, such that 
\\[1.2ex]
(i) the edge length of $\J_i$, $i=2,3$, is three times the edge length of
$\J_{i-1}$, and
\\[1.2ex]
(ii) the sink is contained in the inner box $\J_1$.
See also Figure~\ref{fig:sinkbox}.
\begin{figure}[h]
  \psfrag{J1}{\tiny$\J_1$}
  \psfrag{J2}{$\J_2$}
  \psfrag{J3}{$\J_3$}
  \psfrag{I1}{$\I_1$}
  \psfrag{I2}{$\I_2$}
  \psfrag{psi2=0}{$\psi_2=0$}
  \psfrag{V1}{$V_1$}
  \psfrag{V2}{$V_2$}
  \centering
  \includegraphics[width=0.45\textwidth]{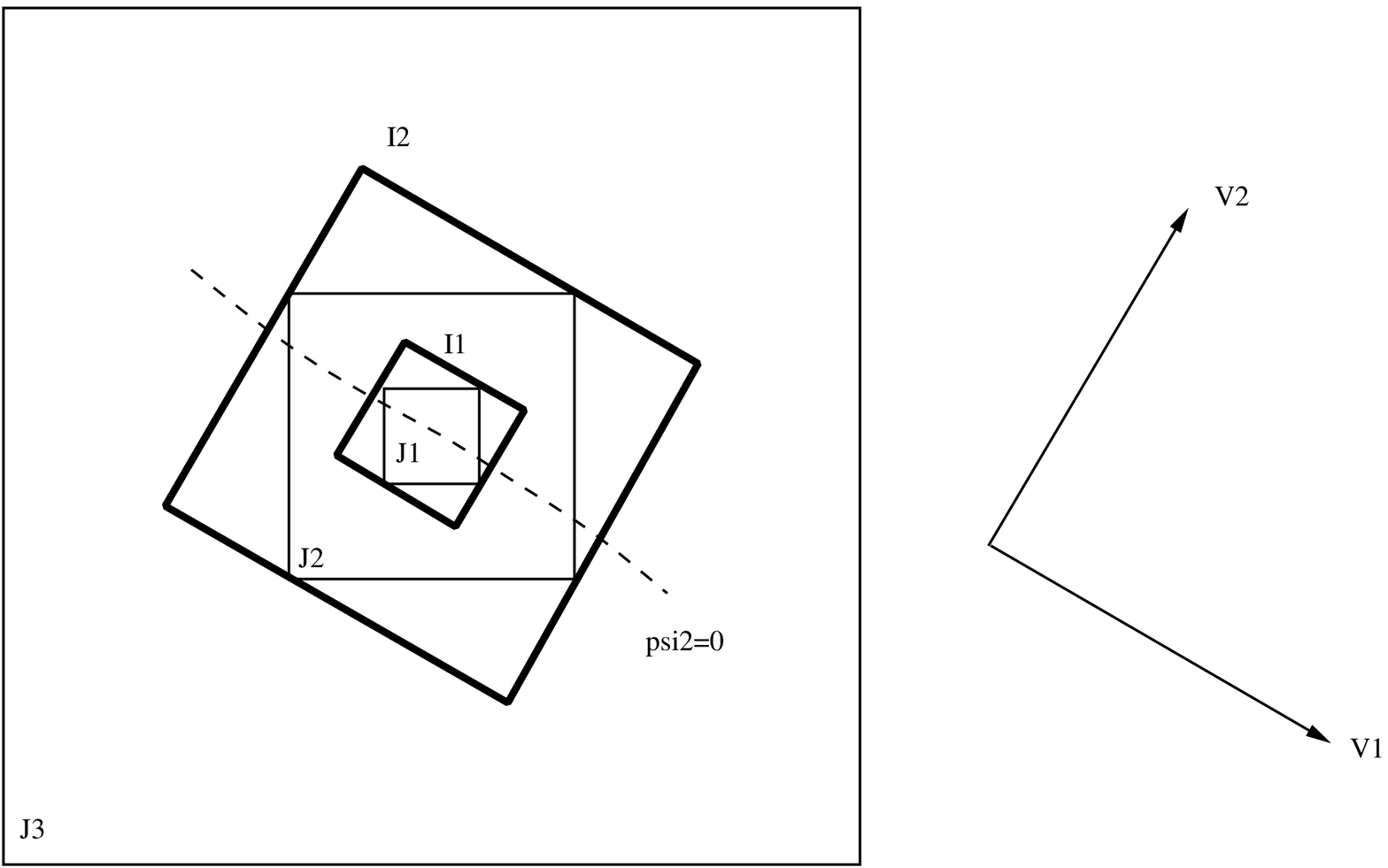}
  \caption{Construction of a sinkbox.}
  \label{fig:sinkbox}
\end{figure}
Moreover, let $V_1$ and $V_2$ be the (orthogonal) eigenvectors of the
Hessian matrix $H^0$ at the center of the boxes. These eigenvectors,
corresponding to the eigenvectors $\lambda_1$ and $\lambda_2$, are
computed as in Section~\ref{sec:refinedSaddle}, cf (\ref{eq:eigenvectors}).

\noindent
We require that 
\\[1.2ex]
(iii) the gradients of the two functions $\psi_1$ and $\psi_2$, defined by
\begin{equation*}
  \psi_i(x,y) = \ip{\gradient{h}(x,y)}{V_i}, \quad i = 1,2,
\end{equation*}
have small angle variation over the outer box $\J_3$.
This condition is made precise in Lemma~\ref{lemma:sinkbox} below.
Note that 
\begin{equation}
  \label{eq:grad-psi12}
  \gradient{\psi_i}(x,y) = H(x,y)V_i,
\end{equation}
where $H(x,y)$ is again the Hessian matrix of $h$ at $(x,y)$. 
Since this matrix is non-singular, we can find a triple of boxes
$\J_1 \subset \J_2 \subset \J_3$, satisfying conditions (i), (ii) and
(iii), such that $H(x,y)$ is nearly constant over the outer box
(again, this is made precise in Lemma~\ref{lemma:sinkbox}).
In particular, $\gradient{\psi_i}$ is nearly parallel to $V_i$, since
$H^0V_i = \lambda_iV_i$.
Now construct boxes $\I_1$ and $\I_2$, which are the smallest boxes
enclosing $\J_1$ and $\J_2$, respectively, with edges parallel to $V^1$
or $V^2$. 

\begin{lem}
  \label{lemma:sinkbox}
  Suppose on the outer box $\J_3$ the following conditions hold:
  \begin{enumerate}
  \item 
    $\norm{HV_i} \geq \tfrac{1}{2}|\lambda_i|\cdot\norm{V_i}$, for $i=1,2$;
  \item 
    $|\ip{HV_1}{V_2}| =|\ip{HV_2}{V_1}| \leq \tfrac{1}{4}\norm{V_1}^2\arctan\tfrac{1}{2}$.
  \end{enumerate}
  Then the gradient vector field is transversal to the boundary of $\I_2$.
\end{lem}
\begin{proof}
  The second condition limits the variation of the angle of $\gradient{\psi_i} =
  HV_i$ and the basis vectors $V_1$ and $V_2$ over $\J_3$.
  Using this bound, we use the same arguments as in
  Section~\ref{sec:eqSolving}, applied to the pair of boxes $\I_1$,
  $\I_2$, to show that the curve $\psi_i(x,y) = 0$
  does not intersect the edges of $\I_2$ perpendicular to $V_i$. 
  Therefore, $\psi_i=\ip{\gradient{h}}{V_i}$ is nowhere zero on
  these edges, so $\gradient{h}$ is nowhere tangent to these edges.
  Therefore, $\I_2$ is the desired sink box, on the boundary of which
  $\gradient{h}$ is pointing inward. In other words, if an unstable
  separatrix intersects the boundary of this box, the part of the
  separatrix beyond this point of intersection lies inside the sink-box.
  Certified source-boxes are constructed similarly.
\end{proof}

%%% Local Variables: 
%%% mode: latex
%%% TeX-master: "saddle_analysis"
%%% End: 

%% file: Tex/5-funnels.tex
% \paragraph{Summary of conditions imposed on a saddle box.}
% Here we collect all conditions introduced in the proof of Lemma~\ref{lemma:wedge}.
%%%%%%%%%%%%%%%%%%%%%%%%%%%%%%%%%%%%%%%%%%%%%%%%%%%%%%%%%%%%%%%%%%
%
% Isolating separatrix strips
%
%%%%%%%%%%%%%%%%%%%%%%%%%%%%%%%%%%%%%%%%%%%%%%%%%%%%%%%%%%%%%%%%%%
\section{Isolating funnels around separatrices}
\label{sec:sepStrips}
% To compute isolating strips (funnels) for unstable separatrices of $\gradient{h}$
% we approximate, for each isolating unstable separatrix interval,
% the forward orbits of its endpoints. Isolating funnels for stable
% separatrices are computed similarly.
If the forward orbits of the endpoints of an unstable separatrix interval
have the same sink of $\gradient{h}$ as $\omega$-limit, these forward
orbits bound a region around the unstable separatrix leaving the saddle
box via this unstable segment. This region is called a \textit{funnel}
for the separatrix (this terminology is borrowed from~\cite{jw-dedsa-91}).

In this section we provide the details of the construction of a
certified funnel for each separatrix.
First we show, in Section~\ref{sec:ConstructionStrips} how to construct
two polylines per separatrix interval those are candidates for the funnel
around the corresponding separatrix.
Then, in Section~\ref{sec:widthFunnel}, we introduce the notion of
\textit{width of a funnel}, and derive upper bounds for its growth.
These upper bounds are the ingredients for a certified algorithm
computing these funnels. The algorithm computes the Morse-Smale complex
by providing disjoint certified funnels for each stable and unstable
separatrix.
The proof of correctness and termination is presented in
Section~\ref{sec:termination}.

\subsection{Construction of fences around a separatrix}
\label{sec:ConstructionStrips}
Let $X_\vartheta$ be the vector field obtained by rotating the vector field
$X=\gradient{h}$ over an angle $\vartheta$, i.e.,
\begin{equation*}
  X_\vartheta = 
  \begin{pmatrix}
    \cos\vartheta & - \sin\vartheta \\
    \sin\vartheta &  \cos\vartheta \\
  \end{pmatrix}
  \begin{pmatrix}
    h_x \\ 
    h_y
  \end{pmatrix}.
\end{equation*}
We compute an \textit{isolating funnel} for the forward orbit of $\gradient{h}$
through a point $p$ by enclosing it between (approximations of) the forward orbits of
$X_{\vartheta}$ and $X_{-\vartheta}$ through~$p$. See Figure~\ref{fig:strip}.
\begin{figure}[h]
  \centering
  \psfrag{p}{$p$}
  \psfrag{Xp}{\tiny $\gradient{h}(p)$}
  \psfrag{Xp2}{}
  \psfrag{Xmin}{$X_{-\vartheta}(p)$}
  \psfrag{Xplus}{$X_{\vartheta}(p)$}
  \psfrag{Lmin}{$L_{-\vartheta}$}
  \psfrag{Lplus}{$L_{\vartheta}$}
  \includegraphics[width=0.22\textwidth]{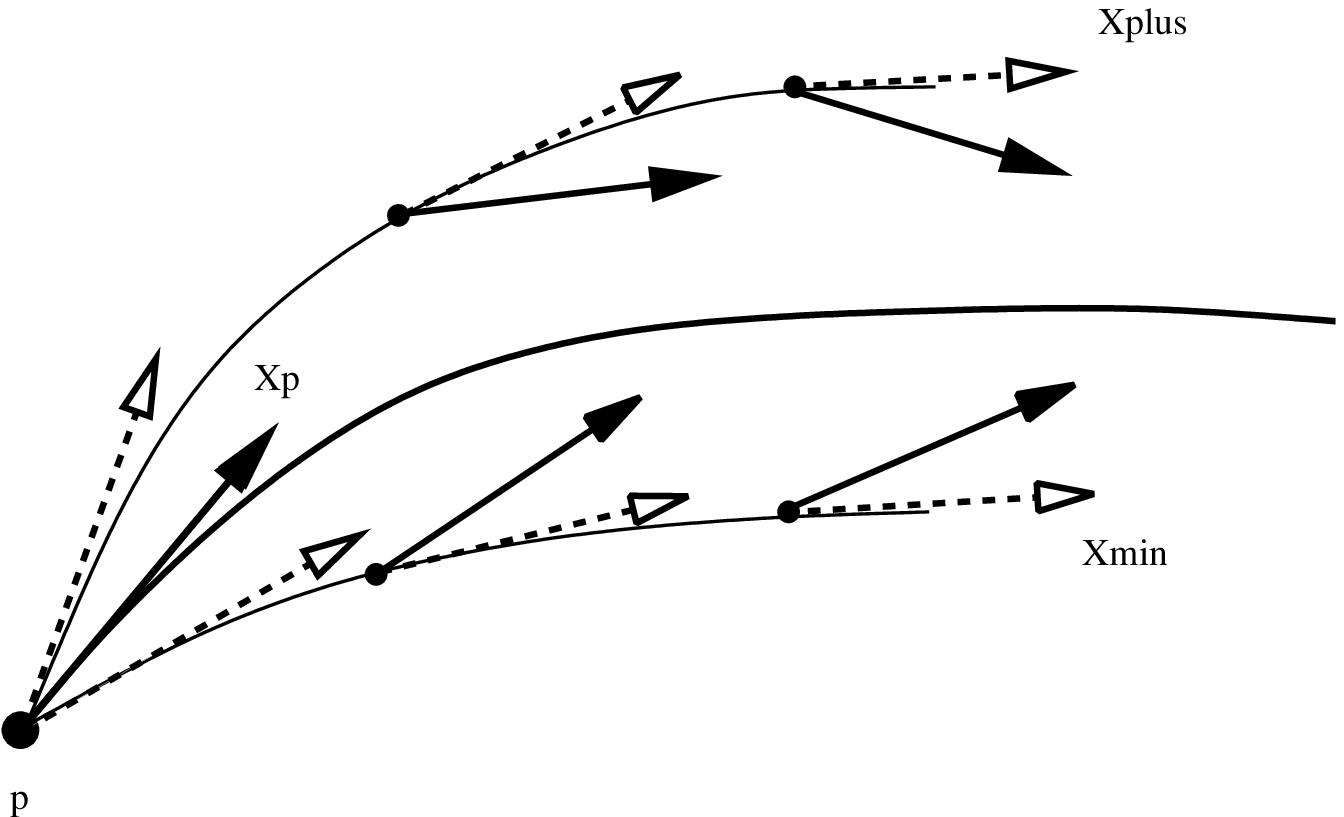}~~%
  \includegraphics[width=0.22\textwidth]{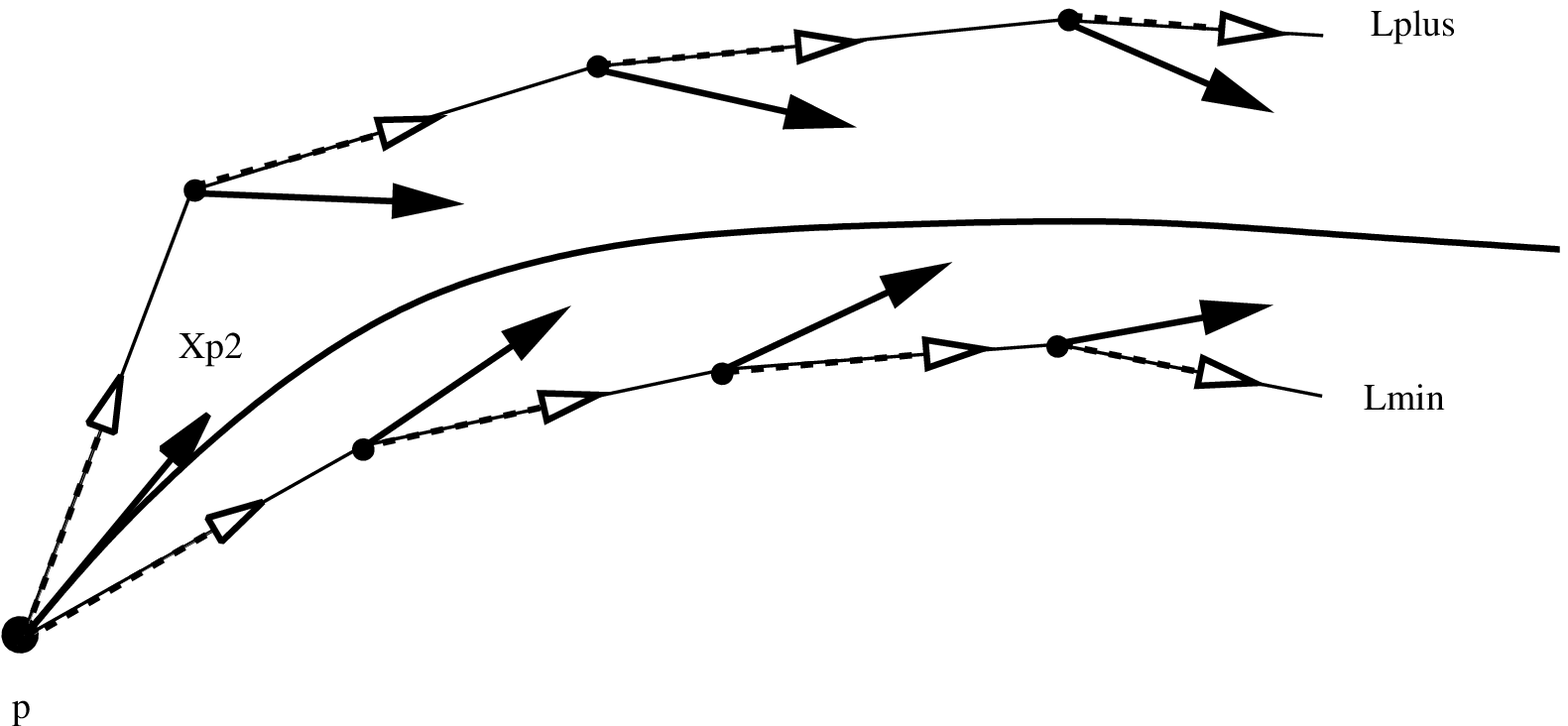}
  \caption{Orbits of the rotated vector fields $X_{\vartheta}$ and $X_{-\vartheta}$
    through a point $p$ enclose the forward orbit of $\gradient{h}$ through $p$. On
    the right polygonal lines approximating these orbits.}
  \label{fig:strip}
\end{figure}

\paragraph{Small angle variation}
We first determine some bounds on the angle variation of the gradient vector field
$\gradient{h}$ over $\Ds$.
We subdivide the region $\Ds$ into square boxes over which the angle variation of 
$\gradient{h}$ is at most $\vartheta$, where $\vartheta$ is to be determined later.
Let $w$ be the edge length of the boxes.
If $X=(f,g)$ is a vector field on $\R^2$, then the angle variation over a regular
curve $\Gamma$ is given by~\cite[Section 36.7]{a-ode-06}:
\begin{equation*}
  \int_{\Gamma}\frac{g\,df-f\,dg}{f^2+g^2}.
\end{equation*}
If $X = \gradient{h}$, this angle variation is equal to
\begin{equation*}
  \int_{\Gamma}\frac{(h_xh_{xy}-h_yh_{xx})\,dx + (h_xh_{yy}-h_yh_{xy})\,dy}{h_x^2+h_y^2}.
\end{equation*}
Let $C_0$ and $C_1$ be constants such that
\begin{equation}
  \label{eq:C0C1}
  \max_{\Ds}\left|\frac{h_xh_{xy}-h_yh_{xx}}{h_x^2+h_y^2}\right| \leq C_0
  \text{~~and~~}
  \max_{\Ds}\left|\frac{h_xh_{yy}-h_yh_{xy}}{h_x^2+h_y^2}\right| \leq C_1.
\end{equation}
Then the angle variation over a curve $\Gamma$ is less than
\begin{equation*}
  \int_{\Gamma} (C_0 \,dx + C_1 \,dy) \leq (C_0+C_1)\,\operatorname{length}(\Gamma).
\end{equation*}
This inequality provides an upper bound for the maximal angle variation over a square box:
\begin{lem}
  \label{lemma:angleBox}
  Let $C_0$ and $C_1$ satisfy (\ref{eq:C0C1}). Then the total angle variation over a
  square box in $\Ds$ with edge length $w$ does not exceed 
  $(C_0+C_1)w\sqrt{2}$.
\end{lem}

\noindent
The grid boxes have edge length $w$ such that the angle variation of $\gradient{h}$
over any box in $\Ds$ is less than $\vartheta$. 
\begin{lem}
  \label{lemma:angleVariation}
  Let $0 < \vartheta < \frac{\pi}{2}$, and let the grid boxes in $\Ds$ have width $w$
  satisfying
  \begin{equation}
    \label{eq:widthBox}
    w \leq \frac{\vartheta}{(C_0+C_1)\,\sqrt{2}}.
  \end{equation}
  Then the following properties hold.

  \medskip\noindent
  1.~The angle variation of $\gradient{h}$ over any gridbox is less than
  $\vartheta$.
  \\[1.2ex]
  2.~Let $p$ be a point on an edge of a gridbox,
  and let $q_{\vartheta}$ be the point on the boundary of the gridbox into
  which $X_{\vartheta}(p)$ is pointing, such that the line segment
  $pq_{\vartheta}$ has direction $X_{\vartheta}(p)$.
  The point $q_{-\vartheta}$ is defined similarly. See Figure~\ref{fig:orientation}.
  Then $\gradient{h}$ is pointing rightward along $pq_{\vartheta}$ and
  leftward along $pq_{-\vartheta}$. 
  \\[1.2ex]
  3.~The function $h$ is strictly increasing on the line segments from $p$ to 
  $q_{\vartheta}$ and from $p$ to $q_{-\vartheta}$.
\end{lem}
\begin{figure}[h]
  \centering
  \psfrag{p}{$p$}
  \psfrag{th}{$\vartheta$}
  \psfrag{qm}{$q_{-\vartheta}$}
  \psfrag{qp}{$q_{\vartheta}$}
  \psfrag{grh}{\small $\gradient{h}$}
  \includegraphics[width=0.48\textwidth]{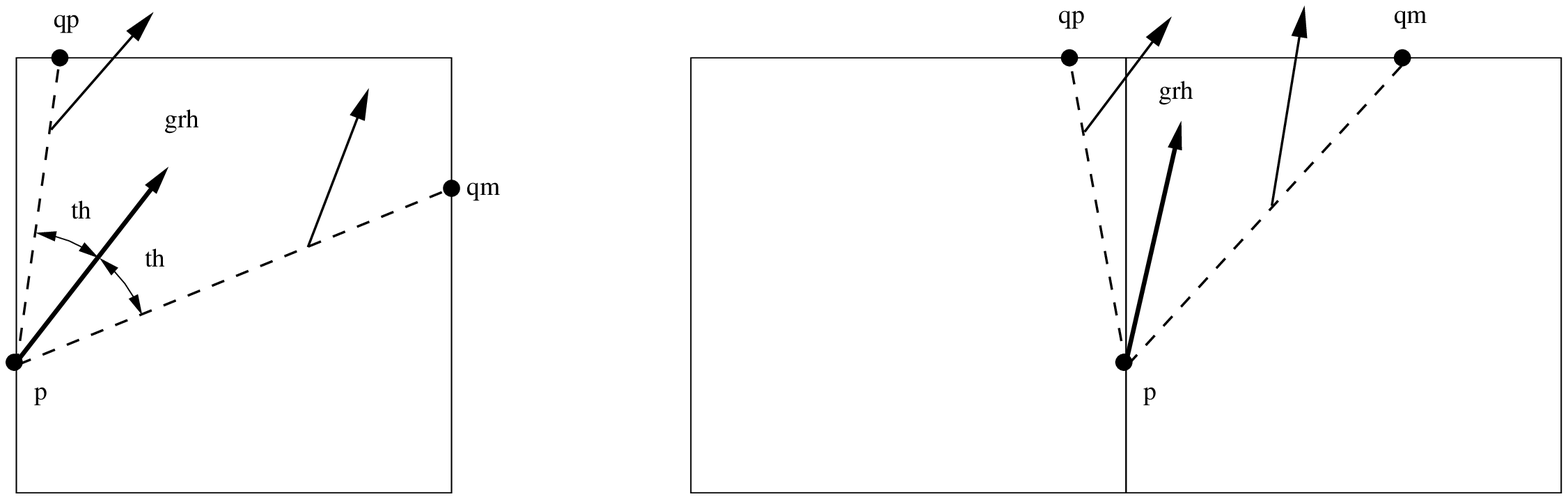}
  \caption{The orientation of $\gradient{h}$ with respect to
    $X_{\vartheta}(p)$ does not change over a grid box.}
  \label{fig:orientation}
\end{figure}
\begin{proof}
  The first claim follows from Lemma~\ref{lemma:angleBox}, using the
  fact that the diameter of a grid box is $w\sqrt{2}$.

  With regard to the second claim, 
  the small angle variation condition implies that the orientation of
  $\{\gradient{h}(q),X_{\vartheta}(p)\}$ does not change as $q$ ranges
  over $\I$. Since this orientation is positive for $q=p$, it is positive
  for all $q \in \I$. Similarly, the orientation of
  $\{\gradient{h}(q),X_{-\vartheta}(p)\}$ is negative for all $q \in \I$.
  Therefore, the second claim also holds.

  At a point $r$ of the line segment $pq_{\pm\vartheta}$ the directional derivative of
  $h$ in the direction of this line segment is $\ip{\gradient{h}(r)}{X_{\pm\vartheta}(p)}$, 
  which is positive since the angle between $\gradient{h}(r)$ and $X_{\pm\vartheta}(p)$, 
  is less than $\vartheta$, and $\vartheta < \frac{\pi}{2}$.
  This proves the third part.
\end{proof}

\paragraph{Fencing in the separatrices}
For each isolating unstable separatrix interval $\J$ on the boundary of a saddle box
we construct two polylines $L_{-\vartheta}(\J)$ and $L_{\vartheta}(\J)$ as follows.
The initial points of these polylines are the endpoints of $\J$,
$q_{-}$ and $q_{+}$, where $q_{-}$ comes before $q_{+}$
in the counterclockwise orientation of the boundary of the saddle box.
The polyline $L_{\vartheta}(\J)$ is uniquely defined by requiring that
its vertices $q_{+}=p_0,p_1,\ldots,p_n$ lie on grid edges, with the property that
\begin{enumerate}
\item 
  The line segment $p_{i-1}p_{i}$, $0 < i \leq n$, lies in a (closed) grid box of
  $\Ds$, and has direction $X_{\vartheta}(p_{i-1})$.
\item 
  $p_n$, the last vertex, lies on the boundary of $\Ds$.
\end{enumerate}
The polyline $L_{-\vartheta}(\J)$ is defined similarly, with the obvious changes:
its initial vertex is $q_{-}$, and each edge has direction equal to the value of the
vector field $X_{-\vartheta}$ at the initial point of this edge.
The polylines $L_{\pm\vartheta}(\J)$ are called \textit{fences} of the (unique)
unstable separatrix of $\gradient{h}$ intersecting $\J$. 

It is not hard to see that that a grid box contains at most two \textit{consecutive}
edges of each of these polylines, but it is not obvious a priori that each box
cannot contain more than two edges of each polyline in total. It follows from the
next result that the intersection of a grid box with any of these polylines is
connected, and, hence, that these polylines are finite.

% The following result is important for the proof of termination.
% \marginpar{Elaborate}
% First we should describe the algorithm, which causes a turn if the next edge
% of at least one of the polylines changes from being both quasihorizontal and
% quasivertical to being only quasivertical (or quasihorizontal).

% Observe that, for $\vartheta$ sufficiently small, the function $h$
% is increasing along $L_{\vartheta}$ and  $L_{-\vartheta}$.
% \marginpar{Elaborate}

The following results states that, when walking along the polylines
$L_{\vartheta}$ and  $L_{-\vartheta}$ in the direction of increasing
$h$-values, each grid box is passed at most once.
\begin{lem}
  \label{lemma:one-hit-per-box}
  Let $\I$ be a box such that the angle variation of $\gradient{h}$ over
  the surrounding box $N(\I)$ is at most $\vartheta$.
  Then the intersection of $L_{\vartheta}(\J)$ and $\I$ ($L_{-\vartheta}(\J)$ and
  $\I$) is either empty or 
  a connected polyline (consisting of one or two segments).
\end{lem}
\begin{figure}[h]
  \centering
  \psfrag{p}{$p$}
  \psfrag{p0}{$p_0$}
  \psfrag{p1}{$p_1$}
  \psfrag{q}{$q$}
  \psfrag{r}{$r$}
  \psfrag{s}{$s$}
  \psfrag{t}{$t$}
  \psfrag{op}{$\ol{p}$}
  \psfrag{oq}{$\ol{q}$}
  \psfrag{b}{\small$\beta$}
  \psfrag{a}{\small$\alpha$}
  \psfrag{grh}{$\gradient{h}$}
  \psfrag{Ltheta}{$L_\vartheta$}
  \includegraphics[width=0.48\textwidth]{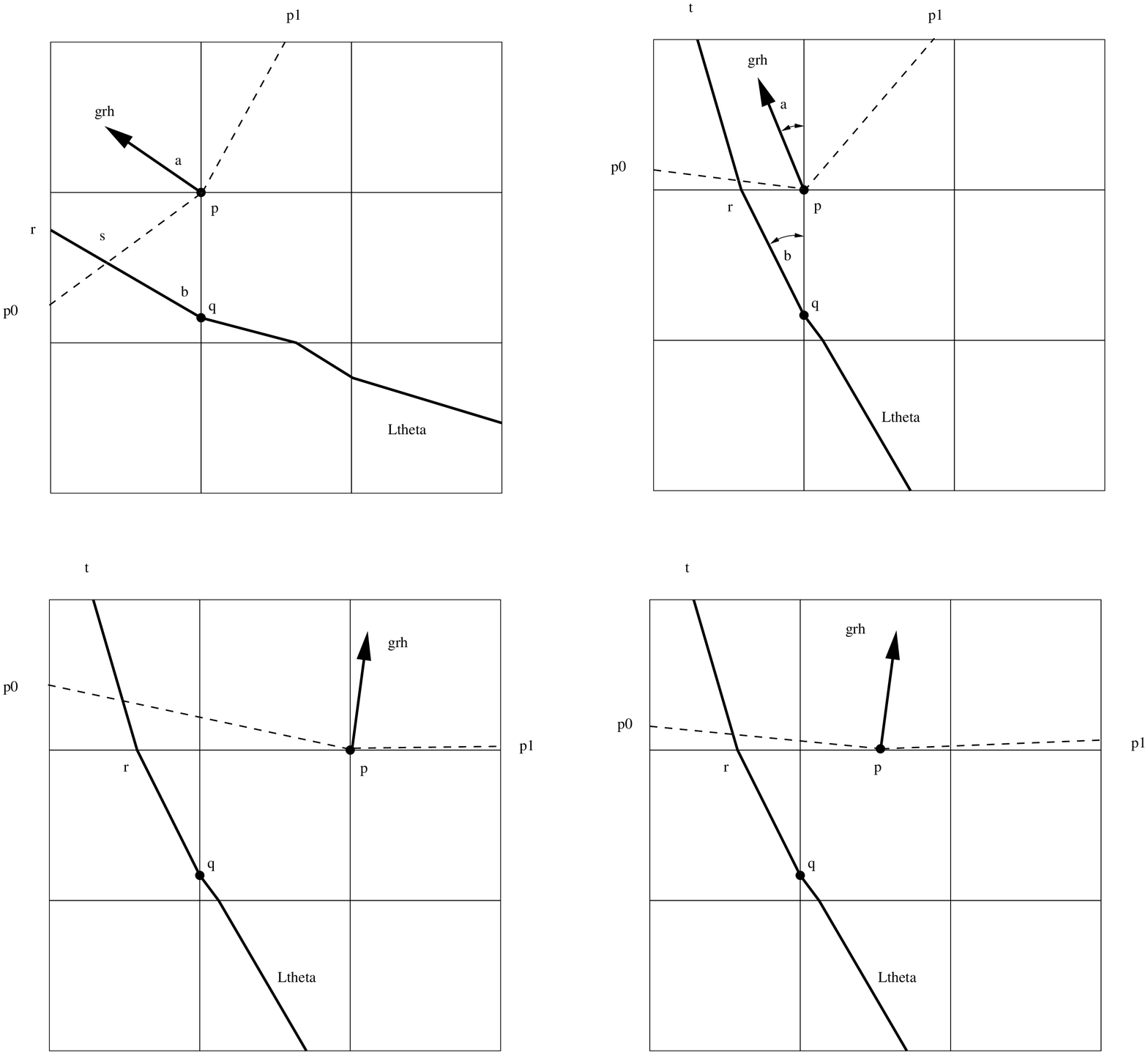}
  \caption{The value of $h$ at the point where polyline $L_{\vartheta}(\J)$
    leaves the surrounding box $N(\I)$ is greater than the maximum value
    of $h$ on $\I$.}
  \label{fig:polyline_box_intersection}
\end{figure}

\begin{proof}
  Let $q$ be a point at which $L_{\vartheta}(\J)$ leaves $\I$, i.e., 
  the segment of $L_{\vartheta}(\J)$ ending at $q$ lies inside $\I$ and the
  segment $qr$ beginning at $q$ lies outside $\I$.
  Let $p$ be a point on the boundary of $\I$ at which $h$ attains its
  maximum value $M$.
  
  \medskip\noindent
  \textit{Case 1: $p$ is a vertex of $\I$, incident to the edge of $\I$
    containing $q$.}
  \\
  See Figure~\ref{fig:polyline_box_intersection}, top row.
  Let $l$ be the line through the edge of $\I$ containing $q$, let
  $\alpha$ be the angle between $l$ and $\gradient{h}(p)$, and let
  $\beta$ be the angle between $l$ and the segment of $L_{\vartheta}(\J)$
  with initial point $q$. 
  The angles $\alpha$ and $\beta$ are both positive, since $p$ is a
  vertex of $\I$.
  The angle between $\gradient{h}(p)$ and $\gradient{h}(q)$
  is at most $\vartheta$, since the angle variation of $\gradient{h}$
  over $\I$ is less than $\vartheta$.
  Therefore, $|\alpha-\beta|\leq 2\vartheta$.

  Let $pp_0$ and $pp_1$, with $p_0$ and $p_1$ on the boundary of
  $N(\I)$, be the line segments that make an angle of
  $\frac{\pi}{2}-\vartheta$ with $\gradient{h}(p)$.
  Since the angle variation of $\gradient{h}$ over $N(\I)$ is at most
  $\vartheta$, the value of $h$ at any point of these line segments is
  at least $M$.
  We shall prove that the connected component of $L_{\vartheta}(\J)\cap
  N(\I)$ containing $q$ intersects one of the line segments $pp_0$ and
  $pp_1$.

  First assume $\alpha \geq \vartheta$. Then the line segment $pp_0$
  lies in the grid box $\J$ containing segment $qr$ of $L_{\vartheta}(\J)$.
  If $r$ lies on an edge of $\J$ incident to $p$, then $qr$ intersects
  $pp_0$.
  So assume $r$ lies on the edge of $\J$ contained in the boundary of
  $N(\I)$ (Figure~\ref{fig:polyline_box_intersection}, leftmost picture).
  Let $s$ be the point of intersection of the line through $pp_0$ and
  the line through $qr$.
  This point lies on the same side of $l$ as $p_0$ and $r$, 
  since $\angle p_0pq = \frac{\pi}{2} -\alpha+\vartheta < \frac{\pi}{2}$ and 
  $\angle pqr = \beta \leq \frac{\pi}{2}$.
  Furthermore, $\angle psq = \pi-(\frac{\pi}{2} -\alpha+\vartheta)-\beta
  \geq \frac{\pi}{2}-\vartheta > \frac{\pi}{4}$. Therefore, $x$ lies
  inside $N(\I)$, in other words, $qr$ intersects $pp_0$ also in this case.

  Now consider the case $\alpha < \vartheta$. 
  Then $p_0$ lies on the side of $N(\I)$ parallel to the line through
  $p$ and $q$.
  Furthermore, $\beta \leq \alpha + 2\vartheta < 3\vartheta < \frac{3\pi}{40}$,
  so $L_{\vartheta}(\J)$ `leaves' $N(\I)$ at a point $t$ on the side of
  $N(\I)$ perpendicular to the line through $p$ and $q$. 
  See Figure~\ref{fig:polyline_box_intersection}, rightmost picture.
  It follows that the part of $L_{\vartheta}(\J)$ between $q$ and $t$
  intersects $pp_0$. In particular, $h(t) > M$.

  \medskip\noindent
  \textit{Case 2: $p$ is not a vertex of $\I$, incident to the edge of $\I$
    containing $q$.}
  Then either $p$ is a vertex of $\I$, not incident to the edge of $\I$, 
  containing $q$, as in Figure~\ref{fig:polyline_box_intersection},
  bottom-left picture,
  or $p$ lies on the relative interior of an edge of $\I$, as in
  Figure~\ref{fig:polyline_box_intersection}, bottom-right picture.

  In this case $\gradient{h}(p)$ is nearly vertical, as are the edges of
  $L_{\vartheta}(\J)$. Similarly, the line segments $pp_0$ and $pp_1$ are
  nearly horizontal, so $L_{\vartheta}(\J)$ intersects $pp_0$. The details are
  similar to those of Case~1 of this proof.
\end{proof}

\noindent
If the endpoints of the fences $L_{\vartheta}(\J)$ and $L_{-\vartheta}(\J)$
lie on the same connected component of the boundary of $\Ds$, 
then these fences split $\Ds$ into two connected regions. 
See Figure~\ref{fig:3funnels}.
\begin{figure}[h]
  \psfrag{Lmin}{$L_{-\vartheta}$}
  \psfrag{Lplus}{$L_{\vartheta}$}
  \psfrag{sep}{$\gamma$}
  \centering
  \includegraphics[width=0.48\textwidth]{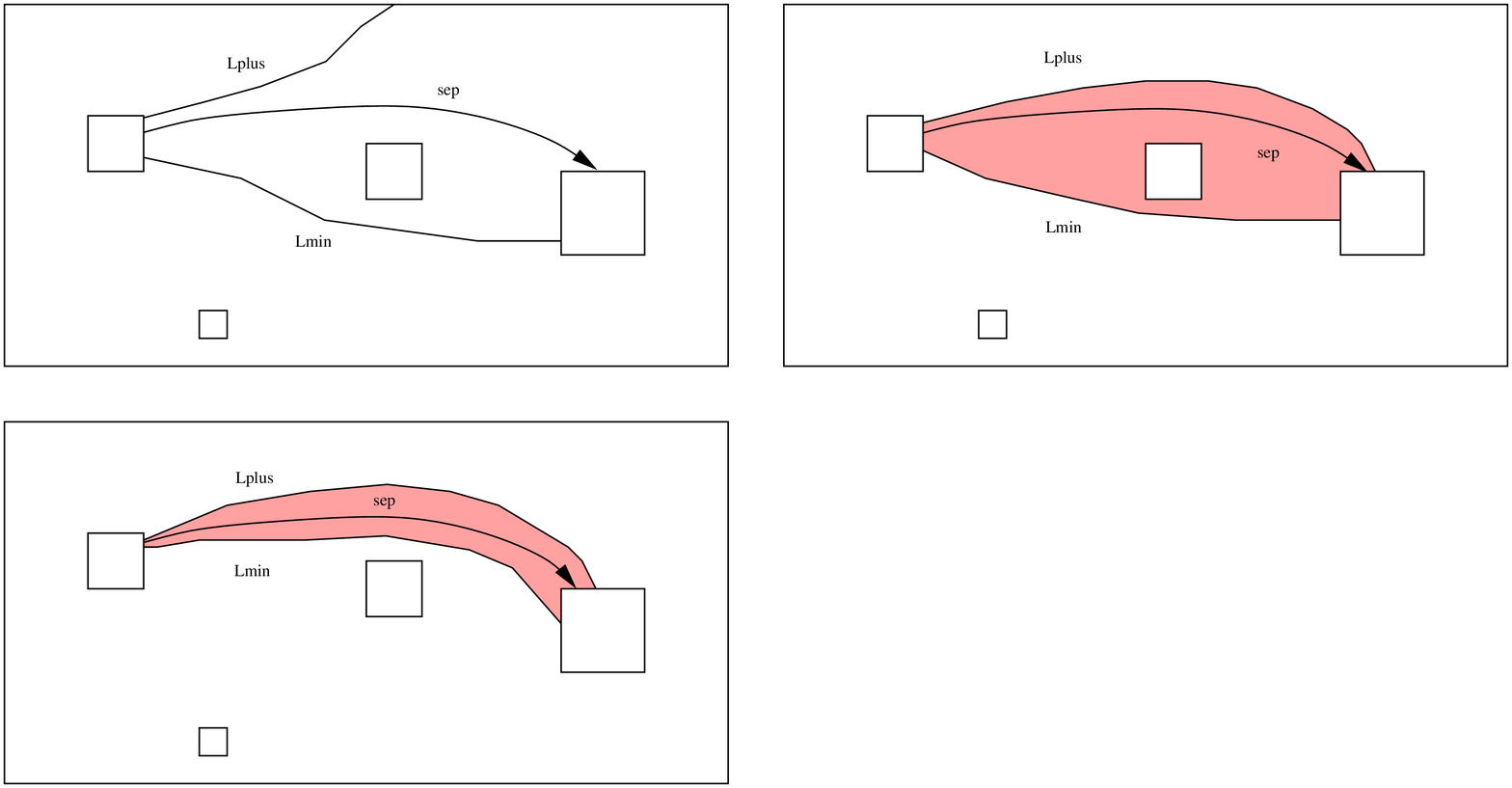}
  \caption{Fences around a separatrix $\gamma$. If the fences end in the same
    connected component of the boundary of $\Ds$, then they enclose a
    funnel (top right picture). If the funnel is simply connected, it isolates the
    separatrix from the source-, sink- and saddle-boxes (bottom picture). }
  \label{fig:3funnels}
\end{figure}
In this case, the region containing the separatrix interval $\J$ in its boundary is 
called the \textit{funnel} of $\J$ (with angle $\vartheta$)
denoted by $F_{\vartheta}(\J)$. Its boundary consists of $\J$, the two
fences $L_{\vartheta}(\J)$ and $L_{-\vartheta}(\J)$, and a curve $\J^\ast$ on the 
boundary of $\partial \Ds$ connecting the endpoints of these fences.
If the funnel is simply connected, it contains the part of the unstable 
separatrix through $\J$ lying inside $\Ds$, which enters the funnel through $\J$ and
leaves it through $\J^\ast$.
Note that $\J^\ast$ is a curve either on the outer boundary of $\Ds$
or on a sink box. 

Similarly, each \textit{stable} separatrix interval has two fences (for an angle
$\vartheta$). If the endpoints of these fences lie on the same connected component of
$\partial \Ds$, the enclosed region is again called a funnel for the stable separatrix
interval.
Our goal is to construct disjoint, simply connected funnels for the stable and unstable
separatrix intervals. If these funnels are disjoint, then they form,
together with the sink boxes, source boxes and saddle boxes,
a (fattened) Morse-Smale complex for $\gradient{h}$.

It is intuitively clear that a funnel $F_{\vartheta}(\J)$ is simply connected if
$\vartheta$, the length of the separatrix interval $\J$, and the edge length $w$ of
the grid boxes are sufficiently small. The next subsection presents computable upper
bounds on these quantities, guaranteeing that the endpoints of two fences of a
separatrix interval lie on the same boundary component. It is then easy to check
whether the enclosed funnel is simply connected. 

%%%%%%%%%%%%%%%%%%%%%%%%%%%%%%%%%%%%%%%%%%%%%%%%%%%%%%%%%%%%
%
%   Controlling the width of the funnel
% 
%%%%%%%%%%%%%%%%%%%%%%%%%%%%%%%%%%%%%%%%%%%%%%%%%%%%%%%%%%%%
\subsection{Controlling the width of the funnel}
\label{sec:widthFunnel}
If the width of a funnel is sufficiently small, in a sense to be made more precise,
it encloses a simply connected region in $\Ds$.
The width of a funnel is, roughly speaking, the number of grid boxes between its
bounding fences in the vertical direction, in regions where the fences are nearly horizontal, and 
in the horizontal direction, in regions where the fences are nearly vertical. 
To define the width of a funnel more precisely, we distinguish
quasihorizontal and quasivertical parts of a funnel, and show that the
width of a funnel does not increase substantially at transitions between
these quasihorizontal and quasivertical parts.

\paragraph{Quasihorizontal and quasivertical parts of a funnel}
A nonzero vector $v=(v_1,v_2)$ is called \textit{quasihorizontal} if 
$|v_2| \leq 2 |v_1|$, and \textit{quasivertical} if $|v_1| \leq 2 |v_2|$.
Note that each nonzero vector is quasihorizontal, quasivertical, or both.
Consider a subdivision of $\Ds$ into boxes of equal width, where non-disjoint boxes
share either an edge or a vertex. 
A \textit{horizontal $\varepsilon$-strip} is the union of a sequence of boxes where
successive boxes share a vertical edge, such that the horizontal edge of the
rectangle thus obtained has length at most $\varepsilon$. 
A \textit{vertical $\varepsilon$-strip} is defined similarly. An
\textit{$\varepsilon$-box} is a square box with edge length at most $\varepsilon$
which is the union of a number of boxes.
Two polygonal curves $L_{+}$ and $L_{-}$ form an $\varepsilon$-funnel if there is
a set ${\mathcal H}_\varepsilon$ of horizontal $\varepsilon$-strips, 
a set ${\mathcal V}_\varepsilon$ of vertical $\varepsilon$-strips, 
and a set ${\mathcal B}_\varepsilon$ of $\varepsilon$-boxes
such that the following holds:
\begin{enumerate}
\item 
  The vertices of $L_{-}$ and $L_{+}$ lie on the edges of the grid-boxes;
  $L_{-}$ intersects a grid box in at most one vertex or in at most one edge;
  the same holds for $L_{+}$;
\item 
  Both $L_{-}$ and $L_{+}$ lie in the union of the rectangles in 
  ${\mathcal H}_\varepsilon$, ${\mathcal V}_\varepsilon$ and ${\mathcal B}_\varepsilon$;
\item 
  An edge of $L_{\pm}$ contained in a horizontal ${\varepsilon}$-strip is
  quasi-vertical, and an edge contained in a vertical ${\varepsilon}$-strip is
  quasihorizontal. Moreover, neither $L_{+}$ nor $L_{-}$ intersect the vertical sides
  of a horizontal $\varepsilon$-strip, or the horizontal sides of a vertical 
  $\varepsilon$-strip.
  Each $\varepsilon$-strip and each $\varepsilon$-box is intersected by both polylines.
\item 
  Either $L_{-}$ or $L_{+}$ intersects an $\varepsilon$-box in exactly one of its
  edges, which is contained in a grid box at the corner of the $\varepsilon$-box.
  This single edge is either quasivertical or quasihorizontal (but not both).
  If this edge is quasihorizontal (quasivertical), all edges of the other polyline
  inside the $\varepsilon$-box are quasihorizontal (quasivertical) as well -- and
  possibly also quasivertical (quasihorizontal).
  The other polygonal curve intersects the same edges of the $\varepsilon$-box, each
  in exactly one point, and is disjoint from the other edges of the $\varepsilon$-box.
\end{enumerate}
See also Figure~\ref{fig:funnel}.
\begin{figure}[h]
  \psfrag{Lp}{$L_{+}$}
  \psfrag{Lm}{$L_{-}$}
  \centering
  \includegraphics[width=0.44\textwidth]{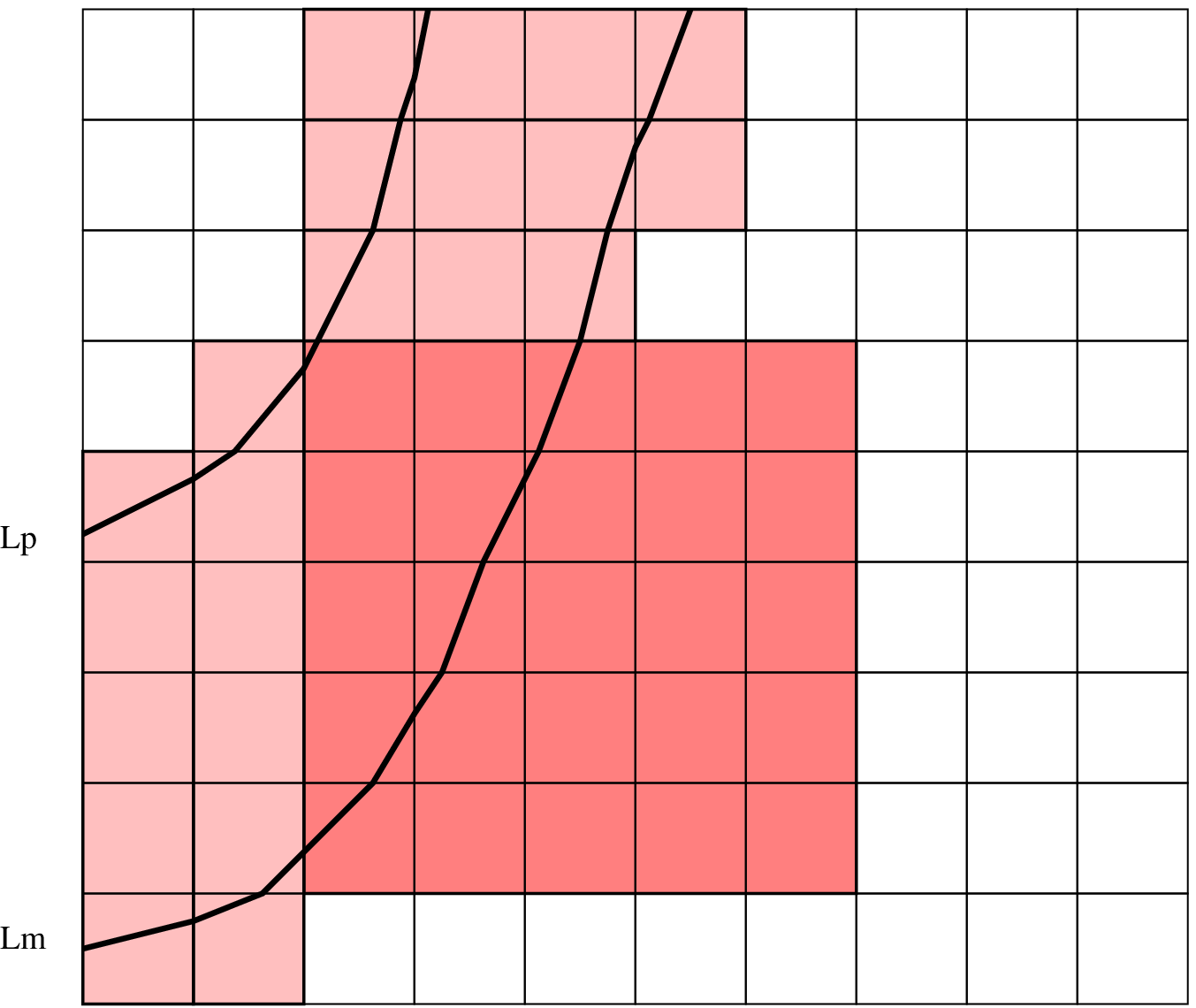}    
  \caption{A funnel formed by two polylines covered by two vertical
    $\varepsilon$-strips, one $\varepsilon$-box and three horizontal
    $\varepsilon$-strips. Here $\varepsilon$ is six times the width of a grid box.
    $L_{+}$ intersects the $\varepsilon$-box in a single edge, which is quasivertical
    but not quasihorizontal. All edges of $L_{-}$ inside the $\varepsilon$-box are
    quasivertical as well (and some of them are also quasihorizontal).
  }
  \label{fig:funnel}
\end{figure}

\noindent
We determine $\vartheta > 0$ later, but for now we assume that 
\begin{equation}
 \label{eq:theta1}
  \vartheta \leq \frac{\pi}{40}.
\end{equation}
We start with a simple observation.
\begin{lem}
  \label{lemma:vh-strips}
  Let $L$ be a polyline with quasihorizontal edges and with vertices on the edges of
  a grid with edge length $w$ satisfying (\ref{eq:widthBox}).
  If $L$ lies in a vertical strip of width $w$, where each of the vertical lines
  bounding the strip contains one of its endpoints,
  then $L$ intersects at most three grid boxes contained in this vertical strip.

  A similar property holds for a polyline with quasivertical edges intersecting a
  horizontal strip.
\end{lem}
\begin{proof}
  We only prove the first part, in which $L$ lies in a vertical strip and has
  quasihorizontal edges.
  The slope of the line segment connecting the endpoints of $L$ does not exceed the
  maximum slope of any of the edges of $L$, so this slope is at most $\arctan 2$.
  Hence the projection of this line segment on any of the vertical lines bounding the
  strip has length at most $2w$, so it intersects at most three boxes.
\end{proof}

\noindent
The next result shows that the width of a funnel does not grow
substantially at a transition between a quasihorizontal and a
quasivertical part.
We take $\varepsilon > 0$ such that the angle variation of
$\gradient{h}$ over a box with edge length $\varepsilon$ is at most $\tfrac{\pi}{20}$. 
Again, by Lemma~\ref{lemma:angleBox}, this is guaranteed by taking 
\begin{equation}
  \label{eq:epsilon1}
  \varepsilon \leq \frac{\pi}{20(C_0+C_1)\,\sqrt{2}}.
\end{equation}
\begin{lem}
  \label{lemma:epsilonBox}
  Let $\J$ be an $\varepsilon$-box intersected by both $L_{-\vartheta}$ and
  $L_{+\vartheta}$, with an edge $e$ which contains the initial vertex of both
  $L_{-\vartheta}\cap \J$ and $L_{+\vartheta}\cap \J$.
  Assume that at least one of the polylines has an edge which is either
  quasihorizontal or quasivertical, but not both. 
  Then both polylines intersect the boundary of $\J$ in exactly two points,
  and there is an edge $e'$ of $\J$, adjacent to $e$, containing the terminal
  vertices of both   $L_{-\vartheta}\cap \J$ and $L_{+\vartheta}\cap \J$.
  See Figure~\ref{fig:funnel}.
\end{lem}
\begin{proof}
  Assume that $L_{+\vartheta}$ has an edge $e_{+}$ which is quasivertical but not
  quasihorizontal.
  We first show that all edges of $L_{+\vartheta}$ are quasivertical (and possibly
  quasihorizontal).
  The angle between $e_{+}$ and the horizontal direction is at
  least $\arctan 2$, which is greater than $\tfrac{\pi}{4} + \tfrac{\pi}{10}$.
  Since the slope of $e_{+}$ is the slope of the vector field
  $X_{\vartheta}$ at the initial vertex of $e_{+}$, and the angle variation of
  $X_{\vartheta}$ over $\J$ is at most $\tfrac{\pi}{20}$, 
  the slope of $X_{\vartheta}$ at any point of $\J$ is at least
  $\tfrac{\pi}{4} + \tfrac{\pi}{20}$.
  Since the slope of an edge of $L_{+\vartheta}$ is the slope of $X_{\vartheta}$ at
  the initial vertex of this edge, we conclude that all edges of $L_{+\vartheta}$ are
  quasivertical.

  All edges of $L_{-\vartheta}$ are also quasivertical (and possibly quasihorizontal).
  To see this, observe that the slope of an edge of $L_{-\vartheta}$ is the slope of
  $X_{-\vartheta}$ at the initial vertex of this edge, and, hence, the slope of 
  $X_{\vartheta}$ at this initial vertex, minus $2\vartheta$.
  In other words, the slope of any edge of $L_{-\vartheta}$ is at least
  $\tfrac{\pi}{4} + \tfrac{\pi}{20} - 2\vartheta$.
  Since $\vartheta \leq \tfrac{\pi}{40}$, this slope is at least $\tfrac{\pi}{4}$.
  Therefore, all edges of $L_{-\vartheta}$ are quasivertical.

  The polylines $L_{+\vartheta}$ and $L_{-\vartheta}$ do not intersect the
  edge of $\J$ opposite $e$, since then at least one of the edges of these polylines
  would have a slope less than $\tfrac{\pi}{4}$.
  Let $e'$ be the edge containing the endpoint of $L_{\vartheta}(\J)\cap \J$.
  Then $e'$ is adjacent to $e$. 
  Given the bounds on the slope variation of the edges of the polylines, it is
  easy to see that 
  \\
  (i) the endpoint of $L_{\vartheta}(\J)$ is the only point of this polyline on $e'$;
  \\
  (ii) the endpoint $L_{-\vartheta}\cap\J$ also lies on $e'$, and this is the only
  point of this polyline on $e'$;
  \\
  (iii) none of the polylines intersects the edge opposite $e'$.
  \\
  This concludes the proof of Lemma~\ref{lemma:epsilonBox}.
\end{proof}

% \paragraph{Construction of the funnel.}
% To be completed.

\paragraph{Growth of the width of quasihorizontal and quasivertical funnel parts}
The width of the funnel may grow exponentially in the number of grid
boxes it is traversing. 
The next result gives an upper bound for the growth of this width.
Even though the bounds are conservative, they
provide the tools for the construction of certified funnels for all separatrices.

%\paragraph{Choice of constants.}
A gridbox is called \textit{quasihorizontal} (\textit{quasivertical}) if it contains
a point at which $\gradient{h}$ is quasihorizontal (quasivertical). Again, a
gridbox may be both quasihorizontal and quasivertical.

An integral curve of $\gradient{h}$ in a quasihorizontal gridbox $[x_0,x_1] \times
[y_0,y_1]$ is the graph of a function $x\mapsto y(x)$, where  $y(x)$ is a solution of
the differential equation 
\begin{align}
  y'(x) &= F(x,y(x)), \label{eq:unperturbed}\\
  y(x_0) &= y_0, \nonumber
\end{align}
where $F(x,y) = \dfrac{h_y(x,y)}{h_x(x,y)}$.
Here $x$ ranges over the full interval $[x_0,x_1]$ if $y_0 \leq
y(x) \leq y_1$. Otherwise, the range of $x$ is restricted to a suitable
maximal subinterval $[\xi_0,\xi_1]$, such that $(\xi_0,y(\xi_0))$ and
$(\xi_1,y(\xi_1))$ are points on the boundary of the gridbox.
Similarly, a trajectory of $X_{\vartheta}$ in a quasihorizontal gridbox
$[x_0,x_1] \times [y_0,y_1]$ is the graph of a function $x\mapsto y(x)$,
where  $y(x)$ is a solution of the differential equation
\begin{equation}
  \label{eq:rotated}
  \frac{dy}{dx} = F_{\vartheta}(x,y),
\end{equation}
with 
\begin{equation*}
  F_\vartheta(x,y) = 
  \dfrac
  {h_x(x,y)\,\sin\vartheta + h_y(x,y)\,\cos\vartheta}
  {h_x(x,y)\,\cos\vartheta - h_y(x,y)\,\sin\vartheta}.
\end{equation*}

Similarly, a trajectory of $X_{-\vartheta}$ is the graph of a function 
$y\mapsto x(y)$, where  $x(y)$ is a solution of the differential equation
\begin{equation*}
  \frac{dx}{dy} = G_{\vartheta}(x,y),
\end{equation*}
with 
\begin{equation*}
  G_\vartheta(x,y) = 
  \frac{1}{F_\vartheta(x,y)} = 
  \dfrac
  {h_x(x,y)\,\cos\vartheta - h_y(x,y)\,\sin\vartheta}
  {h_x(x,y)\,\sin\vartheta + h_y(x,y)\,\cos\vartheta}.
\end{equation*}
Here $y$ ranges over the full interval $[y_0,y_1]$ if $x_0 \leq
x(y) \leq x_1$, or a suitable maximal subinterval otherwise.

The union of all quasihorizontal gridboxes in $\Ds$ is denoted by
$\Ds_{\qh}$, and the union of all quasivertical gridboxes by $\Ds_{\qv}$.

Even though the width of a funnel may grow exponentially in the number of grid boxes
it traverses, this growth is controlled. To this end, we introduce several 
\textit{computable} constants that only depend on the function $h$ and (the size of)
its domain $\Ds$. 
%\marginpar{Singular boxes excluded! Constants may blow up?}
Let $A_{\qh}$, $A_{\qv}$, $B_{\qh}$, $B_{\qv}$,  $C_{\qh}$ and $C_{\qv}$
be positive constants such that
\begin{align*}
  \max_{(x,y)\in\Ds_{\qh}} |F(x,y)| \leq A_{\qh},
  \quad &
  \max_{(x,y)\in\Ds_{\qv}} |G(x,y)| \leq A_{\qv}, 
  \\[1.2ex]
  \max_{(x,y)\in\Ds_{\qh}} |\frac{\partial F}{\partial x}(x,y)| \leq B_{\qh},
  \quad &
  \max_{(x,y)\in\Ds_{\qv}} |\frac{\partial G}{\partial y}(x,y)| \leq B_{\qv},
  \\[1.2ex]
  \max_{(x,y)\in\Ds_{\qh}} |\frac{\partial F}{\partial y}(x,y)| \leq C_{\qh},
  \quad &
  \max_{(x,y)\in\Ds_{\qv}} |\frac{\partial G}{\partial x}(x,y)| \leq C_{\qv}.
\end{align*}
Note that
\begin{equation*}
  F_\vartheta(x,y) - F(x,y) 
  =
  \dfrac
  {(h_x(x,y)^2+h_y(x,y)^2)\,\sin\vartheta}
  {h_x(x,y)^2\,\cos\vartheta - h_x(x,y)\,h_y(x,y)\,\sin\vartheta}.
\end{equation*}
Let $M^{(1)}_{\qh}$ be a dyadic number such that
\begin{equation}
  \label{eq:L1}
  \max_{(x,y)\in\Ds_{\qh}}\left|\frac{h_y(x,y)}{h_x(x,y)}\right| \leq M^{(1)}_{\qh}.
\end{equation}
Take $\vartheta_{\qh} \in (0,\tfrac{1}{2}\pi)$ such that 
$\tan\vartheta_{\qh} \leq \dfrac{1}{2M^{(1)}_{\qh}}$.
Finally, let $M^{(2)}_{\qh}$ be a constant such that 
\begin{equation}
  \label{eq:M2}
  \max_{(x,y)\in\Ds_{\qh}}
  \left|\frac{h_x(x,y)^2+h_y(x,y)^2}{h_x(x,y)^2\cos\vartheta_{\qh}}\right| \leq M^{(2)}_{\qh}. 
\end{equation}
Taking $M_{\qh} = \dfrac{M^{(2)}_{\qh}}{2 M^{(1)}_{\qh}}$, we have,
for $|\vartheta| \leq \vartheta_{\qh}$:
\begin{equation}
  \label{eq:diffFh}
  \max_{(x,y)\in\Ds_{\qh}}| F_\vartheta(x,y) - F(x,y) | \leq M_{\qh} \sin\vartheta.
\end{equation}
Similarly, there are (computable) constants $M_{\qv}$ and $\vartheta_{\qv}$ such
that
\begin{equation}
  \label{eq:diffFv}
  \max_{(x,y)\in\Ds_{\qv}}| G_\vartheta(x,y) - G(x,y) | \leq M_{\qv} \sin\vartheta,
\end{equation}
for $|\vartheta| \leq \vartheta_{\qv}$.
Finally, let the constants $c_0$, $c_1$ and $\vartheta_0$ be defined by
\begin{align}
  c_0 &= 2\max(C_{\qh}+A_{\qh}B_{\qh},C_{\qv}+A_{\qv}B_{\qv})   \label{eq:c0}\\
  c_1 &= \max(\frac{1}{2M_{\qh}K_{\qh}},\frac{1}{2M_{\qv}K_{\qv}})  \label{eq:c1}\\
  \vartheta_0 &= \max(\vartheta_{\qh},\vartheta_{\qv}). \label{eq:theta0}
\end{align}
%
%
%---------------
%
%
The next result provides an upper bound for the growth of the funnel
width along a quasihorizontal part of its bounding polylines. We assume
that the funnel runs from left to right, so its initial points are on the line
with smallest $x$-coordinate. If the funnel runs from right to left, a
similar result is obtained.
%A similar result holds for quasivertical parts.
\begin{lem}
  \label{lemma:fenceWidth}
  Let $y_{\vartheta,w}, y_{-\vartheta,w}:[a,b] \rightarrow [c,d]$ be the
  piecewise linear functions the graphs of which are quasihorizontal
  parts of the polylines $L_{\vartheta}$ and $L_{-\vartheta}$ for a grid
  with edge length $w$, respectively. Let $\Delta$ be an upper bound
  for the distance of the initial points of these polylines, i.e.,
  \begin{equation*}
    |y_{\vartheta,w}(a) - y_{-\vartheta,w}(a)| \leq \Delta.
  \end{equation*}
  Then
%   \begin{equation*}
%     w \leq c_0\varepsilon \text{~~and~~} \sin\vartheta \leq c_1\varepsilon
%   \end{equation*}
  the width of the fence, bounded by $L_{\vartheta}$ and
  $L_{-\vartheta}$, is bounded:
  \begin{equation*}
    |y_{\vartheta,w}(x) - y_{-\vartheta,w}(x)| 
    \leq 
    \Delta\,e^{C_{\qh}(x-a)}
    +
    (c_0w+c_1\sin\vartheta)\,\dfrac{e^{C_{\qh}(x-a)}-1}{C_{\qh}}.
  \end{equation*}
\end{lem}
\begin{proof}
Let $y_{\pm\vartheta}(x)$ be the exact solution of the rotated system with
initial condition $y_{\pm\vartheta}(a)$. In particular, 
$|y_{\vartheta}(a) - y_{-\vartheta}(a)| \leq \Delta$.
Then (\ref{eq:diffFh}) implies
\begin{equation*}
  \bigl|\frac{dy_{\pm\vartheta}}{dx}(x) - F(x,y_{\pm\vartheta}(x))\bigr| 
  = 
  \bigl| F_{\pm\vartheta}(x,y_{\pm\vartheta}(x) - F(x,y_{\pm\vartheta}(x))\bigr| 
  \leq 
  M_{\qh}\sin\vartheta.
\end{equation*}
Therefore, according to the Fundamental
Inequality~\cite[Theorem~4.4.1]{jw-dedsa-91}---See also 
\ref{sec:math}---we have
\begin{equation}
  \label{eq:diffSolutions2}
  |y_\vartheta(x) - y_{-\vartheta}(x)| 
  \leq 
  \Delta\,e^{C_{\qh}(x-a)}
  +
  \dfrac{2M_{\qh}\sin\vartheta}{C_{\qh}}\,(e^{C_{\qh}(x-a)}-1).
\end{equation}
%\marginpar{$[a,b]?$}
The interval $[a,b]$ is subdivided into a finite number of subintervals of length
at most $w$, where the endpoints correspond to the $x$-coordinates of
the breakpoints of the fences $L_{\vartheta}$ and $L_{-\vartheta}$.
Let $y_{\vartheta,w}$ be the Euler approximation to the ordinary differential
equation (\ref{eq:rotated}). Its graph is (a quasihorizontal) part of
the fence $L_{\vartheta}$.
Theorem~4.5.2 in~\cite{jw-dedsa-91}---See also 
\ref{sec:math}---gives the following explicit bound for the
error in Euler's method:
\begin{equation}
  \label{eq:diffSolutions3}
  |y_{\vartheta,w}(x) - y_{\vartheta}(x)|
  \leq 
  w\,\frac{B_{\qh}+A_{\qh}C_{\qh}}{C_{\qh}}\,(e^{C_{\qh}(x-a)}-1).
\end{equation}
% where $A_{\qh}$, $B_{\qh}$ and $C_{\qh}$ are positive constants such that 
% \begin{equation}
%   \label{eq:BC}
%   \max_{(x,y)\in\B} |F(x,y)| \leq A_{\qh},~~
%   \max_{(x,y)\in\B}|\frac{\partial F_{\vartheta}}{\partial x}(x,y)| \leq B_{\qh},
%   \text{~~and~~}
%   \max_{(x,y)\in\B}|\frac{\partial F_{\vartheta}}{\partial y}(x,y)| \leq C_{\qh}.
% \end{equation}
We get a similar inequality for the Euler approximation $y_{-\vartheta,w}$
of $y_{-\vartheta}$.
Combining (\ref{eq:diffSolutions2}) and (\ref{eq:diffSolutions3}), and
using (\ref{eq:c0}) and (\ref{eq:c1}), yields
\begin{align*}
%  \label{eq:diffEuler}
  |y_{\vartheta,w}(x) - y_{-\vartheta,w}(x)|
  & \leq 
  \Delta\,e^{C_{\qh}(x-a)}
  + \\
  & 2(w(B_{\qh}+A_{\qh}C_{\qh}) + M_{\qh}
  \sin\vartheta)\,\frac{e^{C_{\qh}(x-a)}-1}{C_{\qh}} \\
  & = 
  \Delta\,e^{C_{\qh}(x-a)}
  +
  (c_0w+c_1\sin\vartheta)\,\dfrac{e^{C_{\qh}(x-a)}-1}{C_{\qh}}.
\end{align*}
% Taking $K_{\qh}=\dfrac{e^{C_{\qh}(b-a)}-1}{C_{\qh}}$ we see that
% \begin{equation}
%   \label{eq:diffEuler}
%   |y_{\vartheta,w}(b) - y_{-\vartheta,w}(b)|
%   \leq 
%   2K_{\qh}\,(w\,(C_{\qh}+B_{\qh}A_{\qh}) + M_{\qh} \sin\vartheta),
% \end{equation}
% which can be made at most $\varepsilon$ by taking
% \begin{equation}
%   \label{eq:constantChoice}
%   w \leq \frac{1}{4K_{\qh}(C_{\qh}+B_{\qh}A_{\qh})}\,\varepsilon
%   \text{~~and~~}
%   \sin\vartheta \leq \frac{1}{4M_{\qh}K_{\qh}}\,\varepsilon.
% \end{equation}
\end{proof}

\noindent
A similar result holds for quasivertical trajectories. 
Next we need to control the increase of the funnel width upon 
transition from a quasihorizontal to a quasivertical part its bounding
polylines (or from a quasivertical to a quasihorizontal part).

\paragraph{Transitions: bounded increase of funnel width}
Transition from a quasihorizontal to a quasivertical, or from a quasivertical to a
quasihorizontal part of the funnel takes place at an $\varepsilon$-box. 
If the width of the funnel at the `entry' of the box is less than the width $w$ of a
grid box, then the width may increase, but it will not be greater than $2w$ 
at the exit. This is made more precise by the following result.
\begin{lem}
  \label{lemma:errorOnTurn}
  Let $\J$ be a $\varepsilon$-box as in Lemma~\ref{lemma:epsilonBox}, where,
  moreover, the initial points $p$ and $q$ of $L_{\vartheta}\cap\J$ and
  $L_{-\vartheta}\cap\J$,
  respectively, are on the boundary of the gridboxes containing the vertices of edge
  $e$ of $\J$. 
  If the distance between $p$ and $q$ is at least $w$, then the distance between the
  terminal points  $\ol{p}$ and $\ol{q}$ of $L_{\vartheta}\cap\J$ and
  $L_{-\vartheta}\cap\J$, respectively, is less than the distance of $p$ and $q$.
  If the distance between $p$ and $q$ is less than $w$, then the distance of 
  $\ol{p}$ and $\ol{q}$ is at most $2 w$.
\end{lem}
\begin{proof}
  Assume that the first edge of $L_{\vartheta}$ is quasivertical, but not
  quasihorizontal. 
  Edge $e$ of $\J$ is then vertical.
  Assume that this polyline consists of a single edge, namely the line segment
  $p\ol{p}$.
  \begin{figure}[h]
    \psfrag{p}{$p$}
    \psfrag{q}{$q$}
    \psfrag{op}{$\ol{p}$}
    \psfrag{oq}{$\ol{q}$}
    \psfrag{bp}{$\beta_{+}$}
    \psfrag{bm}{$\beta_{-}$}
    \psfrag{bpl}{}
    \psfrag{bmi}{}
    \psfrag{d}{$d$}
    \psfrag{od}{$\ol{d}$}
    \psfrag{a}{$a$}
    \centering
    \includegraphics[width=0.48\textwidth]{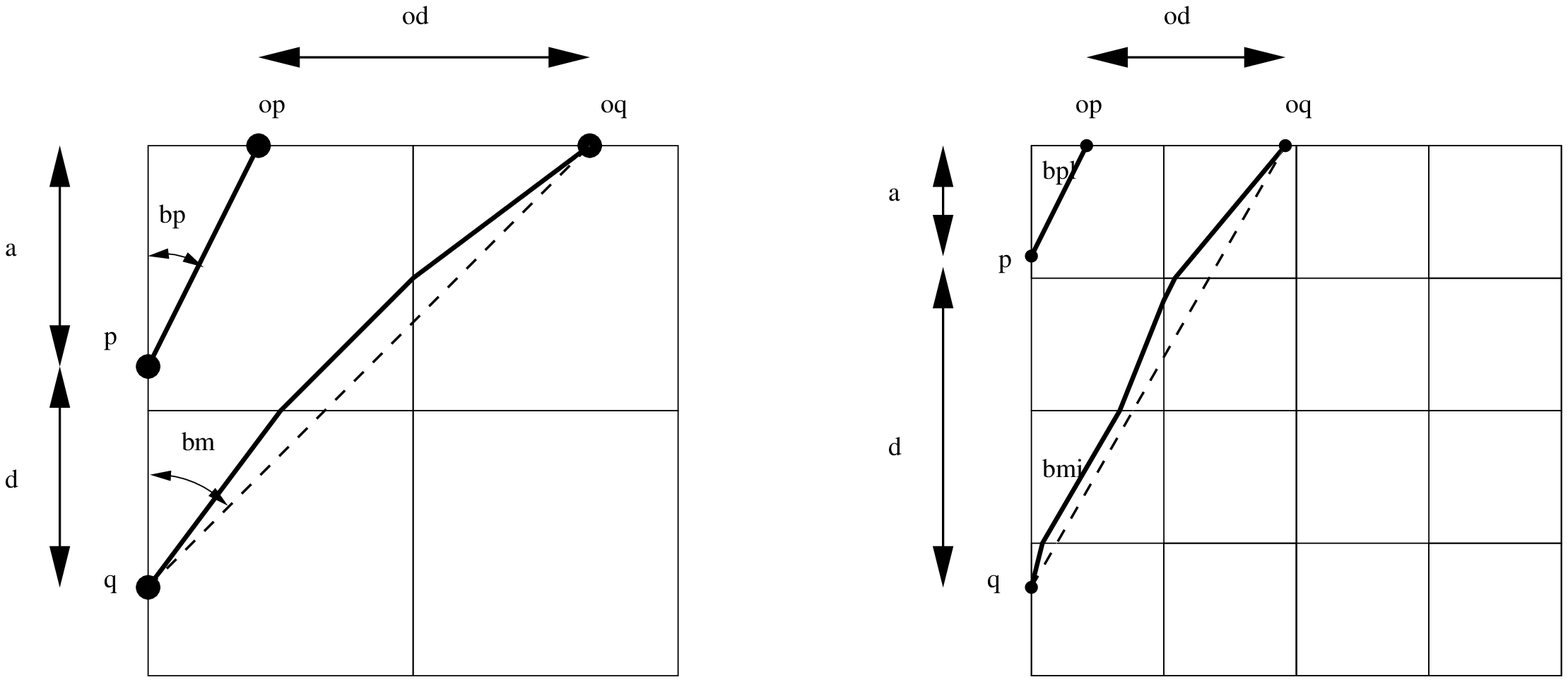}
    \caption{The distance between the two polylines upon entry and exit
      of a box. 
      Left: If the distance $d$ between the initial points $p$ and $q$
      of the polylines is less than the edge-length $w$ of a grid box, then the
      distance $\ol{d}$ between the terminal points $\ol{p}$ and $\ol{q}$ is less
      than $2w$.
      Right: 
%       If the distance between the initial points of the polylines is at
%       least the edge-length of a grid box, then 
      Otherwise, the distance $\ol{d}$ between the terminal points is less than $d$.
    }
    \label{fig:turn_box}
  \end{figure}
  Let $\beta_{+}$ be the angle between $p\ol{p}$ and edge $e$, then
  $\arctan{\frac{1}{2}} \leq \beta_+ \leq \vartheta + \arctan{\frac{1}{2}}$.
  Let $\beta_{-}$ be the angle between the line segment $q\ol{q}$ and edge
  $e$, then $\beta_{-}$ is inbetween the smallest and largest slope of any edge of 
  $L_{-\vartheta}$. Since the angle variation of $X$ over $\J$ is less than
  $\tfrac{\pi}{20}$, the angle $\beta_{-}$ is greater than $\beta_+-\tfrac{\pi}{20}$.
  Let $a$ be the distance of $p$ to the nearest vertex of $e$, then $a \leq w$.
  If $d \geq w$, the distance $\ol{d}$ between $\ol{p}$ and $\ol{q}$ satisfies
  \begin{align*}
    \ol{d} 
    &= 
    (d+a)\,\tan\beta_{-} - a \tan\beta_{+} \\
    &\leq
    d\,\tan(\beta_{+}+\tfrac{\pi}{20}) + 
    a\,(\tan(\beta_{+}+\tfrac{\pi}{20}) - \tan\beta_{+})\\
    &<
    \tfrac{3}{4}d + \tfrac{1}{4}a\\
    &\leq 
    d,
  \end{align*}
  since $a \leq w \leq d$.
  Here we used $\tan\beta_+\leq\tfrac{1}{2}$ to get
  \begin{equation*}
    \tan(\beta_{+}+\tfrac{\pi}{20})
    = 
    \frac
    {\tan\beta_++\tan\frac{\pi}{20}}
    {1-\tan\beta_+\tan\frac{\pi}{20}}
    \leq 
    \frac
    {\tfrac{1}{2}+\tan\frac{\pi}{20}}
    {1-\tfrac{1}{2}\tan\frac{\pi}{20}}
    \leq \tfrac{3}{4}.
  \end{equation*}
  Since $\arctan\frac{1}{2} - \frac{\pi}{40} \leq \arctan\frac{1}{2} - \vartheta \leq
  \beta_{+} \leq \arctan\frac{1}{2}$, a short computation shows that
  $\tan(\beta_{+}+\tfrac{1}{20}\pi)-\tan\beta_{+} < \tfrac{1}{4}$.

%  \marginpar{Elaborate!}
  If $d < w$, then $q$ lies in the same gridbox as $p$, or in a gridbox adjacent 
  to it. Then it is easy to see that $\ol{p}$ lies in the
  same grid box as $p$, and $\ol{q}$ also lies in this box, or in a box adjacent to it.
  Therefore, $\ol{d} \leq 2 w$ in this case.

  \noindent
  If $L_{-\vartheta}$ consists of a single edge, then the argument is similar.
%  \marginpar{Elaborate}
\end{proof}

Lemmas~\ref{lemma:fenceWidth} and \ref{lemma:errorOnTurn} provide the
following result on the upper bound on the funnel width of a separatrix
with $M$ transitions between quasihorizontal and quasivertical parts.
\begin{cor}
  \label{cor:totalWidth}
  Let $T$ be the (computable) edge length of a bounding square of the domain $\D$ of the function $h$), 
  and let $M$ be the total number
  of quasihorizontal and quasivertical parts of the polylines bounding
  a separatrix funnel.
%\marginpar{Define sep funnel, width, etc!}
  Let $C = \max(C_{\qh},C_{\qv})$ and let $D = \min(C_{\qh},C_{\qv})$.
  Then the width of the funnel does not exceed
  % \begin{equation*}
  %   (c_1\sin\vartheta + c_2w)\,\frac{e^{CMT}-1}{D},
  % \end{equation*}
  \begin{equation*}
    (c_1\vartheta + c_2w)\,\frac{e^{CMT}}{D},
  \end{equation*}
 provided $\vartheta\leq\vartheta_0$ where $c_2 = 2 + \dfrac{c_0}{D}$, with $c_0$ and $c_1$ given by
  (\ref{eq:c0}) and (\ref{eq:c1}), respectively. 

  \medskip\noindent
  In particular, this width is at most $\varepsilon$ if 
  % \begin{equation}
  %   \label{eq:errorEpsilon}
  %   c_1\sin\vartheta + c_2w \leq \frac{D}{e^{AMT}-1}\,\varepsilon.
  % \end{equation}
 \begin{equation}
    \label{eq:errorEpsilon}
    c_1\vartheta + c_2w \leq \frac{D}{e^{CMT}}\,\varepsilon.
   \end{equation}
\end{cor}
\begin{proof}
  Let $\Ds \subset [a,b] \times [c,d]$, then $T \leq \max(b-a,d-c)$.
  There are $M-1$ transitions from quasihorizontal to quasivertical
  parts of the funnel, each occurring at an $\varepsilon$-box.
  Let $\Delta_0$ be the width of the initial separatrix interval,
  and let $\Delta_1, \ldots, \Delta_{M-1}$ be the width of the funnel at the
  entry of the corresponding boxes, in other words, $\Delta_k$ is the
  width at the end of the $k$-th part of the funnel.
%  \marginpar{Explain sep interval}
%   First assume that $\Delta_n > w$, for $n = 1, \ldots, M$.
%   Therefore, the width at the initial point of the $n$-th
%   quasihorizontal or quasivertical part of the funnel is at most
%   $\Delta_{n-1}$, cf Lemma~\ref{lemma:errorOnTurn}.
%   Assume that the $i$-th part of the funnel is quasihorizontal, then
%   Lemma~\ref{lemma:fenceWidth} implies that the width of this part 
%   at a point with horizontal coordinate $x$ is at most
%   \begin{equation*}
%     \Delta_{n-1}\,e^{C_{\qh}(x-a)}
%     +
%     (c_0w+c_1\sin\vartheta)\,\dfrac{e^{C_{\qh}(x-a)}-1}{C_{\qh}},
%   \end{equation*}
%   which is at most $\Delta_{n-1}e^{AT} + (c_0w+c_1\sin\vartheta)\frac{e^{AT}-1}{B}$.
%   In particular,
%   $$
%   \Delta_n \leq \Delta_{n-1}e^{AT} + (c_0w+c_1\sin\vartheta)\frac{e^{AT}-1}{B}.
%   $$
  Using induction, we will prove that, for $k = 1, \ldots, M$:
  \begin{equation}
    \label{eq:IH}
    \Delta_k \leq 2w\,e^{kCT} + \frac{c_0w+c_1\sin\vartheta}{D}\,(e^{kCT}-1).
  \end{equation}
  %First we prove (\ref{eq:IH}) for $k=M$. % is the claim of the lemma, this is all
  % we have to prove.
  
  So assume (\ref{eq:IH}) holds for $k = n-1$.
  If $\Delta_{n-1} > w$, the initial width of the $n$-th part of the
  funnel does not exceed $\Delta_{n-1}$, cf Lemma~\ref{lemma:errorOnTurn}.
  Assume that the $n$-th part of the funnel is quasihorizontal, then
  Lemma~\ref{lemma:fenceWidth} implies that the width of this part 
  at a point with horizontal coordinate $x$ is at most
  \begin{equation*}
    \Delta_{n-1}\,e^{C_{\qh}(x-a)}
    +
    (c_0w+c_1\sin\vartheta)\,\dfrac{e^{C_{\qh}(x-a)}-1}{C_{\qh}},
  \end{equation*}
  so in particular, since $D \leq C_{\qh} \leq C$ and $0 \leq x-a \leq T$:
  \begin{equation*}
    \Delta_n 
    \leq 
    \Delta_{n-1}\,e^{CT}
    +
    \frac{c_0w+c_1\sin\vartheta}{D}\,(e^{CT}-1).
  \end{equation*}
  Therefore, (\ref{eq:IH}) holds for $k = n$.
  If $\Delta_{n-1} \leq w$, then the initial width of the $n$-th part of
  the funnel is at most $2w$, cf Lemma~\ref{lemma:errorOnTurn}.
  Therefore, Lemma~\ref{lemma:fenceWidth} implies that the width of this part 
  at a point with horizontal coordinate $x$ is at most
  \begin{equation*}
    2w\,e^{C_{\qh}(x-a)}
    +
    (c_0w+c_1\sin\vartheta)\,\dfrac{e^{C_{\qh}(x-a)}-1}{C_{\qh}},
  \end{equation*}  
  so in particular
  \begin{align*}
    \Delta_n 
    & \leq 
    2w\,e^{CT}
    +
    \frac{c_0w+c_1\sin\vartheta}{D}\,(e^{CT}-1) \\
    & \leq
    2w\,e^{nCT}
    +
    \frac{c_0w+c_1\sin\vartheta}{D}\,(e^{nCT}-1),
  \end{align*}
Therefore, for $n=M$, we have
\begin{align*}
    \Delta_M 
    & \leq 
    (c_1\sin\vartheta+c_2w)\frac{e^{CMT}-1}{D} \\
    & \leq
    (c_1\vartheta + c_2w)\,\frac{e^{CMT}}{D},
  \end{align*}

  which proves the corollary.
\end{proof}

\begin{remark}
The computable constants $\vartheta_0,\,c_1,\,c_2,\,C\,$ and $D$ depend only on $\Ds$ and $h$.  
\end{remark}

In the next section, we assemble the bits and pieces into a certified
algorithm for the
construction of the MS-complex, and show how the upper
bounds on the funnel width are used to prove that this algorithm terminates.

%----------------------

\subsection{Construction of the MS-complex}
\label{sec:termination}

\paragraph{The Algorithm}
The construction of the MS-complex of $h$ is a rather
straightforward application of the preceding results. It uses a
parameter $M$, the (a priori unknown) number of transitions (at
$\varepsilon$-boxes) between quasihorizontal and quasivertical parts of a funnel.
Let $T$ be the edge length of a bounding square of the domain $\D$ of $h$.
Then the algorithm performs the following steps. 
\begin{description}
\item{Step 1.}
Construct certified isolating boxes $\B_1',\ldots,\B_m'$ for 
% saddles, sinks and sources 
the singularities of $\gradient{h}$ (cf~Section~\ref{sec:eqSolving}).

\item{Step 2.}
Let $\Ds$ be the closure of $\D \setminus (B_1' \cup \cdots \cup
\B_m')$. Compute the constants $\vartheta_0$, $c_1$, $c_2$, $C$ and $D$,
which depend only on $h$ and $\Ds$.
Set $\varepsilon$ to the minimum of the width of the source-, sink- and
saddleboxes.
%, and the upper bound imposed by the conditions in Lemma~\ref{lemma:vh-strips}.

\item{Step 3.} 
Let $\vartheta$ and $w$ be such that 
$w\leq \frac{\vartheta}{2C_0\,\sqrt{2}}$, 
$\vartheta \leq \min(\frac{\pi}{40},\vartheta_0)$, 
and $c_1\vartheta + c_2w \leq \varepsilon \,D\, e^{-CMT}$
(cf Corollary~\ref{cor:totalWidth}).
Subdivide $\Ds$ until all gridboxes have maximum width $w$.
For each saddle box, compute four separatrix intervals on its
boundary, of width at most $w$.

\item{Step 4.}  
For each stable and unstable separatrix interval do the following.
Start the computation of a funnel for a separatrix by setting $M$ to a
small number $M_0$ (say 4).
Compute the fences $L_{-\vartheta}$ and $L_{\vartheta}$, keeping
track of the width of the enclosed funnel under construction and of the
number $m$ of transitions between quasihorizontal and quasivertical
parts of this funnel.

If \textit{the width of the funnel exceeds $\varepsilon$ or the number of
transitions $m$ exceeds $M$}, then abort the computation of the current
funnel, discard all funnels constructed so far, set $M$ to twice its
current value and goto Step 3.

If \textit{the funnel intersects an already constructed funnel, or a source- or
sinkbox on which it does not terminate} (i.e., if only one of its fences
intersects this box), then set $\varepsilon$ to
half its current value, discard all funnels constructed so far, and goto
Step 3.

If \textit{the funnel intersects a saddlebox $\B_i'$}, then 
decrease the size of $\B_i'$ by a factor of two via
\textit{subdivision}, discard all funnels constructed so far, 
set $\varepsilon$ to half its current value, 
and goto Step 2. 
(Note that $\Ds$ gets larger, so the constants in Step 2 have to be recomputed.)

\textit{Otherwise}, the fences end on the same component of the boundary of
$\partial\Ds$.
% (assuming that the height and width of $\Ds$ are at least $\varepsilon$). 
The enclosed funnel is simply connected, and does not
intersect any of the funnels constructed so far.
Add the funnel to the output, and reset $M$ to $M_0$ (and repeat
until all separatrices have been processed).
\end{description}

\myPara{Termination}
Since the gradient field $\gradient{h}$ is a $2D$ Morse-Smale system, its
separatrices are disjoint. Their intersections with $\Ds$ are compact, 
and have positive distance (although this distance is not known a
priori). 
In the main loop of the algorithm, the maximal funnel width
$\varepsilon$ is bisected if funnels intersect, and saddleboxes
intersected by the funnel are subdivided, so after a finite number
of iterations of the main loop its value is less than half the minimum
distance between any pair of distinct separatrices, and funnels stay
clear from saddleboxes (apart from the one containing the $\alpha$- or
$\omega$-limit of the enclosed separatrix).

Separatrices that intersect $\partial\D$ do so transversally, cf
Remark~\ref{rem:transversality}.
Therefore, after a finite number of subdivision steps, both fences around such
separatrices will intersect $\partial\D$ transversally.
Hence, eventually all funnels become disjoint, at which point the
algorithm terminates after returning a topologically correct MS-complex
for $\gradient{h}$.

%%% Local Variables: 
%%% mode: latex
%%% TeX-master: "saddle_analysis"
%%% End: 

%% file: Tex/6-experiments.tex
\section{Implementation and experimental results}
\label{sec:implementation}
% \vspace*{-2ex}
The algorithm has been implemented using the \texttt{Boost} library~\cite{boost} for IA.
All experiments have been performed on a 3GHz Intel Pentium 4 machine
under Linux with 1 GB RAM using the \texttt{g++} compiler, version
3.3.5. 
Figures \ref{fig:output1}\subref{fig:squares}-\subref{fig:circular} 
and \ref{fig:output2}\subref{fig:quartic}-\subref{fig:linear7} 
depict the output of our algorithm, for several Morse-Smale functions.
In our implementation the parameter $\vartheta$, used in the construction
of separatrix-funnels, is $\frac{\pi}{30}$, which is larger than the
theoretical bound given by Corollary~\ref{cor:totalWidth}. The algorithm
halves this angle several times, depending on the input, until the 
funnels are simply connected, mutually disjoint, and connect a
saddle-box to a source-box (for stable separatrices) or sink-box (for
unstable separatrices), in which case a topologically correct MS-complex
has been computed.
%
% different values of the initial angle $\theta$ of rotation of the vector
% field.
\begin{figure}[htb]
  \begin{center}
    \subfigure[Contour plot (left) and MS-Complex (right) of 
    $h(x,y) = -10\,x^2+x^4+10\,y^2- y^4+x + x y^2$, on the box
    ${[-4,3.5]\times[-4,3.5]}$.
    CPU-time: 11 seconds.]{
      \def\width{0.21\textwidth}
      \includegraphics[width=\width]{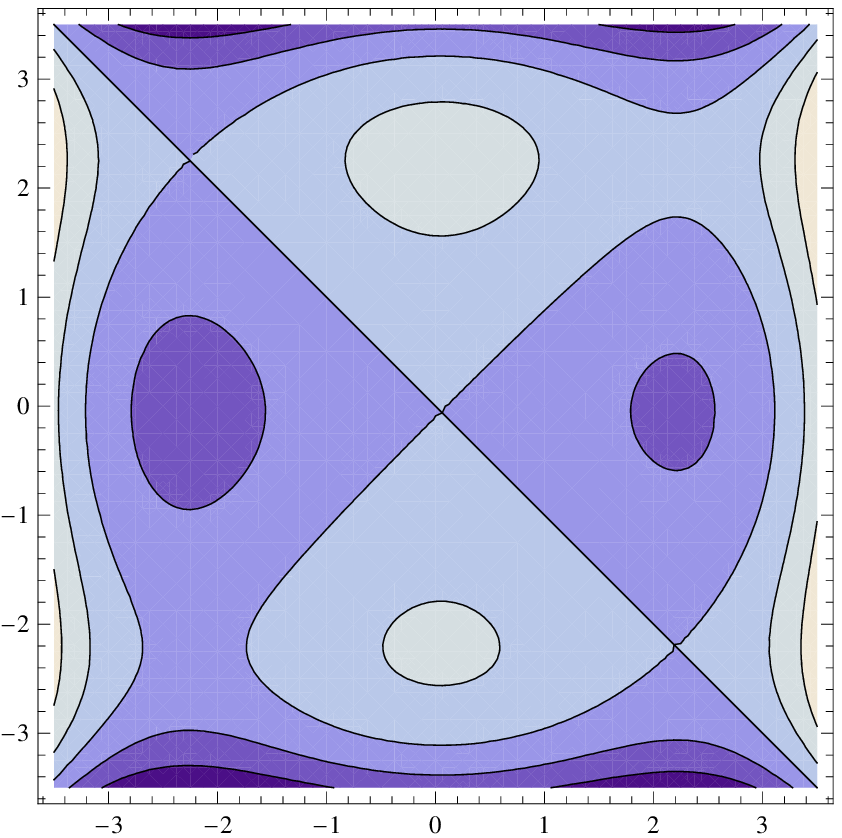}~~
      \includegraphics[width=\width]{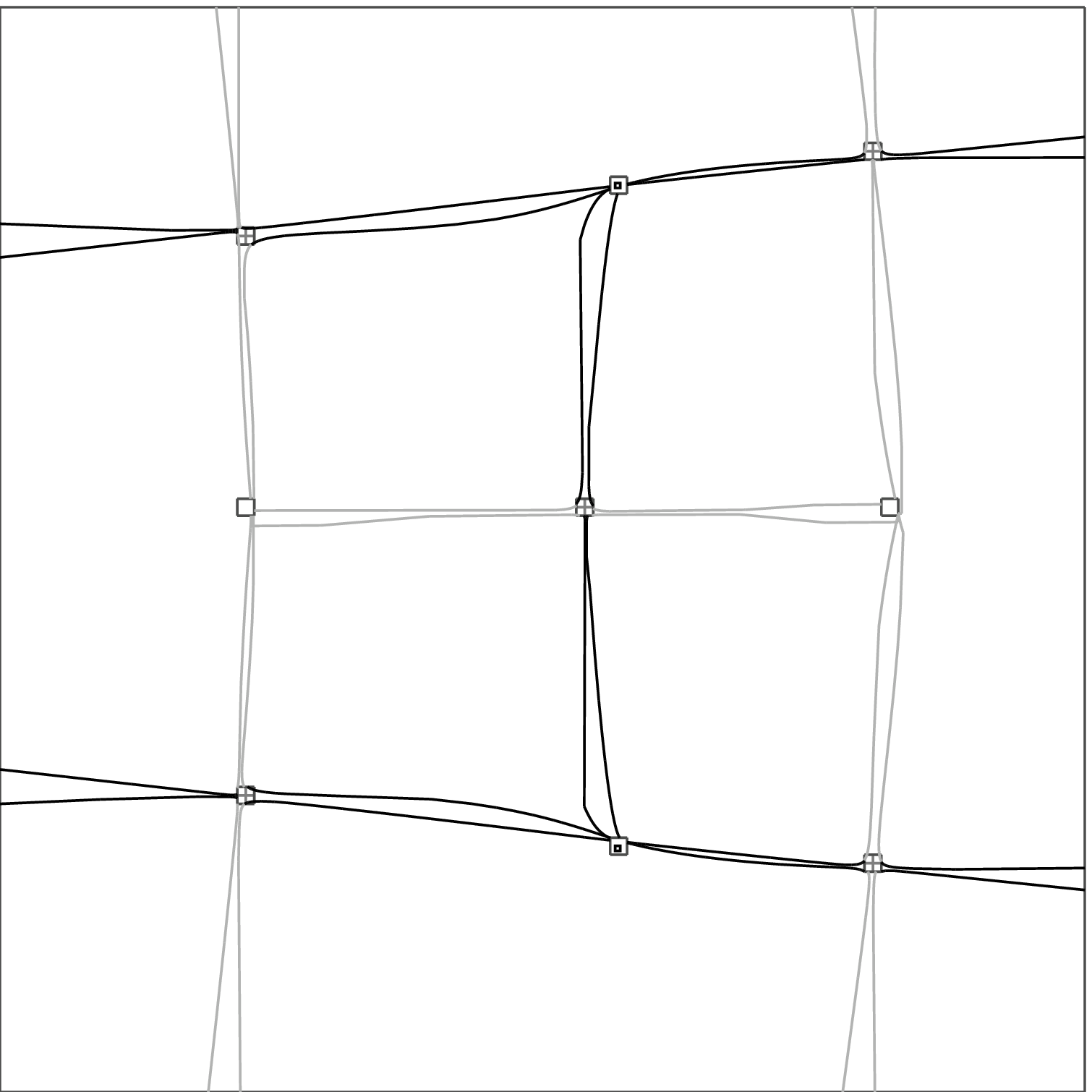}
      \label{fig:quartic}
    }~~~~~\\
  \subfigure[Contour plot (left) and MS-complex (right) of a product of
  seven linear functions, on the box ${[-7,7]\times[-7,7]}$. CPU-time:
  11.5 minutes.]{
    \def\width{0.21\textwidth}
    \includegraphics[width=\width]{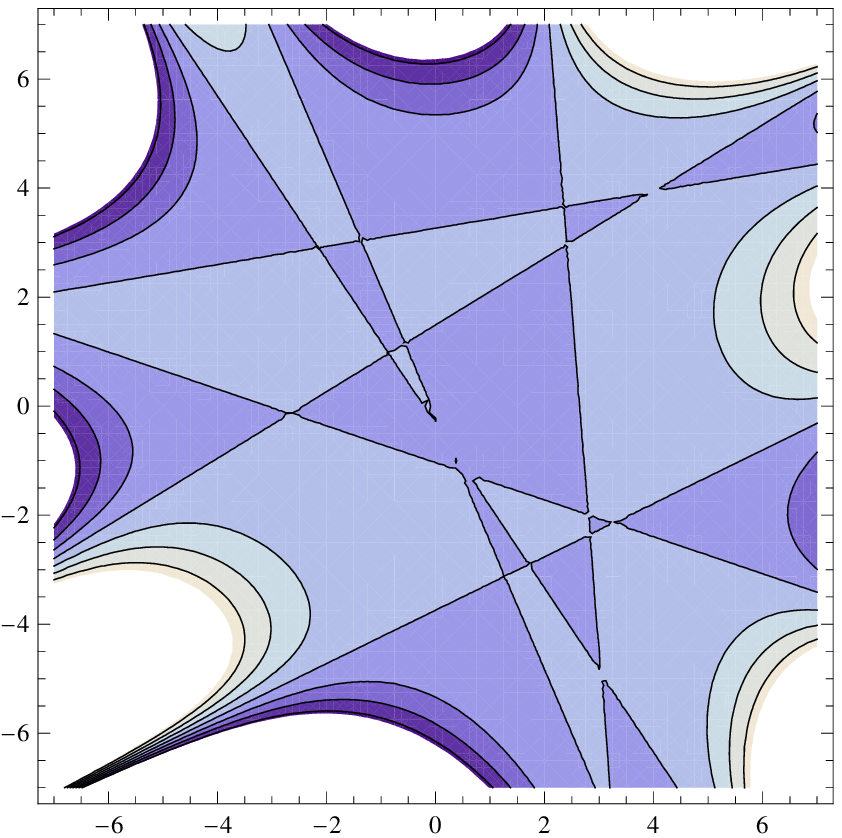}~~
    \includegraphics[width=\width]{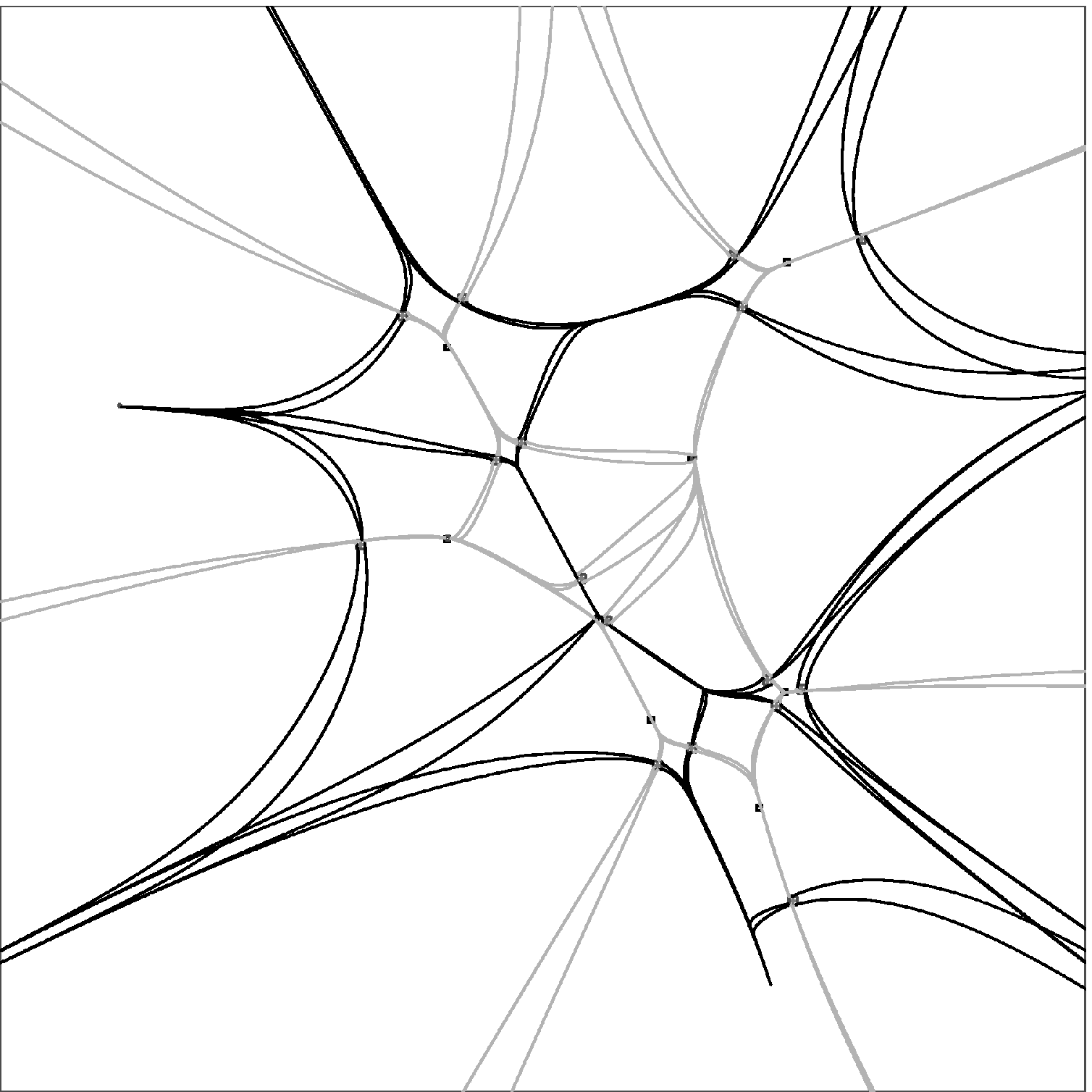}
    \label{fig:linear7}
  }

    \vspace*{-2ex}
    \caption{Contour plots of MS-functions and their Morse-Smale complexes.}
    \label{fig:output2}
  \end{center}
  \vspace*{-2ex}
\end{figure}
% In Figures~\ref{fig:output2}\subref{fig:quartic} and
% \ref{fig:output2}\subref{fig:circular} \ldots
%
% In each of figures~\ref{fig:mscomplex-sincos}-\ref{fig:MS-complex-test5}, the first figure shows the
% contour plot of a function. Then follows its corresponding computed
% MS-complexes for different parameter values.
Each of the funnels with deep black boundaries contains an unstable
separatrix, whereas a funnel with light black boundaries
contains a stable separatrix. 
% 
% \begin{figure}[htb]
%   \begin{center}
%     \def\width{0.2\textwidth}
%     \excludedFig{  
%     \includegraphics[width=\width]{postscript/test5.eps}~~
%     \includegraphics[width=\width]{postscript/test5eps0pt17angpiby30.png.eps}
%   }
%     %\vspace*{-2ex}
%     \caption{{MS-Complex of the function constructed by
%       multiplying $7$ linear functions.
%      (i) Contour plot; (ii) MS-complex with $\theta=\frac{\pi}{30}$, inside box $[-7,7]\times[-7,7]$.}}
%     \label{fig:MS-complex-test5}
%   \end{center}
% \end{figure}
% If the parameter $\theta$ is made smaller, the left and right boundaries
% of each funnel come closer and gradually converge to the actual
% separatrix inside the funnel. Thus we obtain certified and
% geometrically close funnels of separatrices. Table~\ref{tab:cpu-times}
% shows some timing results (CPU-time) corresponding to the figures.
The CPU-time for computing a MS-system increases with the number of
critical points and the complexity of the vector field, as indicated in
the captions of the figures.
% \begin{table}[h!]
% \centering
% \caption{CPU-time corresponding to figures~\ref{fig:mscomplex-sincos}-\ref{fig:MS-complex-test5}\label{tab:cpu-times}}
% \begin{tabular}{|c|c|c|c|c|}
% \hline
% parameter-values    & Figure~\ref{fig:mscomplex-sincos} & Figure~\ref{fig:mscomplex-algb} & Figure~\ref{fig:mscomplex-test3} & Figure~\ref{fig:MS-complex-test5}\\ 
% \hline 
% % $\epsilon_c=0.5$, $\theta=\frac{\pi}{10}$ & 8 sec.  & 5 sec.  & 0.15 sec & --       \\
% % $\epsilon_c=0.2$, $\theta=\frac{\pi}{30}$ & 20 sec  & 11 sec. & 0.5 sec & --       \\
% % $\epsilon_c=0.17$, $\theta=\frac{\pi}{30}$& --      & --      &  -- &    11.5min    \\ 
% $\theta=\frac{\pi}{10}$ & 8 sec.  & 5 sec.  & 0.15 sec & --       \\[1.0ex]
% $\theta=\frac{\pi}{30}$ & 20 sec  & 11 sec. & 0.5 sec & 11.5min    \\[1.2ex]
% \hline
% \end{tabular}
% \end{table}

%% file: Tex/7-conclusion.tex
%%%%%%%%%%%%%%%%%%%%%%%%%%%%%%%%%%%%%%%%%%%%%%%%%%%%%%%%%%%%%%%%%%%%%%%
\section{Conclusion}
The outcome of our research is two-fold. Firstly, we compute the
topologically correct MS-complex of a Morse-Smale system. The
correct saddle-sink or saddle-source connectivity can also be represented as a
graph, which is of special interest from different application
point of view. On the other hand, depending on a user-specified width of funnel one
can compute a geometrically close approximation of the MS-complex. We
give the proof of convergence of our algorithms. Although the complexity of the
given algorithm depends on the input function and the complexity
of the interval arithmetic library used in the algorithm. As we discussed some of the separatrices
inside  a bounding box $B$ may have discontinuous components. The algorithm we propose
here is able to compute only the part of the separatrices which are connected
to the corresponding saddle. Therefore one open question is how to
compute all the components of separatrices inside a bounding box.

%%%%%%%%%%%%%%%%%%%%%%%%%%%%%%%%%%%%%%%%%%%%%%%%%%%%%%%%%%%%%%%%%%%%%%%

%% file: Tex/appendixMath.tex
%\chapter{Appendix: Chapter 3}
%\newpage
\appendix
\section{Mathematical results used in the text}
\label{sec:math}
\paragraph{Error in Euler's method.}
Error bounds for approximate solutions of ordinary differential equation
play a crucial role in the construction of certified funnels for
separatrices. We quote the relevant parts of the book~\cite{jw-dedsa-91}.
\begin{fundIneq}
  Consider the differential equation 
  $$\dfrac{dy}{dx}=F(x,y)$$
  on a box $\B = [a,b]\times[c,d]$, and let $C$ be a constant such that
  \begin{equation*}
    \max_{(x,y)\in\B} |\frac{\partial F}{\partial y}(x,y)| \leq C.
  \end{equation*}
  If $y_1(x)$ and $y_2(x)$ are two approximate piecewise differentiable solutions 
  satisfying 
  \begin{align*}
    |y_1'(x) - F(x,y_1(x))| & \leq \varepsilon_1, \\
    |y_2'(x) - F(x,y_2(x))| & \leq \varepsilon_2
  \end{align*}
  for all $x \in [a,b]$ at which $y_1(x)$ and $y_2(x)$ are differentiable, and if,
  for some $x_0 \in [a,b]$
  \begin{equation*}
    |y_1(x_0) - y_2(x_0)| \leq \delta,
  \end{equation*}
  then, for all $x \in [a,b]$
  \begin{equation*}
    |y_1(x) - y_2(x)| \leq \delta e^{C|x-x_0|} +
    \varepsilon\,\frac{e^{C|x-x_0|} -1}{C},
  \end{equation*}
  where $\varepsilon = \varepsilon_1 + \varepsilon_2$.
\end{fundIneq}

\medskip\noindent
The well-known Euler method for constructing approximate solutions to ordinary
differential equations is also useful for the construction of certified strips.
It proceeds as follows.
For a given initial position $(x_0,y_0)$, define the sequence of points $(x_n,y_n)$
by 
\begin{align*}
  x_n &= x_{n-1} + \eta = x_0 + n\eta\\
  y_n &= y_{n-1} + \eta\,F(x_{n-1},y_{n-1}),
\end{align*}
as long as $(x_n,y_n) \in \B$.
Then the \textit{Euler approximate solution} $y_\eta(x)$ through $(x_0,y_0)$ 
with step $\eta$ is the piecewise linear function the graph of which joins the points
$(x_n,y_n)$, so
\begin{equation*}
  y_h(x) = y_n + (x-x_n)\,F(x_n,y_n) \quad \text{for $x \in [x_n,x_{n+1}]$}.
\end{equation*}
The following result states that the Euler approximate solution converges to the
actual solution as the step tends to zero, and gives a bound for the error.
\begin{Euler}
  Consider the differential equation 
  $$\dfrac{dy}{dx}=F(x,y)$$
  on a box $\B = [a,b]\times[c,d]$, where $F$ is a $C^2$-function on $\B$.
  Let the constants $A$, $B$ and $C$ satisfy
  \begin{equation*}
    \max_{(x,y)\in\B} |F(x,y)| \leq A,
    \quad
    \max_{(x,y)\in\B} |\frac{\partial F}{\partial x}(x,y)| \leq B,
    \quad
    \max_{(x,y)\in\B} |\frac{\partial F}{\partial y}(x,y)| \leq C.
  \end{equation*}
  The deviation of the Euler approximate solution $y_\eta$ with step $\eta$ from a 
  solution $y$ of the differential equation with $|y(a) - y_\eta(a)|
  \leq \Delta$ satisfies
  \begin{equation*}
    |y_\eta(x) - y(x)| \leq 
    \Delta\,e^{C|x-a|}
    +
    \eta\,(B+AC)\,\frac{e^{C|x-a|-1}}{C},
  \end{equation*}
  for all $x \in [a,b]$.
\end{Euler}
The preceding result also holds if, as in the current chapter, $\eta$ is
not the exact step, but an upper bound for a possibly varying step.

\section{Narrowing  separatrix intervals}
\label{sec:smallSepIntervals}
We first sketch the algorithm for narrowing the separatrix intervals.
To this end we subdivide the box $\I$, and hence the box $N(\I)$,
yielding a nested sequence of boxes $\I = \I_0 \supset \I_1 \supset
\ldots $, with surrounding boxes  $N(\I) = N(\I_0) \supset N(\I_1)
\supset \ldots$, such that
\begin{enumerate}
\item
  $\operatorname{width}(\I_{n+1}) = \tfrac{1}{2}\operatorname{width}(\I_n)$
\item
  the saddle point $p$ is contained in box $\I_n$, for all $n$.
\end{enumerate}
See Figure~\ref{fig:zoomSaddle}.
\begin{figure}[h]
  \centering
  \psfrag{xb0}{$x=b_0$}
  \psfrag{xb1}{$x=b_1$}
  \psfrag{xb2}{$x=b_2$}
  \includegraphics[width=0.48\textwidth]{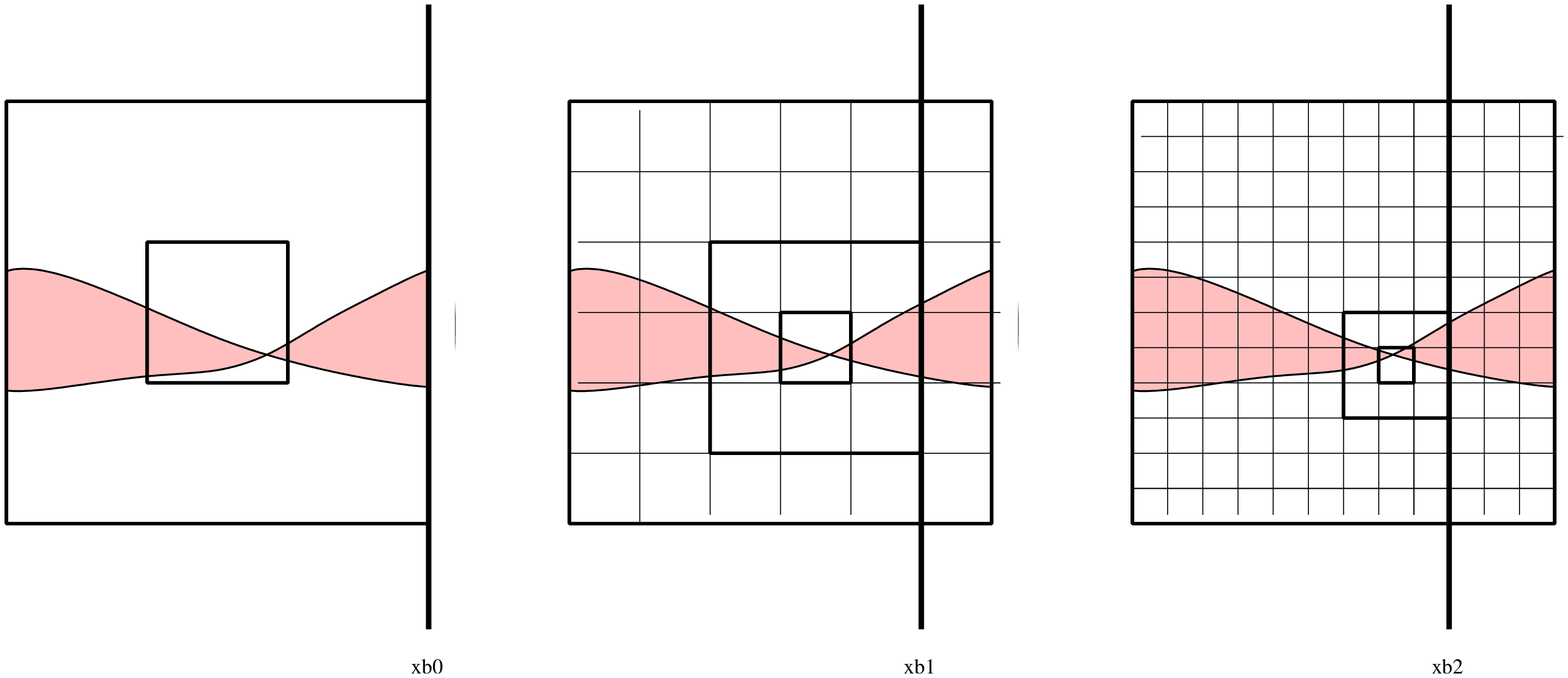}
  \caption{Zooming in on the saddle point by subdivision.}
  \label{fig:zoomSaddle}
\end{figure}
Let $s$ be the $x$-coordinate of the saddle point $p$, and let $b_n$ be the
$x$-coordinate of the rightmost vertical boundary edge of $N(\I_n)$.
Let $w_n$ be the width of $\I_n$, and let $c_n$ be the $x$-coordinate of its center.
Then $b_n = c_n+\tfrac{3}{2}w_n$, and $w_{n+1} = \tfrac{1}{2}w_n$.
Since then $b = b_0 > b_1 > \ldots$, since
\begin{equation*}
  b_{n+1} = c_{n+1}+\tfrac{3}{2}w_{n+1} \leq c_n +\tfrac{1}{4}w_n
  +\tfrac{3}{4}w_n = b_n - \tfrac{1}{2}w_n.
\end{equation*}
Since $|c_n-s| \leq \tfrac{1}{2}w_n$, we get
\begin{equation}
  \label{eq:bn}
  w_n \leq b_n - s \leq 2w_n.
\end{equation}
Consider the forward integral curves of the vector field $\gradient{h}$
through the points of intersection $q^{\pm}_n$ of the line $x = b_n$
and the boundary curves $\Gamma^u_{\pm\beta}$.
See Figure~\ref{fig:narrowingSepInt}.
These curves intersect the rightmost edge of $N(\I)$ in two points
bounding an interval $\J_n$ on this edge.
Arbitrarily good approximations of these integral curves are obtained as follows.
Let $\vartheta_n$ be (an upper bound on) the maximum angle variation of
$\gradient{h}$ over any of the boxes of the $n$-th subdivision of
$N(\I)$.
Since $h$ is $C^2$, the angle variation is a Lipschitz function, so
$\lim_{n\rightarrow\infty}\vartheta_n = 0$. In particular, the rotated
vector fields  $X_{\pm\vartheta_i}$ converge to $\gradient{h}$.
We construct an upper fence with angle $\tfrac{1}{2}\vartheta_n$ for the
upper integral curve, and a lower fence with angle $-\tfrac{1}{2}\vartheta_n$ for the
lower integral curve. See also Section~\ref{sec:ConstructionStrips} for
the construction of a fence. These fences are disjoint, since the angle
variation of $\gradient{h}$ over a grid box is less than $\vartheta_n$.

Since $\lim_{n\rightarrow\infty} b_n = s$, cf (\ref{eq:bn}), the points
$q^{\pm}_n$ converge to the saddle point.
Therefore,  the intervals $\J_0 \supset \J_1 \supset \ldots $, contained
in the intersection of the unstable wedge $C^u_\beta$ and the rightmost
vertical edge of $N(\I)$, converge to the intersection of the unstable
separatrix and the rightmost vertical edge of $N(\I)$.
\begin{figure}[h]
  \centering
  \psfrag{I0}{$J_0$}
  \psfrag{I1}{$J_1$}
  \psfrag{I2}{$J_2$}
  \psfrag{xb0}{$x=b_0$}
  \psfrag{xb1}{$x=b_1$}
  \psfrag{xb2}{$x=b_2$}
  \includegraphics[width=0.44\textwidth]{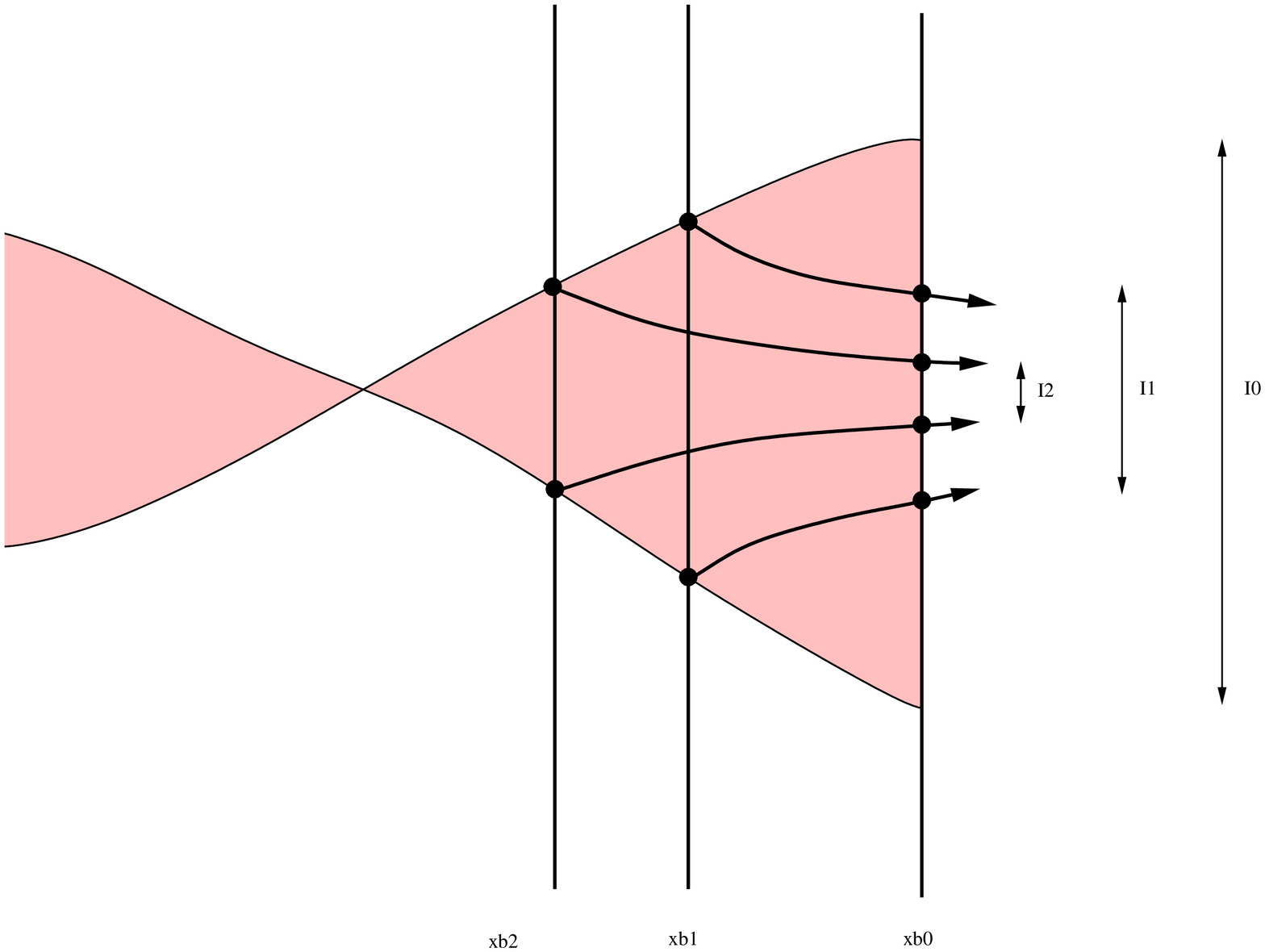}
  \caption{Narrowing separatrix intervals}
  \label{fig:narrowingSepInt}
\end{figure}

\bigskip\noindent
A proof of Lemma~\ref{lemma:smallSepIntervals} can be given along the
lines of~\cite[page 330ff]{CoddLev55} or~\cite[Chapter 3.6]{ch-mbt-82}.
Rather than giving a complete proof we give an example illustrating the
main ideas.
Consider the function $h(x,y) = \tfrac{1}{2}\lambda_u\,x^2 +
\tfrac{1}{2}\lambda_s\,y^2$, with $\lambda_s < 0 < \lambda_u$.
The gradient vector field is given by
$\gradient{h}(x,y) = (\lambda_ux,\lambda_sy)^T$.
Obviously, the origin is a saddle point, with positive eigenvalue
$\lambda_u$ and negative eigenvalue $\lambda_s$, and 
eigenvectors $(1,0)^T$ and $(0,1)^T$, respectively.
The unstable cone of this saddle point is bounded by the curves
$\Gamma^u_{\pm\beta}$, defined implicitly by $  \psi^u_{\pm\beta}(x,y)
=0$, where
\begin{align*}
  \psi^u_{\pm\beta}(x,y) 
  & = 
  \det(V^u,X_\beta(x,y)) \\[1.2ex]
  & =
  \begin{vmatrix}
    1 & (\lambda_u\cos\beta) \, x \mp (\lambda_s\sin\beta) \, y \\[1.2ex]
    0 & (\lambda_u\sin\beta) \, x \pm (\lambda_s\cos\beta) \, y
  \end{vmatrix}\\[1.2ex]
  & = 
  (\lambda_u\sin\beta) \, x \pm (\lambda_s\cos\beta) \, y.
\end{align*}
Therefore, the equation of $\Gamma^u_{\pm\beta}$ is $ y = \pm a x$, 
with $a = - (\tan\beta)\frac{\lambda_u}{\lambda_s} > 0$.
Let the right vertical edge of $\I$ be on the line $x = b$, $b > 0$, and
consider a point $q^{+} = (\xi, a\,\xi)$, with $0 < \xi < b$, 
on the boundary curve $\Gamma^u_\beta$ of the unstable cone.
The integral curve of $\gradient{h}$ through $q^{+}$ satisfies the
differential equation
\begin{equation*}
  \frac{dy}{dx} = \Lambda\,\frac{y}{x},
\end{equation*}
with initial condition $y(\xi) = a\xi$, 
where  $\Lambda = \dfrac{\lambda_s}{\lambda_u} < 0$.
Therefore, 
\begin{equation*}
  y(x) = a\,\xi \bigl(\frac{x}{\xi} \bigr)^{\Lambda}.
\end{equation*}
The integral curve through $q^{+}$ intersects the rightmost edge of $\I$ in
the point $(b,\delta(\xi))$, where
\begin{equation*}
  \delta(\xi) = \frac{ab^{\Lambda}}{\xi^{\Lambda-1}}.
\end{equation*}
Similarly, the integral curve through $q^{-} = (\xi,-a\xi) \in
\Gamma^u_{-\beta}$ intersects the rightmost edge of $\I$ in
the point $(b,-\delta(\xi))$.
Now let $\xi$ range over the sequence $b_0, b_1,\ldots$.
% GV: Replaced b_m in following line with b_n
% Then the interval $\J_n$ has endpoints $(b_m,\pm\delta(b_n))$, so its
Then the interval $\J_n$ has endpoints $(b_n,\pm\delta(b_n))$, so its
width is $2\delta(b_n)$.
In view of (\ref{eq:bn}), with $s=0$, we have
\begin{equation*}
  2 \frac{ab^\Lambda}{w_n^{\Lambda-1}}
  \leq
  \operatorname{width}(\J_n) 
  \leq
  2 \frac{ab^\Lambda}{(2w_n)^{\Lambda-1}}.
\end{equation*}
In other words, with $K = 1 - \Lambda > 1$ and $c = ab^\Lambda \, w_0^K$,
\begin{equation*}
  c \,\bigl(\tfrac{1}{2}\bigr)^{Kn}
  \leq
  \operatorname{width}(\J_n) 
  \leq
  2^K\,c \,\bigl(\tfrac{1}{2}\bigr)^{nK}
\end{equation*}
Hence,
\begin{equation*}
    \operatorname{width}(\J_{n+2}) \leq \dfrac{1}{2^K} \operatorname{width}(\J_n).
\end{equation*}
Since $K>1$, after two subdivision steps the size of the separatrix
interval reduces by more than a factor two. Hence interval arithmetic
provides an arbitrarily good approximation of the intersection of the
unstable separatrix and the boundary of the saddle box.
A similar observation holds for the intersection of the other
separatrices and the boundary of their saddle box.

%%% Local Variables: 
%%% mode: latex
%%% TeX-master: "saddle_analysis"
%%% End: 